%% file: main.tex
\documentclass{article}
\usepackage[utf8]{inputenc} %
\usepackage[T1]{fontenc}    %
\usepackage{graphicx} %
\usepackage{amsmath}
\usepackage{array}
\usepackage{mdframed}

\input{preamble.tex}

\usepackage[ruled,linesnumbered]{algorithm2e}
\SetKwComment{Comment}{/* }{ */}

\usepackage{algorithmic}

\hypersetup{linktocpage = true, linkcolor = purple}

\renewcommand{\Comment}[1]{ \color{teal} #1 \color{black} }

\title{Almost Optimal Multiple Source Shortest Paths and Reachability}
\author{
Barna Saha\thanks{Supported by NSF HDR TRIPODS Phase II grant 2217058 (EnCORE Institute).}  \\
University of California San Diego  \\
\texttt{bsaha@ucsd.edu}
\and
Yinzhan Xu\footnotemark[1]  \\
University of California San Diego  \\
\texttt{xyzhan@ucsd.edu}
\and
Christopher Ye\footnotemark[1]  \\
University of California San Diego  \\
\texttt{czye@ucsd.edu}
}
\date{}

\begin{document}

\maketitle
\author{
}

\thispagestyle{empty}

\begin{abstract} 
    \input{abstract}
\end{abstract}

\thispagestyle{empty}

\newpage
\setcounter{page}{0}
\thispagestyle{empty}
\tableofcontents

\newpage

\pagenumbering{arabic}

\input{intro_new}

\input{overview}

\input{prelims}

\input{decomposition}

\input{mssp}

\input{applications}

\input{directed_reach}

\bibliographystyle{alpha}
\bibliography{biblio}

\appendix
\newpage 
\input{dag_mssp_lb}

\end{document}

%% file: preamble.tex
\usepackage{geometry}
\usepackage{array}
\usepackage{tabularx}
\usepackage{graphicx}
\usepackage{amsmath,amssymb,amsthm,mathtools}
\usepackage{paralist}
\usepackage{bm}
\usepackage{bbm}
\usepackage{xspace}
\usepackage{url}
\usepackage{fullpage, prettyref}
\usepackage{boxedminipage}
\usepackage{wrapfig}
\usepackage{ifthen}
\usepackage{color}
\usepackage{xcolor}
\usepackage{framed}
\usepackage[pagebackref,colorlinks=true,urlcolor=blue,linkcolor=blue,citecolor=blue,pdfstartview=FitH]{hyperref}
\usepackage[nameinlink]{cleveref}
\usepackage{esvect}
\usepackage{mdframed}
\usepackage{tikz}

\usepackage{thmtools}
\usepackage{thm-restate}

\newtheorem{theorem}{Theorem}[section]


\newtheorem{lemma}[theorem]{Lemma}
\newtheorem{claim}[theorem]{Claim}

\newtheorem{definition}[theorem]{Definition}
\newtheorem{proposition}[theorem]{Proposition}

\newtheorem{question}{Question}

\newcommand{\ignore}[1]{}

\newcommand{\Z}{{\mathbb Z}}
\newcommand{\eps}{\varepsilon}

\newcommand{\poly}{\mathrm{poly}}
\newcommand{\polylog}{\mathrm{polylog}}

\newcommand{\what}{\widehat}

\newcommand{\calN}{{\cal N}}

\newcommand{\calP}{{\cal P}}

\newcommand{\calC}{{\cal C}}

\newcommand{\calA}{{\cal A}}

\newcommand{\mbone}{\mathbbm{1}}

\newcommand{\set}[1]{\left\{ #1 \right\}}

\newcommand{\given}{\textrm{ s.t. }}
\newcommand{\andT}{\textrm{ and }}
\newcommand{\otherwise}{\textrm{ o/w }}
\renewcommand{\leq}{\leqslant}
\renewcommand{\geq}{\geqslant}

\newcommand{\ind}[1]{\mbone\left[ #1 \right]}

\newcommand{\Sec}[1]{\hyperref[sec:#1]{\Cref*{sec:#1}}} %
\newcommand{\Eqn}[1]{\hyperref[eq:#1]{(\ref*{eq:#1})}} %
\newcommand{\Fig}[1]{\hyperref[fig:#1]{Fig.\,\ref*{fig:#1}}} %
\newcommand{\Tab}[1]{\hyperref[tab:#1]{Tab.\,\ref*{tab:#1}}} %
\newcommand{\Thm}[1]{\hyperref[thm:#1]{Theorem\,\ref*{thm:#1}}} %
\newcommand{\Fact}[1]{\hyperref[fact:#1]{Fact\,\ref*{fact:#1}}} %
\newcommand{\Lem}[1]{\hyperref[lem:#1]{Lemma\,\ref*{lem:#1}}} %
\newcommand{\Prop}[1]{\hyperref[prop:#1]{Prop.~\ref*{prop:#1}}} %
\newcommand{\Cor}[1]{\hyperref[cor:#1]{Corollary~\ref*{cor:#1}}} %
\newcommand{\Conj}[1]{\hyperref[conj:#1]{Conjecture~\ref*{conj:#1}}} %
\newcommand{\Def}[1]{\hyperref[def:#1]{Definition~\ref*{def:#1}}} %
\newcommand{\Alg}[1]{\hyperref[alg:#1]{Alg.~\ref*{alg:#1}}} %
\newcommand{\Obs}[1]{\hyperref[obs:#1]{Obs.~\ref*{obs:#1}}} %
\newcommand{\Ex}[1]{\hyperref[ex:#1]{Ex.~\ref*{ex:#1}}} %
\newcommand{\Clm}[1]{\hyperref[clm:#1]{Claim~\ref*{clm:#1}}} %
\newcommand{\Step}[1]{\hyperref[step:#1]{Step~\ref*{step:#1}}} %

\newcommand{\tO}[1]{\widetilde{O}\left(#1\right)}
\newcommand{\tOm}[1]{\widetilde{\Omega}\left(#1\right)}
\newcommand{\hO}[0]{\widehat{O}}

\newcommand{\bigO}[1]{O\left(#1\right)}
\newcommand{\bigtO}[1]{\tilde{O}\left(#1\right)}
\newcommand{\bigOm}[1]{\Omega\left(#1\right)}

\newcommand{\TMUL}{\mathsf{T}_{\textup{MUL}}}
\newcommand{\MPP}{\textup{MPP}}
\newcommand{\TMPP}{\mathsf{T}_{\textup{MPP}}}
\newcommand{\BMM}{\textup{BMM}}
\newcommand{\TBMM}{\mathsf{T}_{\textup{BMM}}}

\crefname{algocf}{alg.}{algs.}
\Crefname{algocf}{Alg.}{Algs.}

\newcommand{\ecc}{\mathrm{ecc}}
\newcommand{\diam}{\mathrm{diam}}
\newcommand{\dist}{\mathrm{dist}}
\newcommand{\reach}{\mathrm{reach}}
\newcommand{\pred}{\mathrm{pred}}
\newcommand{\wt}{\mathrm{wt}}
\newcommand{\inT}{\mathrm{in}}

\newcommand{\LP}{\mathrm{LP}}
\newcommand{\TC}{\mathrm{TC}}
\newcommand{\rev}{\mathrm{rev}}

\newcommand{\LDND}{{\rm LDSND}}
\newcommand{\bfs}{{\rm BFS}}
\newcommand{\sssp}{{\rm SSSP}}
\newcommand{\npsp}{{\rm NPSP}}
\newcommand{\mssp}{{\rm MSSP}}
\newcommand{\apsp}{{\rm APSP}}
\newcommand{\SCC}{{\rm SCC}}

\newcommand{\minget}{\overset{\min}{\gets}}
\newcommand{\addget}{\overset{+}{\gets}}
\newcommand{\subget}{\overset{-}{\gets}}
\newcommand{\cupget}{\overset{\cup}{\gets}}
\newcommand{\orget}{\overset{\vee}{\gets}}

\newcommand{\updateScores}{\textsf{updateScores}}
\newcommand{\extendPaths}{\textsf{extendPaths}}
\newcommand{\increaseLevel}{\textsf{increaseLevel}}

\newcommand{\score}{\mathrm{score}}
\newcommand{\level}{\mathrm{level}}

\newcommand{\hdist}{\what{\mathrm{d}}}

\newcommand{\complete}{\textsf{complete}}

%% file: abstract.tex
Given a graph, computing distances and reachabilities from a small set of vertices to the whole graph is an important primitive both in theory and in practice. In this paper, we study this problem in various important settings.

In undirected unweighted graphs, while computing single-source shortest path (SSSP) requires $O(n^2)$ time in dense graphs, all-pairs shortest paths (APSP) can be computed in $\hat{O}(n^\omega) = O(n^{2.372})$ time [Seidel '95] providing significant savings over running $n$ SSSP instances separately. However, if one needs to compute multiple-source shortest paths (MSSP) from a set of $n^\sigma$ vertices, the previously best known running time was $\hat{O}(\min\{n^\omega, n^{2 + \sigma}\})$: either by running APSP or running SSSP from each of the $n^\sigma$ sources. On the other hand,  a standard reduction shows that computing MSSP from $n^\sigma$ sources is at least as hard as computing the Boolean matrix product between an $n^\sigma \times n$ matrix and an $n \times n$ matrix. This lower bound has a gap from the state-of-the-art algorithm and therefore, MSSP can potentially  be solved much faster. The first main result of our paper is an algorithm for MSSP on undirected unweighted graphs that matches this lower bound. In particular, using fast rectangular matrix multiplication, our algorithm runs in $\hat{O}(n^{\omega(\sigma, 1, 1)})$ time, which gives a smooth interpolation between the SSSP algorithm and the APSP algorithm. We also extend our algorithm to the case where the graph has small integer weights. The main technical tool behind our first main result is a novel graph decomposition, which may be of independent interest. As applications, we improve the state-of-the-art running times for hop-sets construction and computing $+4$-distance emulators.

Next, we study the multiple-source reachability problem, where we need to determine whether a given set of $n^\sigma$ vertices can reach each of the vertices in a given directed graph. Multiple-source reachability can also be solved in $\hat{O}(\min\{n^\omega, n^{2 + \sigma}\})$ time, and similarly is as hard as computing Boolean matrix product between an $n^\sigma \times n$ matrix and an $n \times n$ matrix. [Elkin and Trehan '24, '25] gave the first nontrivial algorithm for multiple-source reachability that runs in $\hat{O}(n^{1+2\omega(\sigma, 1, 1) / 3})$ time. This only beats the trivial $\hat{O}(\min\{n^\omega, n^{2 + \sigma}\})$ time for a restricted range of $\sigma$. We significantly improve their results, by giving an algorithm that runs in $\hat{O}(n^{\omega(\sigma, 1, 1)})$ time, again matching the running time for Boolean matrix product. 
Our algorithm for multiple-source reachability can be generalized to MSSP on DAGs with small integer weights. 
As an application, we provide an $O(n^{2.084})$ time algorithm for computing an $\widetilde{O}(n)$-size shortcut set that reduces diameter to $O(n^{1/3})$, improving the previous $\widetilde{O}(n^{7/3})$ time algorithm [Kogan and Parter '22, '23]. 

%% file: intro_new.tex
\section{Introduction}

Computing shortest paths is among the most fundamental problems in computer science. Many of its variants  have been studied in the literature, and the All-Pairs Shortest Paths (\apsp) problem is one of the most fundamental variants and has been extensively studied.

In \apsp, one needs to compute the distance between every pair of nodes in a given $n$-vertex graph.\footnote{Throughout this paper, we consider dense graphs with $m = O(n^2)$.} 
One could run Dijkstra's algorithm $n$ times to obtain an $O(n^{3})$ running time. 
An alternative classical algorithm, the Floyd–Warshall algorithm, also achieves an $O(n^3)$ running time. 
A long series of works 
 (e.g. \cite{fredman1976new,Dob1990,Takaoka92,Takaoka98,Han04,pettie2004apsp,ZwickAPSP04,Chan05,Han06,DBLP:journals/siamcomp/Chan10,HanT12}) 
improved poly-logarithmic factors over the classic $O(n^3)$ running time, culminating in the current fastest algorithm by Williams \cite{Williams18} with $n^3 / 2^{\Theta(\sqrt{\log n})}$ running time for general weighted graphs, a super poly-logarithmic improvement over the classic $O(n^3)$ running time. 
Due to the lack of polynomial improvements, the \apsp{} hypothesis, stating that there is no $O(n^{3-\eps})$ time algorithm for any $\eps>0$ that computes \apsp{} for $n$-node graphs, is among the most popular hypotheses in Fine-Grained Complexity.

Faster algorithms are known in unweighted graphs. 
If the graph is undirected, Seidel \cite{seidel1995apsp} gave an algorithm computing \apsp{} in $\hO(n^{\omega})$ time.\footnote{$\omega \le 2.372$ \cite{alman2025fmm} is the fast matrix multiplication exponent, or the smallest constant $c$ such that an $n \times n \times n$ matrix product can be computed in $n^{c+o(1)}$ time. $\hO$ hides sub-polynomial factors.}
More generally, on undirected graphs with positive integer weights bounded by $B$, \apsp{} can be computed in $\hO(Bn^{\omega})$ time~\cite{shoshan_zwick, alon1997exponent}. 
For directed unweighted graphs, the current fastest algorithm was designed by Zwick \cite{zwick2002apsp}, which runs in $O(n^{2.528})$ time.

A simpler version of distance problems is reachability problems, where one needs to decide whether a vertex can reach another vertex. 
All-Pairs Reachability, also popularly known as transitive closure, is another fundamental problem in graph algorithms. 
It is known to be equivalent to the Boolean Matrix Multiplication problem, and can be solved in $\hO(n^\omega)$ time \cite{fischer1971boolean, munro1971efficient}. 

Computing all the pairwise distances (resp. reachability) gives full distance (resp. reachability) information of the graph, but it also comes with a cost: it is computationally expensive. In contrast, computing the distance between a fixed pair of vertices only takes near-linear time in all the variants mentioned above, though it gives much less information than \apsp{}. A natural question is to study the intermediate cases: 

\begin{question}
    \label{question:question1}
    Can we compute some distances faster than all distances?
\end{question}

Graphs arising in theory and practice applications often contain a small set of hot spots, where distances between these hot spots and other vertices are more likely to be queried. 
For instance, in the network connecting all locations in a city, distances from places like grocery stores, shopping malls and parks are more likely to be queried than distances between two residential buildings. In theory, many applications using distance computations only need distances from a fixed small set of vertices to other vertices (e.g. approximate APSP \cite{DBLP:conf/focs/DorHZ96, DBLP:conf/icalp/DengKRWZ22, DBLP:conf/soda/SahaY24, DBLP:conf/soda/DoryFKNWV24, DBLP:conf/soda/0001KSW26}, approximate diameter \cite{roditty2013diameter, chechik2014better, cairo2016new} and shortcut sets and hopsets \cite{DBLP:journals/siamcomp/UllmanY91}). 

Motivated by these applications, we study the multiple source version of distance and reachability problems. More formally, we study: 

\begin{definition}[Multiple Source Shortest Paths / Reachability]
    \label{def:mssp}
    Given graph $G = (V, E)$ and a set of sources $S_{\inT} \subseteq V$, compute $\dist(s, v)$ for all $s \in S_{\inT}$ and $v \in V$. In the reachability version of the problem, we only need to determine if the distance is $\infty$ or not. 
\end{definition}

\subsection{Shortest Paths in Undirected Unweighted Graphs}

Our first result is for Multiple-Source Shortest Paths in undirected unweighted graphs.
There are two straightforward algorithms to solve \mssp: we can either solve \sssp{} for every $s \in S_{\inT}$, or solve \apsp{}.
For undirected graphs, there were no other known algorithms besides these two approaches, to the best of our knowledge.  
This leads to the following question that motivated our work, instantiating a specific instance of \Cref{question:question1}.

\begin{center}
    {\it Is there an \mssp{} algorithm faster than $\hO(\min(|S_{\inT}|m, n^{\omega}))$?}
\end{center}

We answer this question in the affirmative, giving the first algorithm that beats this running time. In fact, our algorithm is essentially optimal, by tying the running time of our algorithm to Boolean matrix multiplication.

\begin{restatable}{theorem}{MSSP}
    \label{cor:mssp}
    \mssp{} on undirected unweighted graphs is equivalent to $\BMM(|S_{\inT}|, n, n)$, up to $n^{o(1)}$ factors, for randomized algorithms.\footnote{\label{fn:mssp}Our equivalence holds under the mild assumption that $\TBMM(2m, 2n, 2n) \geq 4 \TBMM(m, n, n)$ for all $m, n \geq 1$. 
    Note that this assumption essentially requires that $\BMM$ scales at least linearly with input size.
    We remark that our improved running time holds regardless of the assumption.
    See proof of \Cref{clm:mm-quadratic-growth} for details.}
    In particular, there is a randomized algorithm computing \mssp{} on undirected, unweighted graphs in $n^{o(1)}\TBMM\left(|S_{\inT}|, n, n \right)$ time.
\end{restatable}

Here, $\BMM(a, b, c)$ denotes Boolean matrix product instances between $a \times b$ matrices and $b \times c$ matrices, and $\TBMM(a, b, c)$ denotes the running time of such instances. In particular, $\TBMM(a, b, c)$ can be upper bounded by the running time for computing products of general $a \times b$ and $b \times c$ matrices, utilizing fast matrix multiplication algorithms (see \Cref{fig:msr-comparison} for an illustration of the quantitative improvements we achieve). 
The reduction from Boolean matrix multiplication to \mssp{} in the above theorem can be achieved with standard approaches.\footnote{Let $A$ be an $s \times n$ matrix and $B$ an $n \times n$ matrix.
We build a three layered graph with vertex layers $V_1, V_2, V_3$, where $|V_1| = s$, $|V_2| = |V_3| = n$, encode $A$ into $V_1 \times V_2$ (i.e., there is an edge if the corresponding entry of $A$ is $1$, and no edge otherwise) and $B$ into $V_2 \times V_3$, and note that for $i \in V_1, j \in V_3$, $\dist(i, j) = 2$ if and only if $(AB)[i, j] = 1$.}
We also extend our algorithm to handle graphs with small integer weights.

\begin{restatable}{theorem}{MSSPBounded}
    \label{thm:bounded-weight-mssp}
    \mssp{} on undirected graphs with integer edge weights in $[B]$ is equivalent to $\MPP(|S_{\inT}|, n, n \mid B)$, up to $n^{o(1)} \log B$ factors, for randomized algorithms.\footnote{Like \Cref{cor:mssp}, the equivalence holds under the mild assumption that $\TMPP(2m, 2n, 2n \mid B) \geq 4 \TMPP(m, n, n \mid B)$ for all $m, n \geq 1$.
    As before, our improved running time holds regardless of the assumption.}
    In particular, if $B \leq \poly(n)$, there is a randomized algorithm computing \mssp{} on undirected graphs with integer edge weights in $[B]$ in 
    $n^{o(1)} \TMPP \left(|S_{\inT}|, n, n \mid B \right)$ time. 
\end{restatable}

Here, $\MPP(|S_{\inT}|, n, n \mid B)$ denotes Min-Plus product between an $|S_{\inT}| \times n$ matrix and an $n \times n$ matrix, whose entries are from $[B] \cup \{\infty\}$, 
and $\TMPP(|S_{\inT}|, n, n \mid B)$ denotes the running time of such instances. 
It is known that $\TMPP(a, b, c \mid B) = \tO{B \TMUL(a, b, c)}$ \cite{alon1997exponent}.

Our algorithms give a smooth interpolation between the \sssp{} algorithm and the \apsp{} algorithms \cite{seidel1995apsp} and \cite{shoshan_zwick}. 
When $s = n$, our algorithm serves as an alternative algorithm for \apsp{} in undirected graphs with unit or small integer weights (See \cite{DBLP:conf/icalp/ChanWX21} for another alternative algorithm).

The key technical ingredient of our algorithms is a novel low diameter decomposition, which may be of independent interest.
We define the notion of a low diameter small neighborhood decomposition.
We say $V_1 \sqcup \dots \sqcup V_{k}$ is a $(d, \phi)$-low diameter small neighborhood decomposition if (1) $\diam(G[V_i]) \leq d$ for all $i$, and (2) $\sum_{i} |N(V_i)| \leq \phi n$.
While the first condition is reminiscent of standard low diameter decompositions, the second condition is different: instead of requiring that the number of edges between clusters is small (see e.g. \cite{DBLP:conf/focs/BernsteinNW22}), our decomposition requires that the neighborhoods of the clusters are small.
In particular, while there may be many edges incident to any given cluster, we ensure that they do not expand and hit too many distinct vertices.
Note that our bound on the neighborhood size is global rather than local, there may be clusters where $|N(V_i)| > \phi |V_i|$, but the sum of the neighborhoods is at most $\phi n$.
We show how to obtain an $(n^{o(1)}, n^{o(1)})$-low diameter small neighborhood decomposition in $n^{2 + o(1)}$ time.

\begin{restatable}{theorem}{Decomposition}
    \label{cor:boundary-decomposition}
    There is a randomized $n^{2 + o(1)}$ time algorithm that computes a $\left(2^{\bigO{\sqrt{\log n}}}, 2^{\bigO{\sqrt{\log n}}}\right)$-low diameter small neighborhood decomposition.
\end{restatable}

We also give some natural applications of our \mssp{} algorithm.
We obtain improvements over the state-of-the-art algorithms with simple reductions, establishing the practical impact of the \mssp{} problem (see \Cref{sec:applications} for formal statements).

\paragraph{Hop-sets.}

A (weighted) set of edges $F$ is a $\beta$-hop-set for graph $G = (V, E)$ if every pair of vertices $(u, v)$ in $H = (V, E \cup F)$ has the same distance as in $G$ and there exists a shortest path of at most $\beta$ hops.
Research in hop-sets (and the related concept of shortcut sets) aims to minimize both the size of the hop-set (ideally linear) and the hop diameter $\beta$.
A folklore sampling algorithm produces $\beta$-hop-sets of size $n^2/\beta^{2}$: sample a set $S$ of $n/\beta$ vertices, and add all edges between $(u, v) \in S \times S$ with weight $\dist(u, v)$.
In particular, we may obtain an $\tO{n^{1/2}}$-hop-set of size $\tO{n}$.
Bodwin and Hoppenworth showed that the folklore sampling is optimal up to polylogarithmic factors: any exact hop-set of size $\tO{n}$ must have $\beta = \tOm{n^{1/2}}$ \cite{bodwin2023folklore}.

However, the sampling algorithm does not compute a hop-set efficiently, since previous algorithms for computing all distances between $S \times S$ for $S$ of size $\sqrt{n}$ require time $\hO(\min(n^{2.5}, n^{\omega})) = O(n^{2.372})$.
Our improved \mssp~algorithm allows us to compute all $S \times S$ distances in time $n^{\omega(0.5, 1, 1)} = O(n^{2.043})$ \cite{alman2025fmm}.\footnote{$\omega(a, b, c)$ is the rectangular matrix multiplication exponent, or the smallest constant $d$ such that an $n^{a} \times n^{b} \times n^{c}$ matrix product can be computed in $n^{d+o(1)}$ time.}

\paragraph{Distance Emulators.}

Emulators are sparse graphs that faithfully preserve the distances of a graph.
Formally, given a graph $G=(V,E)$, a (possibly weighted) set of edges $F$ is a $+k$-emulator if $\dist_{G}(u, v) \leq \dist_{(V, F)}(u, v) \leq \dist_{G}(u, v) + k$ for all pairs of vertices $u, v$.
Emulators (and related constructions such as graph spanners) have been the subject of intense study since their introduction \cite{peleg1989graph}, and have been useful objects with many applications, including compact routing schemes \cite{peleg1989trade, cowen2000compact, cowen2001compact, thorup2001compact, roditty2008roundtrip}, distance oracles \cite{thorup2005oracle, baswana2006faster}, and broadcasting \cite{farley2004spanners}.

Several constructions of additive emulators are known: $+2$-emulators of size $\tO{n^{3/2}}$ \cite{aingworth1999fast}, $+4$-emulators of size $\tO{n^{4/3}}$ \cite{dor2000all}, and $+6$-emulators of size $\tO{n^{4/3}}$ \cite{baswana2005spanner}.
We remark that both the $+2$ and $+6$ constructions are spanners ($F$ is a subset of $E$).
However, the sparsest $+4$-spanner known has size $\tO{n^{7/5}}$ \cite{chechik2014spanner}.
Furthermore, it is known that no emulator (or any data structure) of size $n^{4/3 - \Omega(1)}$ can achieve  error $k = n^{o(1)}$ \cite{abboud20174}.
Thus, significant effort has been devoted to characterizing the additive error of emulators of linear size \cite{thorup2006spanners, bodwin2015very, huang2021lower, elkin2023improved, kogan2023new, hoppenworth2024simple, harbuzova2024improved}.

Both the $+2$ and $+6$-emulators can be constructed in near-linear time $\tO{n^2}$ \cite{dor2000all, woodruff2010additive}.
However, the best known algorithm for constructing the $+4$-emulator requires time $\tO{n^{7/3}}$ \cite{dor2000all}.
With our \mssp~algorithm, we obtain a construction in time $\hO(n^{\omega(2/3, 1, 1)}) = O(n^{2.1321})$ \cite{alman2025fmm}.

\subsection{Reachability in Directed Graphs}

In the Multiple-Source Reachability problem, the goal is to determine whether $s$ can reach $v$ for every pair $(s, v) \in S_{\inT} \times V$. 
As all-pairs reachability can be solved in $\hO(n^\omega)$ time \cite{fischer1971boolean, munro1971efficient}, and single-source reachability can be solved in $O(m)$ time, Multiple-Source Reachability can be solved in $\hO(\min(|S_{\inT}| m, n^\omega))$ time, similar to \mssp{} on undirected unweighted graphs. 
Standard reduction from $\BMM(S_{\inT}, n, n)$ also works for Multiple-Source Reachability, establishing an $\Omega(\TBMM(S_{\inT}, n, n))$ lower bound. 

Recently, Elkin and Trehan \cite{elkin2024reach, elkin2025reach} gave an alternative algorithm for Multiple-Source Reachability. Their algorithm runs in $\hO\left(n \TMUL(S_{\inT}, n, n)^{2/3}\right)$, which beats $\hO(\min(S_{\inT} m, n^\omega))$ for certain values of $|S_{\inT}|$, but unfortunately never meets the $\Omega(\TBMM(S_{\inT}, n, n))$ lower bound (see \Cref{fig:msr-comparison}).

In our second main result, we design an algorithm for Multiple-Source Reachability that matches the $\Omega(\TBMM(S_{\inT}, n, n))$ lower bound. 

\begin{restatable}{theorem}{MultiSourceReach}
    \label{thm:msreach}
    Multiple Source Reachability is equivalent to $\BMM(|S_{\inT}|, n, n)$, up to $\polylog(n)$ factors.
    In particular, there is a deterministic algorithm computing Multiple-Source Reachability in $\tO{\TBMM\left(|S_{\inT}|, n, n \right)}$ time.
\end{restatable}

The following figure illustrates a comparison of our algorithm with the baseline algorithm and the algorithm of Elkin and Trehan \cite{elkin2024reach, elkin2025reach}.

\begin{figure}[h!]
    \centering
    \includegraphics[width=0.5\linewidth]{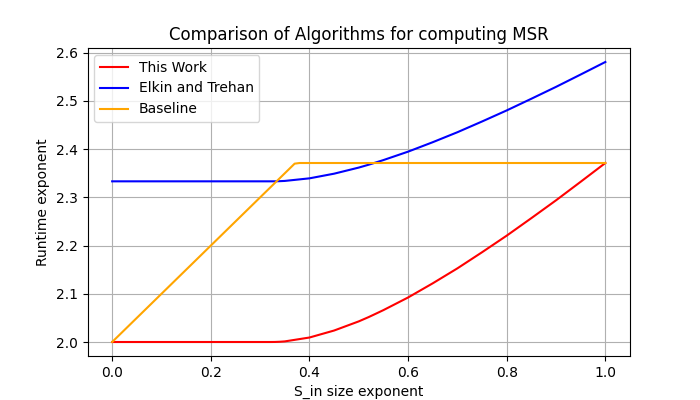}
    \caption{Comparison of running times of algorithms for Multiple-Source Reachability on dense graphs when $|S_{\inT}| = n^{\sigma}$.
    The baseline algorithm (yellow line) runs in time $\hO(\min(n^{2 + \sigma}, n^{\omega}))$.
    The algorithm of Elkin and Trehan (blue line) runs in time $\hO(n^{1 + 2\omega(\sigma, 1, 1)/3})$.
    Note that the figure also illustrates our improvement over the baseline algorithm for undirected \mssp{}, since the baseline algorithms are identical for undirected \mssp{} and Multiple-Source Reachability.}
    \label{fig:msr-comparison}
\end{figure}

\subsection{Shortest Paths in Directed Acyclic Unweighted Graphs}

Finally, we study \mssp{} on directed unweighted graphs, whose running time landscapes seem to be more complicated. Consider \apsp{} as an example. 
\apsp{} for undirected unweighted graphs runs in $\hO(n^\omega)$ time~\cite{seidel1995apsp}, while \apsp{} for directed unweighted graphs runs in $\hO(n^{2+\mu})$ time~\cite{zwick2002apsp}, where $\mu \le 0.528$~\cite{alman2025fmm} is the solution to $\omega(1, \mu, 1) = 2 + \mu$. 
As we can see, even representing the running time requires more effort in the directed case. 

One can rephrase the running time of Zwick's algorithm as 
\[
\tO{\max_{1 \le t \le n} \TMPP(n, t, n \mid n / t)}, 
\]
as observed in \cite{DBLP:conf/icalp/ChanWX21}. 
In fact, \cite{DBLP:conf/icalp/ChanWX21} showed that the running time for \apsp{} in directed unweighted graphs is exactly
\[
\widetilde{\Theta}\left(\max_{1 \le t \le n} \TMPP(n, n/t, n \mid t)\right). 
\]
Grandoni and Vassilevska Williams  \cite{grandoni2019faster} gave an algorithm that can solve \mssp{} on directed unweighted graphs\footnote{Their algorithm is actually more general and can be used to solve multiple-source multiple-target shortest paths.}, and the running time of their algorithm can be phrased as (let $s = |S_{\inT}|$)
\begin{equation}
\label{eq:GV-time}
\tO{\max_{1 \le t \le n} \TMPP(s + t, t, n \mid n / t)}, 
\end{equation}
Based on the approach in \cite{DBLP:conf/icalp/ChanWX21}, we can show the following lower bound for \mssp{} on directed unweighted graphs:
\begin{equation}
\label{eq:directed-lower-bound}
\Omega\left(\max_{1 \le t \le n} \TMPP(s, t, n \mid n / \min(s, t), n / t)\right), 
\end{equation}
where $\TMPP(a, b, c \mid B_1, B_2)$ denotes the running time for Min-Plus product instances between an $a \times b$ matrix with entries from $[B_1] \cup \{\infty\}$ and a $b \times c$ matrix with entries from $[B_2] \cup \{\infty\}$. 

Grandoni and Vassilevska Williams's algorithm \cite{grandoni2019faster} is actually optimal for sufficiently large $s$. To see this, let $s$ be large enough so that the maximizer $t^*$ for \Cref{eq:GV-time} is less than or equal to $s$. In this case, their running time simplifies to $\tO{\TMPP\left(s, t^*, n \mid n / t^*\right)}$, which matches $\TMPP(s, t, n \mid n / \min(s, t), n / t)$ by setting $t = t^*$. On the other hand, their running time is not optimal for small values of $s$. 
By setting $t = n$ in \Cref{eq:GV-time}, their running time is not better than $\tO{\TMPP\left(n, n, n \mid 1\right)} = \widetilde{\Theta}\left(\TBMM\left(n, n, n\right)\right)$. Hence, using fast matrix multiplication for Boolean matrix product (this is the currently best known algorithm for Boolean matrix product), their running time always has an $n^\omega$ term, even for small values of $s$. 

We improve \cite{grandoni2019faster}'s algorithm for small values of $s$, in the important special case of directed acyclic graphs:

\begin{restatable}{theorem}{DAGMSSP}
    \label{thm:dag-mssp}
    There is a deterministic algorithm computing \mssp{} on unweighted directed acyclic graphs in $\tO{\TMPP(|S_{\inT}|, n, n \mid n, 1)}$ time.
\end{restatable}

Surprisingly, even though the running time expression for \Cref{thm:dag-mssp} is quite different from \Cref{eq:directed-lower-bound}, the \textit{numerical} values of the exponents in $n$ appear to match for small values of $s$, utilizing the currently best known algorithm for Min-Plus product between two matrices that can have potentially different entry bounds. 
In fact, if we take the minimum of \cite{grandoni2019faster}'s algorithm and \Cref{thm:dag-mssp}, the upper and lower bounds appear to be numerically matching for all values of $s$! Of course, this does not establish a formal lower bound, as there might be better algorithms for such Min-Plus product instances. 
See \Cref{app:dir-mssp-lb} for details. 

\paragraph{Shortcut Sets}

A set of edges $F$ is a $D$-shortcut-set for graph $G = (V, E)$ if $H = (V, E \cup F)$ has the same transitive closure as $G$ and $\diam(G \cup F) \leq D$.
For a directed graph, $\diam(G) = \max_{(u, v) \in \TC(G)} \dist(u, v)$ is the maximum distance between any two vertices for $(u, v) \in \TC(G)$, where $\TC(G)$ is the transitive closure of $G$.
Shortcut sets have applications in reachability and shortest paths computation in many settings \cite{DBLP:journals/jal/KleinS97, 
DBLP:conf/stoc/HenzingerKN14,
DBLP:conf/icalp/HenzingerKN15,
DBLP:conf/focs/ForsterN18, 
DBLP:conf/focs/LiuJS19,
DBLP:journals/siamcomp/Fineman20, 
DBLP:conf/soda/GutenbergW20a, 
DBLP:conf/focs/BernsteinGW20}. 
Similar to exact hop-sets, a folklore sampling algorithm yields a shortcut set of diameter $O(\sqrt{n})$ and size $\tO{n}$.
A breakthrough result of Kogan and Parter \cite{koganparter2022} showed that in fact there exist shortcut sets of size $\tO{n}$ and diameter $O(n^{1/3})$, while no linear size shortcut set can achieve diameter $o(n^{1/4})$ \cite{bodwin2023folklore, vxx2024shortcut}. 
As with hop-sets, known algorithms for computing shortcut sets are not optimal, with the current best algorithm computing diameter $D$-shortcut-sets of size $\tO{n^2/D^3 + n}$ in time $\tO{mn/D^2} = \tO{n^3 / D^2}$ whenever $D = O(\sqrt{n})$~\cite{DBLP:conf/soda/KoganP23} (\cite{kogan2022beating}'s running time is the same for dense graphs).
As an application of \Cref{thm:dag-mssp}, we compute shortcut-sets in dense graphs of diameter $D$ and size $\tO{n^2/D^3 + n}$ in time $\tO{\TMPP(n/D^2, n, n \mid D, 1)}$ (see \Cref{thm:shortcut-set}).
Our algorithm is faster than \cite{DBLP:conf/soda/KoganP23} for all values of $D = O(\sqrt{n})$, as illustrated in \Cref{fig:shortcut-set}. For the important case $D = n^{1/3}$, our running time is $O(n^{2.064})$. 

\begin{figure}[h!]
    \centering
    \includegraphics[width=0.5\linewidth]{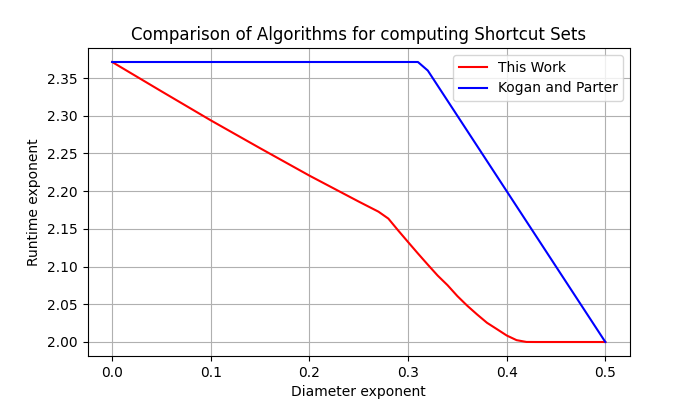}
    \caption{Comparison of running times of algorithms for computing shortcut sets in dense graphs with $m = \Theta(n^2)$ with diameter $D = n^{x}$ for $x \leq 0.5$. 
    The algorithm of Kogan and Parter \cite{koganparter2022, DBLP:conf/soda/KoganP23} is the minimum of $\tO{n^{\omega}}$ and $\tO{n^3/D^2} = \tO{n^{3 - 2x}}$.
    See \Cref{app:dir-mssp-lb} for a discussion on known upper bounds on $\TMPP(n_1, n_2, n_3 \mid B_1, B_2)$.}
    \label{fig:shortcut-set}
\end{figure}

\subsection{Related Works}

\Cref{question:question1} is a long-term research paradigm that has been approached from many angles:
\begin{enumerate}
    \item Single-Source Shortest Paths (\sssp{}): The most classical and successful example \cite{dijkstra1959, fredman1987fibonacci}, \sssp{} computes a set of $n$ distances from $s$.

    \item Extremal distances: The diameter, radius, and eccentricities are well-studied problems in distance computation, resulting in both approximation algorithms \cite{roditty2013diameter, chechik2014better, backurs2018towards, cairo2016new} and hardness results \cite{roditty2013diameter, bonnet20224, dalirrooyfard2025hardness}.

    \item $n$-Pairs Shortest Paths (\npsp): Given $n$ pre-specified pairs of vertices $(s_i, t_i)$, compute $\dist(s_i, t_i)$ for all $i$. 
    In contrast to \sssp{}, the \npsp{} problem computes $n$ arbitrary pre-specified distances.
    The study of \npsp{} was initiated by \cite{aingworth1999fast}, with several recent works reigniting interest in the question
    \cite{dalirrooyfard2022npsp, chechik2025npsp}.
    
    \item Distance Oracles: Preprocess $G$ so that given any query $(u, v)$, return $\dist(u, v)$ efficiently.
    Distance oracles capture the setting where we do \emph{not} know what distances are interesting in advance.
    A long line of work has established a variety of approximation algorithms \cite{Thorup2004planarOracle, thorup2005oracle, patrascu2010oracle, Chechik2014oracle, CA2017oracle, Char2019oracle, chechik2022oracle} and hardness results \cite{abboud2022hardness, jin2023removing, abboud2023stronger}.
\end{enumerate}

\subsection{Open Questions}

We conclude the introduction with some open questions.

\begin{enumerate}
    \item We establish an equivalence between \BMM{} and \mssp{} on unweighted, undirected graphs.
    Can a similar equivalence (with some other matrix product) be established for \mssp{} on directed graphs?
    Concretely, can \Cref{thm:dag-mssp} be generalized to compute shortest paths on directed graphs (not necessarily DAGs)?
    
    \item In our work, we have shown that a $(d, \phi)$-low diameter small neighborhood decomposition can be computed efficiently, where $d, \phi = n^{o(1)}$.
    In principle, it may be possible to obtain a decomposition where $|N(V_i)| = O(|V_i|)$ and $\diam(V_i) = O(\log n)$ for all $V_i$.
    Can we obtain such a decomposition, and do so efficiently? Even the weaker setting where  $d, \phi = \polylog(n)$ would be interesting. 
    
    \item Our decomposition algorithm runs in $n^{2 + o(1)}$ time, which is efficient for dense graphs, but is inefficient on sparse graphs.
    Can we compute a low diameter small neighborhood decomposition in $m^{1 + o(1)}$ time? 
    Are there other applications of low diameter small neighborhood decompositions?
    
    \item Currently, our algorithm for computing \mssp~is randomized, since our algorithm for computing our low diameter small neighborhood decomposition requires randomness. Can we derandomize our low diameter small neighborhood decomposition algorithm?

    \item We use \mssp~to develop faster algorithms for hop-set, shortcut-set, and distance emulator constructions. 
    It would be interesting to find further applications of an efficient \mssp~algorithm.
\end{enumerate}

\paragraph{Outline.}
In \Cref{sec:overview}, we give an overview of our main results and techniques.
In \Cref{sec:decomposition}, we present our algorithm for computing low diameter small neighborhood decompositions.
In \Cref{sec:mssp}, we give our algorithms for computing \mssp{} on undirected graphs and their applications.
In \Cref{sec:dir}, we present our algorithms on directed graphs and their applications.

%% file: overview.tex
\section{Technical Overview}
\label{sec:overview}

We give an overview of our techniques.
Recall that our goal is to design an \mssp~algorithm that performs better than the naive $\tO{\min(|S_{\inT}|n^2, n^{\omega+o(1)})}$ time obtained by either running \sssp~from every source or computing \apsp.
We begin with a simple algorithm that makes an improvement for some regimes of $|S_{\inT}|$.

\subsection{Warm-Up: A Simple Improvement}

Our simple algorithm relies on two key observations: (1) \mssp~can be computed efficiently on graphs with small diameter, and (2) \mssp~can be computed efficiently when each layer of the shortest path tree is small. 
On graphs with diameter at most $D$, we can compute \mssp~by taking the adjacency matrix and computing the matrix product $A[S, V] A^{d}$ for every $d \le D$ in $\bigO{D \TMUL(|S_{\inT}|, n, n)}$ time by iteratively multiplying the adjacency matrix $A$ on the right.
In graphs where each layer of the shortest path tree (i.e. $\set{v \given \dist(s, v) = \ell}$) has at most $T$ vertices, executing \sssp~for a single vertex only requires $O(nT)$ time (in contrast to $O(n^2)$ in general).
This leads us to the following algorithm.
Let $L(s, r) := \set{v \given \dist(s, v) = r}$ be the $r$-th layer of the shortest path tree from source $s$.

\begin{mdframed}
    \begin{enumerate}
        \item Let $T$ be a parameter to be set later. \Comment{(We will set $T \gets \sqrt{\TMUL(|S_{\inT}|, n, n)/|S_{\inT}|}$).}
        
        \item Initialize $L(s, 0) = \set{s}$ for all $s$. Set all $s$ to active.

        \Comment{(Execute \bfs~from all $s$ in parallel.)}
        
        \item While not all distances computed:
        
        \begin{enumerate}
            \item If some source $s$ is active: \Comment{(Small Layer Case)}
            \begin{enumerate}
                \item  Attempt to compute $L(s, \level(s) + 1)$ by relaxing all edges from $L(s, \level(s))$.
            
                \item If $|L(s, \level(s) + 1)| \leq T$, compute $L(s, \level(s) + 1)$ and increment $\level(s)$.
            
                \item If $|L(s, \level(s) + 1)| > T$, set $s$ as inactive.
            \end{enumerate}
            
            \item If all sources $s$ are inactive: \Comment{(Large Layer Case)}
            \begin{enumerate}
                \item Compute $L(s, \level(s) + 1)$ simultaneously for all $s$ by computing $AB$ where:
                \begin{equation*}
                    A[s, u] = \ind{\dist(s, u) = \level(s)},\quad B[u, v] = \ind{(u, v) \in E}
                \end{equation*} 

                \item Increment $\level(s)$ for all $s$ and set all sources $s$ to active.
            \end{enumerate}
        \end{enumerate}
    \end{enumerate}
\end{mdframed}

Our algorithm computes \bfs~in parallel from every source.
For each source $s$, we separately analyze the time required to compute large layers ($>T$) and small layers ($\leq T$) of the shortest path tree.
For small layers (layers $r$ where $\max(|L(s, r - 1)|, |L(s, r)|, |L(s, r + 1)|) \leq T$) to extend layer $r$ to $r + 1$ requires time $O(|L(s, r)|T)$, 
which overall requires $O(|S_{\inT}|nT)$.
For large layers, we note that each source $s$ can only have $O(n/T)$ large layers, and therefore become inactive at most $O(n/T)$ times.
Since we only compute the matrix product when all sources are inactive, all matrix products require $\bigO{\frac{n}{T} \TMUL(|S_{\inT}|, n, n)}$ time.
Balancing terms, we set $T \gets \sqrt{\TMUL(|S_{\inT}|, n, n)/|S_{\inT}|}$ and obtain a running time of $\bigO{n \sqrt{|S_{\inT}|\TMUL(|S_{\inT}|, n, n)}}$.
For example, for $|S_{\inT}| \leq n^{0.32}$, $\TMUL(|S_{\inT}|, n, n) = n^{2+o(1)}$~\cite{alman2025fmm}, which yields a $\sqrt{|S_{\inT}|}n^{2+o(1)}$ running time, improving upon $\bigO{|S_{\inT}|n^2}$.

\subsection{Optimal Multiple-Source Shortest Paths}

The above algorithm suggests a potential more generic approach: when we try to grow the BFS tree of some source by one layer, we could wait until many sources need to do so, and group the computation into a matrix multiplication. To further improve the running time, it is natural to use more thresholds, instead of a single threshold $T$. Now when we group multiple sources together, we need a better bound on the size of the union of the corresponding layers. In order to really take advantage of batch processing, we need these layers to have large overlaps -- if all the layers are disjoint, then batch processing does not help at all. In the above simple improvement, we only perform batch processing for very large layers, and we simply used $n$ as the upper bound on the size of the union. 
For smaller layers, why wouldn't layers of different sources be mostly disjoint (like random sets)? 
Hence, in order to proceed, we develop some structural insights of undirected unweighted graphs, in particular, a novel low diameter decomposition. 

Typically, low diameter decompositions require the number of edges between clusters to be small (see e.g. \cite{DBLP:conf/focs/BernsteinNW22}).
For our application however, we will utilize something different: the sum of the size of the neighborhoods of each cluster is almost linear (i.e. $n^{1 + o(1)}$).\footnote{To avoid confusion, we call our decomposition a low diameter small neighborhood decomposition.}
Formally, we develop a decomposition $V_1 \sqcup \dots \sqcup V_k$ of the vertices satisfying:
\begin{enumerate}
    \item $\diam(V_i) = n^{o(1)}$ for all $i$.
    \item $\sum_{i} |N(V_i)| = n^{1 + o(1)}$.
\end{enumerate}

Before describing how to obtain such a decomposition (see \Cref{thm:boundary-decomp-alg}), let us see how it is useful.
Recall that we use $\level(s)$ to index the most recently computed layer of the shortest path tree rooted at $s$.
Now, instead of computing $L(s, \level(s) + 1)$ from $L(s, \level(s))$ by considering all pairs of vertices, we compute $L(s, \level(s) + 1) \cap N(V_i)$ from $L(s, \level(s)) \cap V_i$ for all $i$ where $V_i \cap L(s, \level(s)) \neq \emptyset$.
Note that correctness follows from the fact that $\{V_i\}_{i=1}^k$ cover $L(s, \level(s))$ so that $\{N(V_i)\}_{i=1}^k$ cover $L(s, \level(s) + 1)$.

Since we do not consider $L(s, \level(s))$ all at once, we need to keep track of which vertices in the layer whose neighborhoods we have already examined.
Let $\calA(s)$ denote the vertices in $L(s, \level(s))$ which we still need to examine, called the active vertices.
Note that $\calA(s)$ is initialized to $L(s, \level(s))$ and each iteration of the while loop removes some $V_i$ from $\calA(s)$.

We say $V_i$ is relevant to $s$ if $V_i \cap \calA(s) \neq \emptyset$.
For each cluster $V_i$, denote $\score(i)$ the number of sources $V_i$ is relevant to.
We say a cluster $V_i$ is \emph{ready} when $\score(i) \geq \frac{|S_{\inT}||N(V_i)|}{n^{1 + o(1)}}$.
Observe that there is always one ready cluster.
If not, by the second property of our decomposition, we have $\sum_{i} \score(i) < |S_{\inT}|$.
However, since there is at least one cluster relevant to each source, we obtain a contradiction via $\sum \score(i) \geq |S_{\inT}|$.

We now describe our improved algorithm (see \Cref{thm:mssp-reduction} for full algorithm).
\begin{mdframed}
    \begin{enumerate}
        \item Let $\calA(s) \gets L(s, \level(s))$ for all $s$.
        \Comment{(Recall $\level(s) \gets 0$ and $L(s, 0) = \set{s}$.)}
        \item While not all distances computed, select a ready cluster $V_i$:
        \begin{enumerate}
            \item Compute $L(s, \level(s) + 1) \cap N(V_i)$  from $L(s, \level(s)) \cap V_i$ for all $s$ that $V_i$ is relevant to.

            \item Remove $V_i$ from $\calA(s)$.

            \item For any $\calA(s) = \emptyset$, increment $\level(s)$ and set $\calA(s) \gets L(s, \level(s))$.
        \end{enumerate}
    \end{enumerate}
\end{mdframed}

Correctness follows from our discussion above.
We proceed to analyze the running time, where we will bound the time required by all iterations of the while loop where $V_i$ is chosen as the ready cluster.

First, we claim that for any fixed source $s$, $V_i$ is relevant to $s$ for $n^{o(1)}$ iterations (see \Cref{lem:s-iteration-bound}).
This follows from the first property of our decomposition ($\diam(V_i) \leq n^{o(1)}$) so that $V_i$ can intersect at most $n^{o(1)}$ distinct layers of the shortest path tree of $s$.
The result follows from observing that $V_i$ is relevant to $s$ for one iteration per layer of the shortest path tree (i.e. one value of $\level(s)$).

Second, we claim that $V_i$ is selected as a ready cluster in $\frac{n^{1 + o(1)}}{|N(V_i)|}$ iterations (see \Cref{lem:round-ub}).
Let $T(i)$ denote all iterations where $V_i$ is selected.
Let $S'(t)$ denote the set of sources $V_i$ is relevant to in iteration $t \in T(i)$.
Then, over all iterations where $V_i$ is selected, we have $\sum_{t} |S'(t)| = |S_{\inT}|n^{o(1)}$.
On the other hand, since $V_i$ is ready, $\score(i) = |S'(t)| \geq \frac{|S_{\inT}||N(V_i)|}{n}$.
Thus, we have $|T(i)| \leq \frac{n^{1 + o(1)}}{|N(V_i)|}$.

We now bound the time required by all iterations where $V_i$ is selected, breaking our analysis into two cases: small clusters and large clusters.
For small clusters with $|N(V_{i})| \leq \frac{n}{|S_{\inT}|}$, the $t$-th iteration requires time $O(S'(t)|V_i||N(V_i)|)$.
Since $\sum_{t} S'(t) = |S_{\inT}|n^{o(1)}$, we bound the total time as $\sum_{t} O(S'(t)|V_i||N(V_i)|) = n^{1 + o(1)} |V_i|$.
Summing over all small clusters, the total time bound is $n^{2 + o(1)}$.

For large clusters, the $t$-th iteration involves a matrix product of dimension $S'(t) \times V_i \times N(V_i)$.
Summing over all iterations, we obtain (see \Cref{lem:mat-mul-convex} for the deduction)
\begin{equation*}
    \sum_{t \in T(i)} \TMUL(S'(t), V_i, N(V_i)) \le \bigO{\frac{n^{1 + o(1)}}{|N(V_i)|} \TMUL\left( \frac{|S_{\inT}||N(V_i)|}{n}, V_i, N(V_i) \right)}.
\end{equation*}
Let $Q := n / |N(V_i)|$  so that the time is $n^{o(1)} \cdot Q \TMUL(|S_{\inT}| / Q, n / Q, n / Q)$.
Since there are at most $n^{1 + o(1)}/|N(V_i)|$ clusters of size $O(|N(V_i)|)$, we may bound the total time over all large clusters as $n^{o(1)} \cdot Q^2 \TMUL(|S_{\inT}| / Q, n / Q, n / Q)$, which can be upper bounded by $n^{o(1)} \cdot \TMUL(|S_{\inT}|, n, n)$ as desired (see the running time analysis of \Cref{thm:mssp-reduction} for the deduction).

\subsection{Low Diameter Small Neighborhood Decomposition}

We now describe how to efficiently obtain a low diameter small neighborhood decomposition (\Cref{thm:boundary-decomp-alg}).
We say $V_1 \sqcup \dots \sqcup V_k$ is a $(d, \phi)$-low diameter small neighborhood decomposition (\LDND) if $\max_{i} \diam(V_i) \leq d$ and $\sum_{i} |N(V_i)| \leq \phi n$ (see \Cref{def:boundary-decomp}).

Consider the following algorithm for obtaining a single low diameter cluster.
Pick an arbitrary vertex $v$ and let $N(v, r) := \set{u \in V: \dist(v, u) \leq r}$ denote the $r$-neighborhood of $v$.
We claim $|N(v, 3^{k})| \leq \phi^{k}$ for some $k \leq \lceil\log_{\phi} n\rceil $. 
Otherwise $|N(v, 3^{\lceil\log_{\phi} n\rceil})| > n$ becomes a contradiction (see \Cref{fig:one-cluster}).
Setting $\phi = 2^{O(\sqrt{\log n})}$, we get that $k \leq O(\sqrt{\log n})$ and $3^{k} \leq 2^{O(\sqrt{\log n})}$ so that $N(v, 3^{k})$ has small diameter.

\begin{figure}
    \centering
    \input{tikz/one_cluster}
    \caption{Process of growing one cluster until it stops expanding.
    Here, $\phi = 3$, and we look for the first $k$ such that $|N(v, 3^{k})| \leq 3^{k}$.
    Each shaded circle illustrates $N(v, 3^{k})$ for some $k$.
    The process must terminate before $k = \log_{3} n$.}
    \label{fig:one-cluster}
\end{figure}
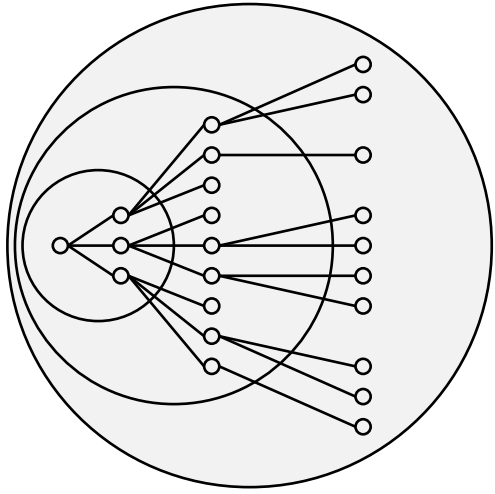

Motivated by this observation, our decomposition $\calC$ will consist of balls $N(v, 3^{k})$ for some $v$ and $k$.
For each $0 \leq k \leq \log_{\phi} n$, we sample $S_{k} \subseteq V$ by including each vertex with probability $p_{k} \sim \phi^{-k}$.
Then, for each $u \in S_{k}$, we add $N(u, 3^{k})$ to $\calC$ if $|N(u, 3^{k} + 1)| \leq \phi^{k + 1}$.
Since $|S_{k}| = O(n\phi^{-k})$ with high probability, the total time required is at most $O(\sum_{k} |S_{k}| \phi^{2k + 2}) = O(\sum_{k} n \phi^{k+2}) = n^{2 + o(1)}$.

Above, we have argued that each set has small diameter.
Furthermore, since $|S_{k}| = O(n\phi^{-k})$ with high probability, we have that $\sum_{u \in S_{k}} |N(u, 3^{k} + 1)| = O(n \phi) = n^{1 + o(1)}$, as desired.

It remains to argue that every vertex $v \in N(u, 3^{k})$ for some $N(u, 3^{k}) \in \calC$.
Let $k \geq 1$ be the smallest positive integer satisfying $|N(v, 3^{k})| \leq \phi^{k}$ so that $|N(v, 3^{k - 1})| > \phi^{k - 1}$. 
Note that such a $k$ always exists since $|N(v, 1)| > 1$ unless $v$ is an isolated vertex.
Then, with high probability, some $u \in S_{k - 1}$ hits $N(v, 3^{k - 1})$.
Then, since $\dist(u, v) \leq 3^{k - 1}$ and $2 \cdot 3^{k - 1} + 1 \leq 3^{k}$, we have $|N(u, 3^{k - 1} + 1)| \leq |N(v, 2 \cdot 3^{k - 1} + 1)| \leq |N(v, 3^{k})| \leq \phi^{k}$ 
, resulting in $N(u, 3^{k - 1})$ being added to $\calC$, as desired.

Note that the above process gives a covering of the vertices. We can obtain a decomposition by removing duplicated vertices from the clusters.

\subsection{Extension to Bounded Weights}

Finally, we describe how to extend the algorithm to handle weighted graphs with small integer weights bounded by $B$.
Note that an algorithm that computes \mssp~on graphs with integer weights bounded by $\Theta(B)$ also computes the rectangular $(\min, +)$-product of two integer matrices with entries bounded by $B$.
Given $S_{\inT} \times n$ matrix $X$ and $n \times n$ matrix $Y$, construct a three-layer graph with $V_1 = S_{\inT}, V_2 = [n], V_3 = [n]$ and edges $(s, u)$ for $s \in V_1, u \in V_2$ with weight $10B + X[s, u]$ and edges $(u, v)$ for $u \in V_2, v \in V_3$ with weight $10B + Y[u, v]$.
Then, $\dist(s, v) = 20B + (X \star Y)[s, v]$ whenever $(X \star Y)[s,v] < \infty$ and is at least $40B$ otherwise.

Thus, our goal is to develop an algorithm that runs in $\TMPP(|S_{\inT}|, n, n \mid B) = \tO{B\TMUL(|S_{\inT}|, n, n)}$ time, where $\TMPP(n^{a}, n^{b}, n^{c} \mid B)$ denotes the time required to compute the $(\min, +)$-product between integer matrices of dimension $n^{a} \times n^{b} \times n^{c}$ with integer entries bounded by $B$.

We first describe an $\hat{O}(B^2 \TMUL(|S_{\inT}|, n, n))$ time algorithm.
Treat $G$ as an unweighted graph, and compute an $(n^{o(1)}, n^{o(1)})$-low diameter small neighborhood decomposition.

Now, fix a source $s$.
To ensure correctness, in order to compute $L(s, \level(s) + 1)$ we need two modifications to account for the fact that edge weights can be up to $B$.
We need to (1) consider any cluster $V_i$ that intersects $L(s, r)$ for any $\level(s) - B \leq r \leq \level(s)$, and (2) compute $N(V_i) \cap L(s, \level(s) + 1)$ via a $B$-bounded $(\min, +)$-product.
Revisiting the analysis of our unweighted algorithm, we obtain a running time of $\hat{O}(B^2 \TMUL(|S_{\inT}|, n, n))$ since (1) implies that $V_i$ is relevant to $s$ for $Bn^{o(1)}$ iterations (in contrast to $n^{o(1)}$), while (2) implies that we must use $B$-bounded $(\min, +)$-products in place of a Boolean product.

To ensure that each cluster is relevant to any source $s$ only for $n^{o(1)}$ iterations of the while loop, we instead take the following approach.
Consider graphs $G_{b}$ for $b \leq \log B$ consisting of all edges with weight between $2^{b}$ and $2^{b + 1}$.
We obtain low diameter small neighborhood decompositions for all $G_{b}$ as unweighted graphs, denoted $\calC_{b} = \set{V_{i}^{(b)}}$.
Now, while each cluster $V_{i}^{(b)}$ has diameter $2^{b}n^{o(1)}$, we only need to consider cluster edges in $G_{b}$ incident to $V_{i}^{(b)}$ for every $2^{b}$ layers of the shortest path tree from $s$, since each of these edges has weight at least $2^{b}$.
Our final algorithm for bounded weights selects a cluster that is relevant to many sources, and relaxes all edges of the appropriate weight range incident to this cluster.
To ensure correctness and efficiency, we carefully manage which vertices are active for each source and edge weight class, ensuring that the shortest path tree of each source is built accurately (see \Cref{thm:bounded-weight-mssp} for details).

\subsection{Multiple Source Reachability}

Our algorithm for Multiple Source Reachability (and \mssp{} on DAGs) employs a simple recursive scheme.
Assume that $G$ is sorted in topological order.
Let $\pred(\pi, v)$ denote the predecessor of $v$ in path $\pi$.
For any interval $[L, R]$ of $\ell$ vertices, we employ a recursive scheme $\complete(G, S_{\inT}, L, R)$ with the following input/output conditions:

Input: $\reach(s, v)$ is computed if there is a path $\pi$ from $s \in S_{\inT}$ to $v \leq R$ where $\pred(\pi, v) \leq L$.
    
Output: $\reach(s, v)$ is computed for all $s \in S_{\inT}$ and $v \leq R$.

Our Multiple Source Reachability algorithm computes $\complete(G, S_{\inT}, 0, |V|)$.
In this overview, we describe how to compute $\complete(G, S_{\inT}, L, R)$.
Let $M = (L+R)/2$ be the midpoint of the interval.
Note that we can compute $\complete(G, S_{\inT}, L, M)$ since the input condition is satisfied.
Our goal is to satisfy the input condition of $\complete(G, S_{\inT}, M, R)$ which would satisfy the output condition of $\complete(G, S_{\inT}, L, R)$.
To do so, we need to ensure $\reach(s, v)$ is computed for any path $\pi$ from $s$ to $v$ with $L \leq \pred(\pi, v) \leq M$.
We can do this with a single computation of \BMM{} by constructing matrix $X$ as a $S_{\inT} \times [L, M]$ containing $\reach(s, u)$ for $u \in [L, M]$ and a matrix $Y$ as a $[L, M] \times [M, R]$ submatrix of the adjacency matrix, ensuring that all paths passing through $[L, M]$ are accounted for.
The running time of the recursion is thus
\begin{equation*}
    T(\ell) = 2 T(\ell/2) + \TBMM(S_{\inT}, \ell, \ell)
\end{equation*}
which yields $\TBMM(S_{\inT}, n, n)$ by appealing to the root-heavy case of the Master Theorem.
Our algorithm for \mssp{} on DAGs employs a similar recursive structure, replacing each \BMM{} instance with an appropriate Min-Plus Product.

%% file: tikz/one_cluster.tex
\begin{tikzpicture}[scale=1]
  \draw [fill=gray!10, line width=1pt] (2.5, 0) circle (3.2);
  
  \draw [fill=gray!10, line width=1pt] (1.5, 0) circle (2.1);

  \draw [fill=gray!10, line width=1pt] (0.5, 0) circle (1);

  \draw[line width=1pt] (0, 0) circle (0.1);
  
  \draw[line width=1pt] (0.8, 0.4) circle (0.1);
  \draw[line width=1pt] (0.8, 0) circle (0.1);
  \draw[line width=1pt] (0.8, -0.4) circle (0.1);

  \draw[line width=1pt] (2, 1.6) circle (0.1);
  \draw[line width=1pt] (2, 1.2) circle (0.1);
  \draw[line width=1pt] (2, 0.8) circle (0.1);
  \draw[line width=1pt] (2, 0.4) circle (0.1);
  \draw[line width=1pt] (2, 0) circle (0.1);
  \draw[line width=1pt] (2, -0.4) circle (0.1);
  \draw[line width=1pt] (2, -0.8) circle (0.1);
  \draw[line width=1pt] (2, -1.2) circle (0.1);
  \draw[line width=1pt] (2, -1.6) circle (0.1);

  \draw[line width=1pt] (4, 2.4) circle (0.1);
  \draw[line width=1pt] (4, 2) circle (0.1);
  \draw[line width=1pt] (4, 1.2) circle (0.1);
  \draw[line width=1pt] (4, 0.4) circle (0.1);
  \draw[line width=1pt] (4, 0) circle (0.1);
  \draw[line width=1pt] (4, -0.4) circle (0.1);
  \draw[line width=1pt] (4, -0.8) circle (0.1);
  \draw[line width=1pt] (4, -1.6) circle (0.1);
  \draw[line width=1pt] (4, -2) circle (0.1);
  \draw[line width=1pt] (4, -2.4) circle (0.1);

  \draw[line width=1pt] (0.1, 0) -- (0.7, 0.4);
  \draw[line width=1pt] (0.1, 0) -- (0.7, 0);
  \draw[line width=1pt] (0.1, 0) -- (0.7, -0.4);

  \draw[line width=1pt] (0.9, 0.4) -- (1.9, 1.6);
  \draw[line width=1pt] (0.9, 0.4) -- (1.9, 1.2);
  \draw[line width=1pt] (0.9, 0.4) -- (1.9, 0.8);

  \draw[line width=1pt] (0.9, 0) -- (1.9, 0.4);
  \draw[line width=1pt] (0.9, 0) -- (1.9, 0);
  \draw[line width=1pt] (0.9, 0) -- (1.9, -0.4);

  \draw[line width=1pt] (0.9, -0.4) -- (1.9, -1.6);
  \draw[line width=1pt] (0.9, -0.4) -- (1.9, -1.2);
  \draw[line width=1pt] (0.9, -0.4) -- (1.9, -0.8);

  \draw[line width=1pt] (2.1, 1.6) -- (3.9, 2.4);
  \draw[line width=1pt] (2.1, 1.6) -- (3.9, 2);

  \draw[line width=1pt] (2.1, 1.2) -- (3.9, 1.2);

  \draw[line width=1pt] (2.1, 0) -- (3.9, 0);
  \draw[line width=1pt] (2.1, 0) -- (3.9, 0.4);

  \draw[line width=1pt] (2.1, -0.4) -- (3.9, -0.4);
  \draw[line width=1pt] (2.1, -0.4) -- (3.9, -0.8);

  \draw[line width=1pt] (2.1, -1.2) -- (3.9, -1.6);
  \draw[line width=1pt] (2.1, -1.2) -- (3.9, -2);

  \draw[line width=1pt] (2.1, -1.6) -- (3.9, -2.4);
\end{tikzpicture}

%% file: prelims.tex
\section{Preliminaries}

Let $[n] = \set{1, \dots, n}$. 
Let $\log$ denote the base-2 logarithm, unless stated otherwise.
For any set $S$, let $|S|$ denote the size of the set.
For an empty set, we use the conventions $\min(\emptyset) = \infty$ and $\max(\emptyset) = - \infty$.
In our algorithms, we use $x \addget y$ to denote $x \gets x + y$ and $x \minget y$ to denote $x \gets \min(x, y)$.
For sets, we use $A \cupget B$ to denote $A \gets A \cup B$.

For a graph $G = (V, E)$, we typically use $n = |V|$ and $m = |E|$ to denote the number of vertices and edges respectively.
We use $G = (V, E, \wt)$ to denote a weighted graph with weight function $\wt: E \mapsto \Z$.
For any $S \subseteq V$, $E(S)$ denotes all edges incident to some vertex in $S$.
Unless otherwise stated, all graphs are undirected.
Furthermore, we assume (without loss of generality) that all graphs are connected.

For any two vertices $s, t \in V$, let $\dist_{G}(s, t)$ denote the length of the shortest path from $s$ to $t$. 
When the context is clear, we omit $G$ and write $\dist(s, t)$.
For any path $\pi \subseteq G$ (not necessarily shortest), let $|\pi|$ denote the length, or number of edges, of $\pi$.
For a subset of vertices $S \subseteq V$ and a vertex $v \in V$, let $\dist(S, v) := \min_{s \in S} \dist_{G}(s, v)$.
For any vertex $s \in V$, let $N(s, r) := \set{v \given \dist(s, v) \leq r}$ and for a subset of vertices $S \subset V$, let $N(S, r) := \set{v \given \dist(S, v) \leq r}$.
Let $N(v) := N(v, 1)$ and $N(S) := N(S, 1)$.
Let $\diam_{G}(S) := \max_{s, t \in S} \dist_{G}(s, t)$ denote the diameter of a subset $S$.
For a given vertex $s \in V$ and distance $r$, let $L(s, r) := \set{v \given \dist(s, v) = r}$.

Given a matrix $A$, we let $A[i, j]$ denote the entry of $A$ in the $i$-th row and $j$-th column.
Let $0_{a \times b}$ denote an $a \times b$ matrix of $0$s.
Let $\TMUL(n_1, n_2, n_3)$ denote the time required to compute a $n_1 \times n_2 \times n_3$ dimensional matrix product i.e. the product of an $n_1 \times n_2$ matrix and an $n_2 \times n_3$ matrix.
Given matrices $A, B$, we let $AB$ denote the standard matrix product (either Boolean or integer), and $A \star B$ the $(\min, +)$ (or Min-Plus) matrix product, where $(A \star B)[i,j] = \min_{k} \{A[i,k] + B[k,j]\}$.
Let $\TBMM(n_1, n_2, n_3)$ denote the time required to compute a $n_1 \times n_2 \times n_3$ dimensional boolean matrix product, and $\TMPP(n_1, n_2, n_3 \mid B_1, B_2)$ the time required to compute the $(\min, +)$ product of a $n_1 \times n_2$ matrix with entries in $[B_1] \cup \set{\infty}$ and a $n_2 \times n_3$ matrix with entries in $[B_2] \cup \set{\infty}$.
Let $\TMPP(n_1, n_2, n_3 \mid B) := \TMPP(n_1, n_2, n_3 \mid B, B)$. 
In some cases, we write $\TBMM(A, B, C) := \TBMM(|A|, |B|, |C|)$ where $A, B, C$ are sets.
We also allow combinations of sets and integers, e.g. for set $A$ and integers $n_2, n_3$, $\TBMM(A, n_2, n_3) := \TBMM(|A|, n_2, n_3)$.

%% file: decomposition.tex
\section{Low Diameter Small Neighborhood Decomposition}
\label{sec:decomposition}

In this section, we present an efficient algorithm for computing a low diameter small neighborhood decomposition.
We begin by defining the desired properties.

\begin{definition}
    \label{def:boundary-decomp}
    A partition of $V$ into subsets $V_{1}, \dots, V_{k}$ is a $(d, \phi)$-low diameter small neighborhood decomposition (\LDND) of $G = (V, E)$ if:
    \begin{enumerate}
        \item (Low Diameter) $\diam_{G}(V_{i}) \leq d$ for all $i \in [k]$.
        \item (Total Neighborhood Size) $\sum_{i = 1}^{k} |N(V_{i})| \leq \phi n$.
    \end{enumerate}
\end{definition}

We now give an efficient algorithm for computing a low diameter small neighborhood decomposition.

\newcommand{\expandbase}{3}

\begin{theorem}
    \label{thm:boundary-decomp-alg}
    For any $\psi > 1$, there is a randomized $\bigO{n^2 \psi^{4} \log n \log_{\psi} n}$ time algorithm that given an $n$-vertex $m$-edge undirected, unweighted graph $G$ computes a $(d, \phi)$-\LDND{} where $d = 6 \cdot \expandbase^{\log_{\psi} n}$ and $\phi= \bigO{\psi \log n \log_{\psi} n}$
    Furthermore, this algorithm computes $N(V_i)$ for all subsets $V_i$ in the decomposition.
\end{theorem}

\begin{proof}
    We give an algorithm that computes a collection of subsets $V_1, \dots, V_{k}$ such that (1) $V = V_1 \cup \dots \cup V_{k}$, (2) $\diam_{G}(V_i) \leq d$ for all $i$, and (3) $\sum_{i} |N(V_i)| \leq \phi n$.
    From here, it is easy to observe that by removing duplicate vertices, we obtain a partition that is a $(d, \phi)$-\LDND, as removing vertices can only reduce the diameter and the total neighborhood size.
    
    First, we show that any vertex $v$ has a small-diameter neighborhood that does not expand.

    \begin{claim}
        \label{clm:t-existence}
        For every $v \in V$, there exists $0 \leq k < \lceil \log_{\psi} n \rceil$ such that $|N(v, \expandbase^{k + 1})| \leq \psi^{k + 1}$ and $|N(v, \expandbase^{k})| > \psi^{k}$.
    \end{claim}

    \begin{proof}
    Observe that $|N(v, \expandbase^{0})| = |N(v, 1)| \ge 2 > \psi^0$ (recall we assume the graph is connected). Also, $|N(v, \expandbase^{\lceil \log_{\psi} n \rceil})| = n \le \psi^{\lceil \log_{\psi} n \rceil}$. The claim thus follows. 
    \end{proof}

    We now argue that an appropriate collection of $N(v, \expandbase^{k})$ suffice to provide a low diameter small neighborhood decomposition of $G$.
    Consider the following procedure. 

    \begin{mdframed}
        \begin{enumerate}
            \item Let $\calC \gets \emptyset, \calN \gets \emptyset$.
            \item For $0 \leq k < \lceil \log_{\psi} n \rceil$:
            \begin{enumerate}
                \item Sample a random subset $S_{k} \subseteq V$ by including each $u \in V$ with probability $p_{k} = \min\left(1, \frac{3 \log n}{\psi^{k}}\right)$.
                \item For each $u \in S_{k}$:
                \begin{enumerate}
                    \item Compute $N(u, \expandbase^{k} + 1)$ using \bfs.
                    
                    \item If $|N(u, \expandbase^{k} + 1)| \leq \psi^{k + 1}$: update $\calC \cupget \set{N(u, \expandbase^{k})}$ and $\calN \cupget \set{N(u, \expandbase^{k} + 1)}$. 

                    \Comment{(Whenever $|N(u, \expandbase^{k} + 1)| > \psi^{k + 1}$, we terminate the invocation of \bfs~early.)}
                \end{enumerate}
            \end{enumerate}
            \item Return $\calC, \calN$.
        \end{enumerate}
    \end{mdframed}

    We claim that our algorithm returns a low diameter small neighborhood decomposition.
    We begin with the first property.
    Let $v \in V$. Let $k$ be an arbitrary integer guaranteed by \Cref{clm:t-existence}.
    We claim that $S_{k}$ contains a vertex in $N(v, \expandbase^{k})$ with high probability.
    The claim holds trivially whenever $p_{k} = 1$ (i.e. $S_{k} = V$).
    For $p_{k} < 1$, since $|N(v, \expandbase^{k})| > \psi^{k}$,
    \begin{equation*}
        \Pr(N(v, \expandbase^{k}) \cap S_{k} = \emptyset) \leq \left( 1 - p_{k} \right)^{\psi^{k}} < e^{-p_{k} \psi^{k}} < \frac{1}{n^{3}}.
    \end{equation*}
    
    Let $u \in S_{k} \cap N(v, \expandbase^{k})$.
    Then, since $u \in N(v, \expandbase^{k})$, we have $v \in N(u, \expandbase^{k})$. 
    We now argue that $N(u, \expandbase^{k}) \in \calC$.
    To see this, since $k \geq 0$, we have $N(u, \expandbase^{k} + 1) \subseteq N(v, 2 \cdot \expandbase^{k} + 1) \subseteq N(v, \expandbase^{k + 1})$ so that 
    $|N(u, \expandbase^{k} + 1)| \leq |N(v, \expandbase^{k + 1})| \leq \psi^{k + 1}$, as desired.
    Hence, by union bound, the collection of sets in $\calC$ covers all vertices in $V$ with probability at least $1 - 1/n^2$. 

    Towards the second property, we bound the diameter of each set $N(u, \expandbase^{k})$ with $k \leq \lceil \log_{\psi} n \rceil$ by 
    \begin{equation*}
        2 \left( \expandbase^{k} \right) \leq 2 \cdot 3^{\lceil \log_{\psi} n \rceil} \leq 6 \cdot 3^{\log_{\psi} n} \text{.}
    \end{equation*}

    Towards the third property, we can bound the size of all $N(V_{i}) = N(u, 3^{k} + 1)$ by $\sum_{k} |S_{k}| \psi^{k + 1}$.
    Using a standard Chernoff bound, we can conclude with high probability that all $|S_{k}| = \bigO{\frac{n \log n}{\psi^{k}}}$.
    In particular, $\sum_{k} |S_{k}| \psi^{k + 1} = \bigO{\psi n \log n \log_{\psi} n}$.
    Thus, we have achieved a $\left( 6 \cdot \expandbase^{\log_{\psi} n} , \bigO{\psi \log n \log_{\psi} n} \right)$-\LDND.

    It remains to prove the running time.
    For a fixed $u \in S_{k}$, since we terminate whenever $N(u, 3^{k} + 1) > \psi^{k + 1}$, we can bound the running time of each \bfs{} by $\bigO{\psi^{2k + 2}}$ since the number of edges incident to $N(u, 3^{k})$ can be bounded by $O(|N(u, 3^{k} + 1)|^2)$.
    Summing over all $S_{k}$, the total running time is
    \begin{equation*}
        \sum_{k} |S_{k}| \psi^{2k + 2} = \bigO{n \psi^{\log_{\psi} n + 4} \log n \log_{\psi} n} = \bigO{n^{2} \psi^{4} \log n \log_{\psi} n} \text{.}
    \end{equation*}
    Above, we have used that $|S_{k}| = \bigO{\frac{n \log n}{\psi^{k}}}$ and $\psi^{k} \leq \psi^{\lceil \log_{\psi} n \rceil}$.
\end{proof}

From our above theorem, we obtain a specific decomposition.

\Decomposition*

\begin{proof}
    Set $\psi = 2^{\sqrt{\log n}}$ and apply \Cref{thm:boundary-decomp-alg}.
    We obtain an $n^{2} \cdot 2^{\bigO{\sqrt{\log n}}}$ time algorithm that computes a $\left(2^{\bigO{\sqrt{\log n}}}, 2^{\bigO{\sqrt{\log n}}}\right)$-\LDND.
\end{proof}

We remark that by appropriately setting parameters, we can obtain a low diameter small neighborhood decomposition where exactly one of the diameter and expansion factor is polylogarithmic and both are sub-polynomial.
However, for our applications, it suffices to obtain a decomposition where both are sub-polynomial.

%% file: mssp.tex
\section{Undirected Multiple Source Shortest Paths}
\label{sec:mssp}

In this section, we present our algorithms for \mssp{} in undirected graphs.

\subsection{Undirected Unweighted Multiple Source Shortest Paths}
\label{sec:mssp-unweighted}
We begin with our algorithm for undirected, unweighted graphs.

\begin{restatable}{theorem}{MSSPReduction}
    \label{thm:mssp-reduction}
    Given a $(d, \phi)$-low diameter small neighborhood decomposition, there is a deterministic algorithm computing Multiple Source Shortest Paths (\mssp) on undirected, unweighted graphs in time \[\bigO{n^{2} \phi d \log n + \TBMM\left( \frac{|S_{\inT}|}{\phi \log n}, n, n \right)  \phi^2 d (\log n)^2}.\]
\end{restatable}

Combined with \Cref{cor:boundary-decomposition}, we prove our main result.

\let\oldfootnote\footnote
\renewcommand{\footnote}[1]{\textup{\textsuperscript{\ref{fn:mssp}}}}
\MSSP*
\let\footnote\oldfootnote

\begin{proof}[Proof of \Cref{cor:mssp}]
    We obtain an $(n^{o(1)}, n^{o(1)})$-low diameter small neighborhood decomposition in $n^{2 + o(1)}$ time via \Cref{cor:boundary-decomposition}.
    Applying \Cref{thm:mssp-reduction}, we compute \mssp~in $n^{o(1)} \cdot \TBMM \left(|S_{\inT}|, n, n\right)$ time.
\end{proof}

We prove \Cref{thm:mssp-reduction} in the remainder of \Cref{sec:mssp-unweighted}. 
Consider the following procedure. 
    
\begin{mdframed}
    {\bf Input: } Graph $G = (V, E)$, sources $S_{\inT} \subseteq V$, and $(d, \phi)$-\LDND{} $\calC = \set{V_1, \dots, V_{k}}$.
    \begin{enumerate}

        \item Let $S \gets S_{\inT}$.

        \item Let $\hat{L}(s, 0) \gets \set{s}, \calA(s) \gets \set{s}$, $\hat{L}(s, 1) \gets \emptyset$, $\hdist(s, s) \gets 0, \level(s) \gets 0$, and $\hdist(s, v) \gets \infty$ for all $s \in S, v \in V \setminus \{s\}$.

        \item Let $\score(i) \gets |\set{s \in S \given \calA(s) \cap V_{i} \neq \emptyset}|$ for all $V_{i}$.

        \item While there exists $V_{i}$ with $\score(i) > \frac{|S||N(V_i)|}{2 \phi n}$:

        \begin{enumerate}

            \item Let $S' \gets \set{s \in S \given \calA(s) \cap V_{i} \ne \emptyset}$.

            \item If $|N(V_i)| \leq \frac{\phi n \log n}{|S_{\inT}|}$: 
            \begin{enumerate}
                \item For all $s \in S'$ where $u \in \calA(s) \cap V_i$ and $v \in N(u)$: update $\hdist(s, v) \minget \level(s) + 1$.
            \end{enumerate}
            Else ($|N(V_i)| > \frac{\phi n \log n}{|S_{\inT}|}$):
                \begin{enumerate}
                    \item Construct an $S' \times V_{i}$ matrix $X$ and a $V_i \times N(V_{i})$ matrix $Y$ with entries
                    \begin{equation*}
                        X[s, u] = \ind{u \in \calA(s)},\quad Y[u, v] = \ind{(u, v) \in E} \text{.}
                    \end{equation*}
                    \item Compute $XY$ and for all $(XY)[s, v] = 1$, update $\hdist(s, v) \minget \level(s) + 1$.
                \end{enumerate}
            
            \item For all $s \in S'$:
            \begin{enumerate}
                \item Update $\hat{L}(s, \level(s) + 1) \cupget \set{v \given \hdist(s, v) = \level(s) + 1}$.
                
                \item Remove $V_{i}$ from $\calA(s)$ and decrement $\score(i) \subget 1$.

                \item If $\calA(s) = \emptyset$: 
                \begin{enumerate}

                    \item Update $\level(s) \addget 1$, $\hat{L}(s, \level(s) + 1) \gets \emptyset$.

                    \item If $\bigcup_{r = 0}^{\level(s)} \hat{L}(s, r) = V$: then remove $s$ from $S$. Else, set $\calA(s) \gets \hat{L}(s, \level(s))$.

                    \item For all $V_i \cap \calA(s) \neq \emptyset$, increment $\score(i)$.
                \end{enumerate}
                
            \end{enumerate}
        \end{enumerate}
    \end{enumerate}
    {\bf Output: } $\hdist(s, v)$ for all $s \in S_{\inT}$ and $v \in V$.
\end{mdframed}

\begin{proof}[Proof of \Cref{thm:mssp-reduction}]
    We prove the correctness and running time of our algorithm.

    \paragraph{Correctness.}
    Beginning with correctness, we argue that our algorithm terminates and upon termination of our algorithm, all distances are correct i.e. $\hdist(s, v) = \dist(s, v)$ for all $s \in S_{\inT}, v \in V$.

    First, we claim that in every iteration of the while loop $\hdist(s, v) \geq \dist(s, v)$ for all $s \in S_{\inT}, v \in V$.

    \begin{lemma}
        \label{lem:dist-s-v-lb}
        In every iteration of the while loop $\hdist(s, v) \geq \dist(s, v)$ for all $s \in S_{\inT}, v \in V$.
    \end{lemma}

    \begin{proof}[Proof of \Cref{lem:dist-s-v-lb}]
        Suppose for contradiction that there is some $s \in S_{\inT}$ where $\hdist(s, v) < \dist(s, v)$ in some iteration of the while loop.
        Consider the first iteration of the while loop where a violation $\hdist(s, v) < \dist(s, v)$ occurs and let $v$ be the violating vertex.
        In particular, at the start of this iteration, we have $\hdist(s, v) \geq \dist(s, v)$.
        During this iteration, we update $\hdist(s, v) \gets \level(s) + 1 < \dist(s, v)$.
        By our two update rules, there exists $u \in \calA(s) \cap V_{i}$ with $(u, v) \in E$.
        Observe that $\calA(s) \subseteq \hat{L}(s, \level(s))$ so $u \in \hat{L}(s, \level(s))$.
        Since $u \in \hat{L}(s, \level(s))$, we have at the start of this iteration that $\dist(s, u) \leq \hdist(s, u) \leq \level(s)$.
        By the triangle inequality, for $(u, v) \in E$ we have $\dist(s, v) \leq \dist(s, u) + 1 \leq \level(s) + 1$, a contradiction.

    \end{proof}

    Now, we argue that $L(s, r) = \hat{L}(s, r)$ for all $s \in S_{\inT}$ and $r \leq \level(s)$ at the end of each iteration of the while loop.

    \begin{lemma}
        \label{lem:dist-correct-ub}
        At the end of each iteration of the while loop, $\hdist(s, v) = \dist(s, v)$ for all $v$ with $\dist(s, v) \leq \level(s)$.
        In particular, for all $r \leq \level(s)$, $L(s, r) = \hat{L}(s, r)$.
    \end{lemma}

    \begin{proof}[Proof of \Cref{lem:dist-correct-ub}]
        We proceed by induction on $\level(s)$.
        In the base case, $L(s, 0) = \set{s}$ and we initialize $\hdist(s, s) = 0$ and $\hat{L}(s, 0) = \set{s}$.
        In any future iteration of the while loop, we do not modify $\hat{L}(s, 0)$ or $\hdist(s, 0)$ since $\level(s) + 1 > 0$.

        Now, consider $\ell := \level(s) > 0$.
        Consider the iteration of the while loop where $\level(s)$ is incremented to $\ell$ for the first time.
        At the end of the previous iteration, we have $\level(s) = \ell - 1$.
        Induction then implies that $\hdist(s, u) = \dist(s, u)$ for all $\dist(s, u) < \ell$ and $L(s, r) = \hat{L}(s, r)$ for all $r < \ell$.
        
        Fix $v$ with $\dist(s, v) = \ell$ and let $s = w_0, \dots, w_{\ell - 1}, w_{\ell} = v$ denote a shortest path.
        By induction, $w_{\ell - 1} \in L(s, \ell - 1) = \hat{L}(s, \ell - 1)$ holds at the iteration where $\level(s)$ is incremented to $\ell - 1$ for the first time.
        At this iteration, the algorithm sets $\calA(s) \gets \hat{L}(s, \ell - 1) = L(s, \ell - 1)$ so that $w_{\ell - 1} \in \calA(s)$.
        
        Since we only increment $\level(s)$ when $\calA(s)$ is empty, there was an iteration of the while loop (possibly the current one) $w_{\ell - 1}$ must have been removed from $\calA(s)$.
        In particular, there was an iteration where $w_{\ell - 1} \in V_{i} \cap \calA(s)$ and therefore $s \in S'$.
        For this iteration, we consider two cases.

        {\bf Case 1: $|N(V_i)| \leq \frac{\phi n \log n}{|S_{\inT}|}$.}
        Since $w_{\ell - 1} \in V_i \cap \calA(s)$ and $(w_{\ell - 1}, v) \in E$, so we set $\hdist(s, v) \leq \ell$ since at this time $\level(s) = \ell - 1$ (i.e. we have inserted $w_{\ell - 1}$ into $\calA(s)$ so $\level(s) \geq \ell - 1$ but we have yet to increment $\level(s)$ to $\ell$).
        
        {\bf Case 2: $|N(V_i)| > \frac{\phi n \log n}{|S_{\inT}|}$.}
        Since $w_{\ell - 1} \in V_i \cap \calA(s)$ and $(w_{\ell - 1}, v) \in E$, we have $X[s, w_{\ell - 1}] = Y[w_{\ell - 1}, v] = 1$ so that $(XY)[s, v] = 1$ and we set $\hdist(s, v) \leq \ell$ (where $\level(s) = \ell - 1$ as observed in Case 1).
        
        By \Cref{lem:dist-s-v-lb}, we have in both cases that $\hdist(s, v) \geq \dist(s, v) = \ell$ so $\hdist(s, v) = \ell$.
        In any future iterations, the estimated distance does not change, as $\level(s) + 1 \geq \ell + 1$ is increasing.

        We now argue $L(s, \ell) = \hat{L}(s, \ell)$.
        In our above argument, in the iteration where we set $\hdist(s, v)$ we also add $v$ to $\hat{L}(s, \ell)$, so $L(s, \ell) \subseteq \hat{L}(s, \ell)$.
        Again, this follows as $\level(s) = \ell - 1$ and we add $v$ to $\hat{L}(s, \level(s) + 1)$.
        
        To argue the reverse inclusion, let $v$ be some vertex where $\hdist(s, v) = \ell$ in some iteration of the while loop where $s \in S'$ and $\level(s) = \ell - 1$.
        By \Cref{lem:dist-s-v-lb}, we have $\dist(s, v) \leq \ell$.
        If $v \not\in L(s, \ell)$, we have $\dist(s, v) < \ell$.
        Then, since $\level(s) = \ell - 1$, the inductive hypothesis states that $\hdist(s, v) = \dist(s, v) \leq \ell - 1$, a contradiction.

        Finally, we again observe that in any future iterations where $\level(s) \geq \ell$, only $\hat{L}(s, \level(s) + 1)$ is modified so $L(s, \ell) = \hat{L}(s, \ell)$ for all future iterations of the while loop, as desired.
    \end{proof}

    Next, we argue that the while loop does not terminate until all distances are computed.

    \begin{lemma}
        \label{lem:while-continues}
        If there exists $\hdist(s, v) > \dist(s, v)$ at the end of an iteration, then the while loop will not terminate.
    \end{lemma}

    \begin{proof}[Proof of \Cref{lem:while-continues}]
        We claim $\calA(s) \neq \emptyset$ for all $s \in S$.
        If $\calA(s) = \emptyset$ at the end of the iteration, then it is empty at the if statement right before the end of the iteration.
        If $\calA(s)$ is still empty after the if statement, $\hat{L}(s, \level(s)) = \emptyset$.
        By \Cref{lem:dist-correct-ub} and $G$ being connected, we have $\ecc(s) < \level(s)$.
        Consider the first iteration where $\level(s) = \ecc(s)$ (which must be a previous iteration of the while loop).
        By \Cref{lem:dist-correct-ub}, during that iteration $s$ is already removed from $S$ as
        \begin{equation*}
            \bigcup_{r = 0}^{\level(s)} \hat{L}(s, r) = \bigcup_{r = 0}^{\level(s)} L(s, r) = \bigcup_{r = 0}^{\ecc(s)} L(s, r) = V \text{.}
        \end{equation*}
        Thus, $s \not\in S$, a contradiction.

        Applying our claim above, we argue that $\sum_{i} \score(i) \geq |S|$ since 
        \begin{equation*}
            \sum_{i \in [k]} \score(i) = \sum_{i \in [k]} |\set{s \in S \given \calA(s) \cap V_{i} \neq \emptyset}| = \sum_{s \in S} |\set{i \in [k] \given \calA(s) \cap V_{i} \neq \emptyset}| \geq |S| \text{.}
        \end{equation*}
        If all $\score(i) \leq \frac{|S||N(V_i)|}{2 \phi n}$, then since $\calC$ is a $(d, \phi)$-\LDND,
        \begin{equation*}
            \sum_{i} \score(i) \leq |S| \sum_{i} \frac{|N(V_i)|}{2 \phi n} \leq \frac{|S|}{2} \text{.}
        \end{equation*}
        Thus, there is some $i$ with $\score(i) > \frac{|S||N(V_i)|}{2 \phi n}$ and the while loop does not terminate.
    \end{proof}

    Thus, if the while loop terminates, then we return $\hdist(s, v) = \dist(s, v)$ for all $s \in S_{\inT}, v \in V$.

    \paragraph{Running Time.}
    We now bound the running time of our algorithm.
    This will also complete the correctness argument as we will show that the while loop terminates.
    Initializing $S$, $\hat{L}(s, 0)$, $\calA(s)$, $\hdist(s, s)$, $\level(s)$, and $\hdist(s, v)$ for all $v$ requires $O(|S_{\inT}|n)$ time.
    Initializing $\score(i)$ requires $O(n)$ time by enumerating all $V_i$.
    Overall, we have used time linear in the output so far.

    It now remains to analyze the running time of the while loop.
    Fix a cluster $V_{i}$.
    We bound the running time of all iterations of the while loop where $i$ is chosen as the cluster with $\score(i) > \frac{|S||N(V_i)|}{2 \phi n}$.
    Let $T(i)$ denote the iterations where $V_i$ is the chosen cluster.
    We now argue that (1) for all $s \in S_{\inT}$, there are few (roughly $d$) iterations $t \in T(i)$ where $V_i \cap \calA(s) \neq \emptyset$, and (2) $|T(i)|$ is not too large (roughly $\frac{n}{|V_i|}$).
    We begin with the first claim.

    \begin{lemma}
        \label{lem:s-iteration-bound}
        For any $s \in S_{\inT}$ and $V_{i}$, $V_i \cap \calA(s) \neq \emptyset$ (equivalently $s \in S'$) for at most $d + 1$ distinct $t \in T(i)$.
    \end{lemma}

    \begin{proof}[Proof of \Cref{lem:s-iteration-bound}]
        Let $V_i \cap \calA(s) \neq \emptyset$ (equivalently $s \in S'$) for $k := k(s, i)$ distinct iterations of the while loop where $V_i$ is the chosen cluster (i.e. distinct $t \in T(i)$).

        First, observe that for all $s \in S_{\inT}$, $\calA(s) \subseteq \hat{L}(s, \level(s))$ at every iteration, since every time $\level(s)$ increases, we set $\calA(s) \gets \hat{L}(s, \level(s))$.
        In later iterations, we only remove vertices from $\calA(s)$ and leave $\hat{L}(s, \level(s))$ untouched until $\calA(s)$ is empty, at which point we increment $\level(s)$ to a larger value and reset $\calA(s)$.
        Furthermore, for each value of $\level(s)$, $s \in S'$ for at most one iteration where $V_i$ is chosen, since whenever $s \in S'$, we remove $V_i$ from $\calA(s)$. 
        Since $\calA(s)$ does not add elements until $\level(s)$ is incremented again, $s \not\in S'$ for further iterations where $V_i$ is chosen.
        Thus, it suffices to bound the number of distinct $\level(s)$ for which $V_i \cap \hat{L}(s ,\level(s)) \neq \emptyset$.
        By \Cref{lem:dist-correct-ub}, we can bound the number of distinct $\level(s)$ for which $V_i \cap L(s ,\level(s)) \neq \emptyset$.

        We conclude by arguing that $V_{i}$ can only intersect few $L(s, r)$.
        Since $\diam(V_i) \leq d$, the triangle inequality implies that for any $u, v \in V_i$, 
        \begin{equation*}
            |\dist(s, u) - \dist(s, v)| \leq \dist(u, v) \leq d \text{.}
        \end{equation*}
        In particular, for all $s \in S$, $V_i \subseteq \bigcup_{r = t}^{t + d} L(s, r)$ for some $t \geq 0$.
        Thus, we have $k \leq d + 1$ as desired.
    \end{proof}

    Next, we bound the size of $T(i)$.

    \begin{lemma}
        \label{lem:round-ub}
        For all $V_{i}$, $|T(i)| = \bigO{\frac{d \phi n \log n}{|N(V_i)|}}$.
    \end{lemma}

    \begin{proof}[Proof of \Cref{lem:round-ub}]
        Observe that $|S|$ is monotonically decreasing from $|S_{\inT}|$.
        Let $T_{k}(i) \subseteq T(i)$ denote the iterations of the while loop where $\frac{|S_{\inT}|}{2^{k - 1}} \geq |S| \geq \frac{|S_{\inT}|}{2^{k}}$.
        Thus, $T(i) = \bigcup_{k = 1}^{\log |S_{\inT}|} T_{k}(i)$ as the while loop terminates once $|S|$ is empty.
        
        We bound $|T_{k}(i)|$ for $k \geq 1$.
        Fix an iteration of the while loop $t \in T_{k}(i)$ where $V_i$ is the chosen cluster.
        Since $|S| \geq \frac{|S_{\inT}|}{2^{k}}$, we have $\score(i) > \frac{|S||N(V_i)|}{2 \phi n} \geq \frac{|S_{\inT}||N(V_i)|}{2^{k + 1} \phi n}$.
        In particular, there are $\frac{|S_{\inT}||N(V_i)|}{2^{k + 1} \phi n}$ distinct sources in $S$ where $\calA(s) \cap V_{i} \neq \emptyset$: i.e. $S'(t)$ contains at least $\frac{|S_{\inT}||N(V_i)|}{2^{k + 1} \phi n}$ vertices, where $S'(t)$ denotes the set $S'$ chosen in the $t$-th iteration of the while loop.
        From \Cref{lem:s-iteration-bound}, we know that each source $s \in S_{\inT}$ is in $S'$ for at most $d + 1$ distinct iterations in $T_{k}(i)$.
        Combining the above observations,
        \begin{equation*}
            |S|(d + 1) \geq \sum_{s \in S} \left| \set{t \in T_{k}(i) \given s \in S'(t)} \right| = \sum_{t \in T_{k}(i)} \left| \set{s \in S'(t)} \right| \geq |T_{k}(i)| \frac{|S_{\inT}||N(V_i)|}{2^{k + 1} \phi n} \text{.}
        \end{equation*}
        Rearranging and applying the upper bound $|S| \leq \frac{|S_{\inT}|}{2^{k - 1}}$,
        \begin{equation*}
            |T_{k}(i)| \leq \frac{2^{k + 1} \phi n}{|N(V_i)||S_{\inT}|} \frac{|S_{\inT}|(d + 1)}{2^{k - 1}} \leq \frac{8 d \phi n}{|N(V_i)|} \text{.}
        \end{equation*}
        Summing over $k$ and using $|S_{\inT}| \leq n$ we get $|T(i)| = \bigO{\frac{d \phi n \log n}{|N(V_i)|}}$.

    \end{proof}

    Finally, we are ready to bound the total running time of all iterations where $V_i$ is the chosen cluster.
    We break into two cases.
    Recall that $S'(t)$ denotes the set $S'$ chosen in the $t$-th iteration of the while loop.

    \paragraph{Case 1: $|N(V_i)| \leq \frac{\phi n \log n}{|S_{\inT}|}$.}
    An iteration of the while loop requires time $\bigO{|S'(t)||V_i||N(V_i)|}$ (since the running time is dominated by the time taken to relax all edges from $\calA(s) \cap V_{i} \subseteq V_i$ to $N(\calA(s) \cap V_i) \subseteq N(V_i)$).
    Over all iterations in $T(i)$, the total time is
    \begin{equation*}
        \sum_{t \in T(i)} \bigO{|S'(t)||V_i||N(V_i)|} = \bigO{|V_i||N(V_i)||S_{\inT}|d} = \bigO{|S_{\inT}||V_i| \frac{\phi n d \log n}{|S_{\inT}|}} = \bigO{|V_i| \phi n d \log n} \text{.}
    \end{equation*}
    In the first equality, we use $\sum_{t} |S'(t)| = \bigO{|S_{\inT}|d}$ by \Cref{lem:s-iteration-bound}.
    In the second equality, we apply our upper bound on $|N(V_i)|$.

    \paragraph{Case 2: $|N(V_i)| > \frac{\phi n \log n}{|S_{\inT}|}$.}
    Similarly, the running time of each while loop is dominated by the boolean matrix product between $S' \times V_i$ matrix $X$ and $V_i \times N(V_i)$ matrix $Y$.
    Summing over all iterations, we bound the total time as
    \begin{equation*}
        \bigO{\sum_{t \in T(i)} \TBMM(S'(t), V_i, N(V_i))} = \bigO{\frac{d \phi n \log n}{|N(V_i)|} \TBMM \left(\frac{|S_{\inT}||N(V_i)|}{\phi n \log n}, N(V_i), N(V_i)\right)} 
    \end{equation*}
    using \Cref{lem:s-iteration-bound}, \Cref{lem:round-ub}, the fact that $V_i \subseteq N(V_i)$, and the following lemma.

    \begin{restatable}{lemma}{MatMulConvex}
        \label{lem:mat-mul-convex}
        Let $a_1, \ldots, a_p$ be positive integers with $a_1 + a_2 + \dots + a_{p} = a$, and $q$ be a positive integer. 
        Then,
        \begin{equation*}
            \sum_{j = 1}^{p} \TBMM(a_j, q, q) = \bigO{p \TBMM(a/p, q, q)} \text{.}
        \end{equation*}
        The result holds also for Min-Plus matrix multiplication (i.e. $\TMPP$).
    \end{restatable}
    
    \begin{proof}[Proof of \Cref{lem:mat-mul-convex}]
        For all $a_{j}$, we need to compute the product between an $a_{j} \times q$ matrix $X$ and a $q \times q$ matrix $Y$.
        To do so, we construct $\lceil \frac{a_{j}}{a/p} \rceil$ instances of matrix multiplication between $\lceil \frac{a}{p} \rceil \times q$ and $q \times q$ matrices, where each product computes $\lceil a/p \rceil$ rows of the $a_{j} \times q \times q$ matrix product.
        Formally, we split $X$ by its rows into $\lceil \frac{a_{j}}{a/p} \rceil$ blocks of size $\lceil \frac{a}{p} \rceil \times q$ and multiply each block by $Y$.
        Overall, we need 
        \begin{equation*}
            \bigO{\sum_{j} \left\lceil \frac{a_{j}}{a/p} \right\rceil } = \bigO{p}
        \end{equation*}
        such matrix products.
    \end{proof}

    Finally, we are ready to bound the running time over all clusters.
    First, for clusters with $|N(V_i)| \leq \frac{\phi n \log n}{|S_{\inT}|}$, the total time spent over all iterations of the while loop is then $\bigO{\sum_{i} |V_{i}| \phi n d \log n} = \bigO{n^2 \phi d \log n}$.
    For clusters with $|N(V_i)| > \frac{\phi n \log n}{|S_{\inT}|}$, we observe that for any $k > 0$, there are at most $2^{k}$ clusters with $|N(V_i)| \geq \frac{\phi n}{2^{k}}$.
    Summing over $O(\log n)$ values of $k$, we obtain total running time:
    \begin{equation*}
        \bigO{d \phi^2 n^2 (\log n)^2 \max_{T> \frac{\phi n \log n}{|S_{\inT}|}} \frac{\TBMM \left( \frac{|S_{\inT}|T }{\phi n \log n}, T, T\right)}{T^2}} = \bigO{d \phi^2 n^2 (\log n)^2 \max_{T> \frac{n}{|S_{\inT}|}} \frac{\TBMM \left( \frac{|S_{\inT}|T }{n}, T, T\right)}{T^2}} \text{.}
    \end{equation*}

    We use the following claim to upper bound the final term.
    
    \begin{claim}
        \label{clm:mm-quadratic-growth}
        Let $0 < z < Z$ be constants independent of $T$.
        Then,
        \begin{equation*}
            \max_{z < T \leq Z, T = 2^{k}} \frac{\TBMM \left( T/z , T, T\right)}{T^2} = Z^{\omega(p(Z), 1, 1) - 2 + o(1)} 
        \end{equation*}
        where $p := p(Z)$ is defined so that $T^{p} = T/z$.
        
        Furthermore, if $\TBMM$ satisfies
        \begin{equation*}
            \TBMM(2m, 2n, 2n) \geq 4 \TBMM(m, n, n)
        \end{equation*}
        for all $m, n \geq 1$, then
        \begin{equation*}
            \max_{z < T \leq Z} \frac{\TBMM \left( T/z , T, T\right)}{T^2} = \bigO{\frac{\TBMM \left( Z/z , Z, Z\right)}{Z^2}}
        \end{equation*}
        The claim holds also if we replace $\TBMM$ with $\TMPP$.
    \end{claim}

    \begin{proof}
        For any $T \geq z$, fix $p := p(T) > 0$ so that $T/z = T^{p}$.
        Then, we have
        \begin{equation*}
            p(T) = \frac{\log(T/z)}{\log(T)} = \frac{\log T - \log z}{\log T}
        \end{equation*}
        is increasing for $T \geq 1$.
        In particular, we can upper bound
        \begin{equation*}
            \frac{\TBMM \left( T/z , T, T\right)}{T^2} \leq \frac{T^{\omega(p(T), 1, 1) + o(1)}}{T^2} = T^{\omega(p(T), 1, 1) - 2 + o(1)} \text{.}
        \end{equation*}
        Since $\omega(p(T), 1, 1)$ is increasing in $T$, we have
        \begin{equation*}
             \max_{z < T \leq Z} \frac{\TBMM \left( T/z , T, T\right)}{T^2} = \max_{z < T \leq Z}  T^{\omega(p(T), 1, 1) - 2 + o(1)} = Z^{\omega(p(Z), 1, 1) - 2 + o(1)} 
        \end{equation*}
        is maximized at $T = Z$.
        
       We now prove the second statement, which is necessary for establishing equivalence to \BMM.
       It suffices to consider $z \leq T \leq n$ where $T$ is a power of $2$, as changing $T$ by a constant factor changes $\TBMM(T/z, T, T)/T^2$ by a constant factor.
       It follows from our assumption that
        \begin{equation*}
            \frac{\TBMM \left( 2T/z , 2T, 2T\right)}{(2T)^2} \geq \frac{\TBMM(T/z, T, T)}{T^2}
        \end{equation*}
        for all $T \geq z$ so that the expression is maximized at $T = Z$.
    \end{proof}

    To conclude the running time analysis, we apply \Cref{clm:mm-quadratic-growth} with $z = n/|S_{\inT}|$ and $Z = n$ to obtain $p := p(Z) = p(n)$ satisfying $n^{p} = |S_{\inT}|$.
    In particular, we can bound the final running time as
    \begin{equation*}
        \bigO{d \phi^2 (\log n)^2 n^{\omega(p, 1, 1) + o(1)}} \text{.}
    \end{equation*}
    
    To establish equivalence with \BMM, we apply the second statement of \Cref{clm:mm-quadratic-growth}, and bound the final running time as
    \begin{equation*}
        \bigO{d \phi^2 (\log n)^2 \TBMM \left( |S_{\inT}|, n, n \right)} \text{.} 
    \end{equation*}
\end{proof}

\subsection{Extension to Bounded Weights}

We now extend our algorithm to handle undirected graphs with small integer weights.

\MSSPBounded*

Note that it is without loss of generality that we assume all weights are positive: a negative edge in undirected graphs creates a negative cycle, and all pairwise distances become $-\infty$ (recall we assume the graph is connected). Also, any vertices joined by a weight $0$ edge can be merged and treated as one vertex.
Assume without loss of generality that $B$ is a power of $2$.
We will require a simple generalization of bounded Min-Plus products.

\begin{lemma}
    \label{lem:non-0-bounded-min-plus}
    Let $X$ be an $n_1 \times n_2$ integer matrix with all finite entries bounded between $[B_0, B_0 + B]$ and $Y$ be an $n_2 \times n_3$ integer matrix with entries bounded between $[0, B]$.
    Then, $X \star Y$ can be computed in $\bigO{\TMPP(n_1, n_2, n_3 \mid B)} = \bigO{B \cdot \TMUL(n_1, n_2, n_3)}$ time.
\end{lemma}

\begin{proof}[Proof of \Cref{lem:non-0-bounded-min-plus}]
    Consider the matrix $X' := X - \min_{i, j} X[i, j]$ with integer entries between $[0, B]$, which can be computed in $O(n_1 n_2)$ time.
    Then, $X' \star Y$ can be computed in $\bigO{B \cdot \TMUL(n_1, n_2, n_3)}$ time \cite{alon1997exponent}. 
    To recover $X \star Y$, observe that $X \star Y = (X' \star Y) + \min_{i, j} X[i, j]$, which can be computed in $O(n_1 n_3)$ time.
\end{proof}

We are now ready to prove \Cref{thm:bounded-weight-mssp}.
Consider the following procedure which for presentation we break into an initialization phase and a computation phase.
In the initialization phase, we decompose the graph and initialize various useful data structures.

\begin{mdframed}
    {\bf Input: } Graph $G = (V, E, \wt)$, weight bound $B$, and sources $S_{\inT} \subseteq V$.
    \begin{enumerate}
        \item For $t = 0 \ldots \log B$ and $b = 2^{t}$:
        \begin{enumerate}
            \item Let $G_{b} = (V, E_{b})$ be the graph with $E_{b} = \set{e \in E \given b \leq \wt(e) < 2b}$.

            \item Let $d, \phi \gets 2^{O(\sqrt{\log n})}$.
            
            \item Let $\calC_{b}$ be a $(d, \phi)$-\LDND~of $G_{b}$ as an unweighted graph.
            Denote $\calC_{b} = \set{V_{1}^{(b)}, \dots, V_{k_{b}}^{(b)}}$.
        \end{enumerate}
        
        \item Let $S \gets S_{\inT}$.
        
        \item For all $s \in S$:
        \begin{enumerate}
            \item Let $\hat{L}(s, 0) \gets \set{s}$, $\hat{L}(s, r) \gets \emptyset$ for $1 \leq r \leq B n$, and $\hdist(s, s) \gets 0$, $\hdist(s, v) \gets \infty$ for all $v \neq s$.
            \item Let $\calA(s, b) \gets \set{s}$ and $\level(s, b) \gets 0$ for all $s \in S$ and $b = 2^{t}$. Let $\level(s) \gets 0$.
        \end{enumerate}

        \item Let $\score(i, b) \gets \left|\set{s \in S \given \level(s) = \level(s, b) \andT \calA(s, b) \cap V_{i}^{(b)} \neq \emptyset}\right|$ for all $V_{i}^{(b)}$.
    \end{enumerate}
{\bf Output: } $\calC_{b}$, $\calA(s, b)$, $\level(s, b)$, $\level(s)$, $\hat{L}(s, r)$, $\hdist(s, v)$, and $\score(i, b)$ for all $s, b, r, i$.
\end{mdframed}

In the computation phase, we use the computed decompositions and data structures to compute shortest paths.
\begin{mdframed}
{\bf Input: } Graph $G$, weight bound $B$, sources $S_{\inT}$. Outputs of initialization phase.
    \begin{enumerate}
        \item While $S \neq \emptyset$:
        \begin{enumerate}
            \item If exists $s \in S$ with $\calA(s, b) = \emptyset$ for all $b = 2^{t}$:
            \begin{enumerate}
                \item Run $\increaseLevel\left(s, \set{\level(s, b)}, \level(s), \set{\hat{L}(s, r)}, \set{\calC_{b}}\right)$.
            \end{enumerate}
            \item Otherwise, there exists $V_{i}^{(b)}$ with $\score(i, b) > \frac{|S|\left|N\left(V_{i}^{(b)}\right)\right|}{2 \phi n \log B}$:
            \begin{enumerate}
                \item Run $\extendPaths\left(G, S, V_{i}^{(b)}, \set{\calA(s, b)}, \set{\level(s, b)}, \set{\level(s)}, \set{\hat{L}(s, r)}\right)$.
            \end{enumerate}
        \end{enumerate}
    \end{enumerate}
{\bf Output: } $\hdist(s, v)$ for all $s \in S_{\inT}$ and $v \in V$.
\end{mdframed}

\begin{mdframed}
    {\bf $\extendPaths\left(G, S, V_{i}^{(b)}, \set{\calA(s, b)}, \set{\level(s, b)}, \set{\level(s)}, \set{\hat{L}(s, r)}\right)$:}
    \begin{enumerate}
        \item Let $S' \gets \set{s \in S \given \calA(s, b) \cap V_{i}^{(b)} \ne \emptyset}$.
        
        \item If $|N(V_i^{(b)})| \leq \frac{n}{|S_{\inT}|}$: 
        \begin{enumerate}
            \item For all $s \in S'$, $u \in \calA(s, b) \cap V_i^{(b)}$, $v \in N_{G_{b}}(u)$: update $\hdist(s, v) \minget \hdist(s, u) + \wt(u, v)$.
        \end{enumerate}
        \item If $|N(V_i^{(b)})| > \frac{n}{|S_{\inT}|}$:
        \begin{enumerate}
            \item Construct $S' \times V_{i}^{(b)}$ matrix $X$ and $V_i^{(b)} \times N_{G_{b}}(V_{i}^{(b)})$ matrix $Y$ with entries
            \begin{equation*}
                X[s, u] = \begin{cases}
                    \hdist(s, u) & u \in \calA(s, b) \\
                    \infty & \otherwise
                \end{cases}
                ,\quad Y[u, v] = \begin{cases}
                    \wt(u, v) & (u, v) \in E_{b} \\
                    \infty & \otherwise
                \end{cases} \text{.}
            \end{equation*}
            \item Compute $X \star Y$ using \Cref{lem:non-0-bounded-min-plus}.
            \item For all $s \in S'$ and $v \in N_{G_{b}} (V_i^{(b)})$, update $\hdist(s, v) \minget (X\star Y)[s, v]$.
        \end{enumerate}
        
        \item For all $s \in S'$:
        \begin{enumerate}
            \item Update $\hat{L}(s, \level(s) + 1) \cupget \set{v \given \hdist(s, v) = \level(s) + 1}$.

            \item Remove $V_{i}^{(b)}$ from $\calA(s, b)$ and decrement $\score(i, b) \subget 1$.

            \item If $\calA(s, b) = \emptyset$: run $\updateScores(s, b)$.
        \end{enumerate}
    \end{enumerate}
\end{mdframed}

\begin{mdframed}
    {\bf $\updateScores(s, b)$:}
    \begin{enumerate}
        \item Update $\level(s, b) \addget b$ and $\level(s) \gets \min_{b} \level(s, b)$.

        \item If $\bigcup_{r = 0}^{\level(s)} \hat{L}(s, r) = V$: then remove $s$ from $S$ and return.

        \item If $\level(s)$ increased, then for all $b$ with $\level(s, b) = \level(s)$:

        \begin{enumerate}
            \item Update $\calA(s, b) \gets \bigcup_{r = \level(s) - b + 1}^{\level(s)} \hat{L}(s, r)$.
            \item For all $V_i^{(b)} \cap \calA(s, b) \neq \emptyset$, increment $\score(i, b)$.
        \end{enumerate} 
    \end{enumerate}
\end{mdframed}

\begin{mdframed}
    {\bf $\increaseLevel\left(s, \set{\level(s, b)}, \level(s), \set{\hat{L}(s, r)}, \set{\calC_{b}}\right)$:}
    \begin{enumerate}
        \item Update $\hat{L}(s, \level(s) + 1) \cupget \set{v \given \hdist(s, v) = \level(s) + 1}$.
        \item Increment $\level(s, b) \addget b$ for all $b$ with $\level(s, b) = \level(s)$ and update $\level(s) \gets \min_{b} \level(s, b)$.

        \item For all $b$ with $\level(s, b) = \level(s)$:

        \begin{enumerate}
            \item Update $\calA(s, b) \gets \bigcup_{r = \level(s) - b + 1}^{\level(s)} \hat{L}(s, r)$.
            \item For all $V_i^{(b)} \cap \calA(s, b) \neq \emptyset$, increment $\score(i, b)$.
        \end{enumerate} 
    \end{enumerate}
\end{mdframed}

\begin{proof}[Proof of \Cref{thm:bounded-weight-mssp}]
    We prove the correctness and running time of our algorithm.
    
    \paragraph{Correctness.}
    We argue that our algorithm terminates and outputs correct distances.
    As before, we can show that the estimated distances are always valid.
    
    \begin{lemma}
        \label{lem:dist-s-v-lb-wt}
        In every iteration of the while loop $\hdist(s, v) \geq \dist(s, v)$ for all $s \in S_{\inT}, v \in V$.
    \end{lemma}
    
    \begin{proof}[Proof of \Cref{lem:dist-s-v-lb-wt}]
        We argue that at the end of each iteration of the while loop $\hdist(s, v)$ can be realized by a path in $G$ from $s$ to $v$ for all $s \in S_{\inT}$, $v \in V$.
        Consider the first iteration where this is not the case.
        Then, during this iteration, $\hdist(s, v)$ was updated, i.e. we update $\hdist(s, v) \gets \hdist(s, u) + \wt(u, v) < \dist(s, v)$ for some $u \in \calA(s, b) \cap V_{i}^{(b)}$ with $\hdist(s, v) = \hdist(s, u) + \wt(u, v)$.
        Note that $\hdist(s, u)$ was computed in some previous iteration.
        Thus, there is a path from $s$ to $u$ with length $\hdist(s, u)$, which we may extend to $v$ via the edge $(u, v)$ with weight $\wt(u, v)$. This is a contradiction. 
    \end{proof}

    The remainder of the proof is dedicated to showing that our algorithm produces exact shortest distances.
    We begin with some useful claims about $\level(s, b)$, $\level(s)$, whose proofs are deferred.

    \begin{restatable}{claim}{LevelIncBound}
        \label{clm:level-inc-bound}
        In any iteration of the while loop, $\level(s)$ increases by at most one for all $s \in S_{\inT}$.
    \end{restatable}

    \begin{restatable}{claim}{MultipleBExist}
        \label{clm:multiple-b-exists}
        Let $s \in S_{\inT}$ and $\ell \leq \ecc(s)$.
        Then, for every $b$ such that $\ell$ is a multiple of $b$, there exists an iteration where $\level(s) = \level(s, b) = \ell$ at the beginning of the iteration.
    \end{restatable}

    Next, we prove that our algorithm iteratively computes correct distances, analogously to \Cref{lem:dist-correct-ub}.

    \begin{lemma}
        \label{lem:dist-s-v-ub-wt}
        At the end of each iteration of the while loop, $\hdist(s, v) = \dist(s, v)$ for all $\dist(s, v) \leq \level(s)$.
        In particular, for all $r \leq \level(s)$, $L(s, r) = \hat{L}(s, r)$.
    \end{lemma}

    \begin{proof}
        We proceed by induction, where the base case follows identically as \Cref{lem:dist-correct-ub}.
        Now, let $\ell > 0$.
        Consider the first iteration of the while loop where $\level(s) = \ell$.
        In this iteration, we incremented $\level(s)$.
        By \Cref{clm:level-inc-bound}, $\level(s)$ increments by at most $1$, so that at the end of the previous iteration, we have $\level(s) = \ell - 1$.
        
        By induction, we have $\hdist(s, u) = \dist(s, u)$ for all $\dist(s, u) < \ell$.
        Consider a vertex with $\dist(s, v) = \ell$ and a shortest path $s = w_{0}, \dots, w_{p - 1}, w_{p} = v$ where $\dist(s, w_{p - 1}) < \ell$.
        Suppose $(w_{p - 1}, v) \in E_{b}$ and $w_{p - 1} \in V_{i}^{(b)}$.
        By induction, $\hdist(s, w_{p - 1}) = \dist(s, w_{p - 1})$ and $w_{p - 1} \in \hat{L}(s, r)$ for $r = \dist(s, w_{p - 1}) \geq \ell - 2b + 1$ since edge weights in $E_{b}$ are less than $2b$.
        Let $\ell_0$ be the multiple of $b$ such that $\ell_0 - b < \dist(s, w_{p - 1}) \leq \ell_0$.
        Note that we have $\ell_0 < \ell \leq \ecc(s)$ since $\wt(u, v) \geq b$.

        Consider the first iteration where $\level(s, b) = \level(s) = \ell_0$, which exists by \Cref{clm:multiple-b-exists}.
        In this iteration, we set $\calA(s, b) \gets \bigcup_{r = \ell_0 - b + 1}^{\ell_0} \hat{L}(s, r) = \bigcup_{r = \ell_0 - b + 1}^{\ell_0} L(s, r)$ so that $w_{p - 1} \in \calA(s, b)$.
        Now, in the current iteration, we have $\level(s, b) \geq \level(s) = \ell$, which implies that $\updateScores$ or $\increaseLevel$ was triggered i.e. $\calA(s, b) = \emptyset$.
        Thus, there was an intermediate iteration where $w_{p - 1}$ was removed from $\calA(s, b)$ i.e. $V_{i}^{(b)}$ was chosen.
        We argue that at the start of this intermediate iteration, $\hdist(s, w_{p - 1}) = \dist(s, w_{p - 1})$.
        This follows via induction, as above we showed $\level(s) \geq \ell_0 \geq \dist(s, w_{p - 1})$ (in fact this was already true in the iteration where $\calA(s, b)$ was set).
        
        We now consider two cases.

        {\bf Case 1: $|N(V_i^{(b)})| \leq \frac{n}{|S_{\inT}|}$.}
        Since $w_{p - 1} \in V_i^{(b)} \cap \calA(s, b)$ and $(w_{p - 1}, v) \in E_{b}$, so we set $\hdist(s, v) \leq \hdist(s, w_{p - 1}) + \wt(w_{p - 1}, v) \leq \dist(s, w_{p - 1}) + \wt(w_{p - 1}, v) = \dist(s, v)$.

        {\bf Case 2: $|N(V_i^{(b)})| > \frac{n}{|S_{\inT}|}$.}
        Since $w_{p - 1} \in V_i^{(b)} \cap \calA(s, b)$ and $(w_{p - 1}, v) \in E_{b}$, we have $X[s, w_{p - 1}] = \dist(s, u)$, $Y[w_{p - 1}, v] = \wt(u, v)$ so that $(X \star Y)[s, v] \leq \dist(s, u) + \wt(u, v)$ and we set $\hdist(s, v) \leq \dist(s, u) + \wt(u, v) = \dist(s, v)$.

        Thus, we have $\hdist(s, v) = \dist(s, v)$ in this intermediate iteration.
        Note that we already have $\hdist(s, v) = \dist(s, v)$ in an iteration where $\level(s) = \level(s, b) = \ell_0$.
        The distance equality holds also in any subsequent iteration since $\hdist$ does not increase, and does not decrease below $\dist(s, v)$ \Cref{lem:dist-s-v-lb-wt}.

        Finally, we argue that $\hat{L}(s, \ell) = L(s, \ell)$.
        Suppose $v \in \hat{L}(s, \ell)$, so that $\hdist(s, v) \leq \ell$.
        We can in fact conclude $\hdist(s, v) = \ell$ since if $\hdist(s, v) = \dist(s, v) < \ell$ then $v \notin \hat{L}(s, \ell)$ by induction.
        
        Conversely, if $v \in L(s, \ell)$, consider the iteration where $\level(s) = \ell$ for the first time.
        We increment $\level(s)$ either in $\updateScores$ or $\increaseLevel$.
        Recall from our previous discussion that we have $\hdist(s, v) = \dist(s, v)$ in some iteration where $\level(s) = \ell_0 < \ell$, where $\ell_0$ is the multiple of $b$ satisfying $\ell_0 - b < \dist(s, w_{p - 1}) \leq \ell_0$ above.
        We consider two cases.

        If $\level(s) \gets \ell$ is triggered through $\updateScores$, then before the last if statement in $\extendPaths$, we have $\level(s) = \ell - 1 \geq \ell_{0}$ and $v$ is added to $\hat{L}(s, \ell)$.
        On the other hand, if $\level(s) \gets \ell$ is triggered through $\increaseLevel$, at the first step of $\increaseLevel$, we have $\level(s) = \ell - 1 \geq \ell_{0}$ so $v$ is added to $\hat{L}(s, \ell)$.
        
        In any future iterations, the equality holds as we only update $\hat{L}(s, \level(s) + 1)$ when $\level(s)$ increases.
    \end{proof}

    Thus, we have established that $\hdist(s, v) = \dist(s, v)$ for all $\dist(s, v) \leq \level(s)$ at the end of each iteration of the while loop.
    Again, we must argue that the while loop does not terminate until all distances are computed.

    \begin{lemma}
        \label{lem:while-continues-wt}
        If there exists $\hdist(s, v) > \dist(s, v)$ at the end of an iteration, then the while loop will not terminate.
    \end{lemma}

    \begin{proof}
        Suppose the while loop terminates while there is some wrong distance: i.e. $v$ with $\hdist(s, v) > \dist(s, v)$ at the end of the iteration.
        We claim $s \in S$, ensuring the while loop does not terminate.
        It suffices to argue $v \not\in \bigcup_{r = 0}^{\level(s)} \hat{L}(s, r)$.
        By \Cref{lem:dist-s-v-ub-wt}, we have $\hdist(s, v) > \dist(s, v) > \level(s)$, implying that $v \not\in \bigcup_{r \leq \level(s)} \hat{L}(s, r)$ since $\hat{L}(s, r)$ only contains vertices $w$ with $\hdist(s, w) \leq r$, as desired.
    \end{proof}

    Finally, we argue that either the condition to run $\increaseLevel$ or the condition to run $\extendPaths$ is satisfied in each iteration.

    \begin{lemma}
        \label{lem:score-lb}
        If all $s$ have $\calA(s, b) \neq \emptyset$ for some $b$, there exists $V_{i}^{(b)}$ with $\score(i, b) > \frac{|S|\left|N\left(V_{i}^{(b)}\right)\right|}{2 \phi n \log B}$.
    \end{lemma}

    \begin{proof}
        Following a similar argument to \Cref{lem:while-continues}, we bound the sum of scores from below.
        \begin{equation*}
            \sum_{i, b} \score(i, b) = \sum_{i, b} \left| \set{s \in S \given \calA(s, b) \cap V_{i}^{(b)} \neq \emptyset} \right| = \sum_{s} \left| \set{V_{i}^{(b)} \given \calA(s, b) \cap V_{i}^{(b)} \neq \emptyset} \right| \geq |S|
        \end{equation*}
        If all $\score(i, b) \leq \frac{|S|\left|N\left(V_{i}^{(b)}\right)\right|}{2 \phi n \log B}$, then we obtain a contradiction via
        \begin{equation*}
            \sum_{i, b} \score(i, b) \leq \sum_{b} \sum_{i} \frac{|S|\left|N\left(V_{i}^{(b)}\right)\right|}{2 \phi n \log B} \leq (1 + \log B) \frac{|S|}{2 \log B} < |S| \text{.}
        \end{equation*}
    \end{proof}

    We have thus established that our algorithm (if it terminates) outputs exact distances between $S_{\inT}$ and $V$.
    As before, it remains to show that the algorithm terminates, which we prove below.

    \paragraph{Running Time.}
    We now bound the running time, thus proving that our algorithm terminates and produces the correct output.
    First, the initialization phase takes $\bigO{n^{2 + o(1)} \log B + |S_{\inT}|Bn} = B n^{2 + o(1)}$ time.

    We separately bound the time required by the while loop of the computation phase.
    First, we bound the total time required by all iterations where $\increaseLevel$ is invoked.
    We defer the proof to the end of this section.

    \begin{restatable}{lemma}{IncreaseLevelTimeBound}
        \label{lem:increase-level-time-bound}
        The total time used by all calls to $\increaseLevel$ is $\bigtO{|S_{\inT}|nB} = \bigtO{Bn^{2}}$.
    \end{restatable}

    Next, we bound the total time required by all iterations where $\extendPaths$ is invoked.
    As in \Cref{thm:mssp-reduction}, we bound the time required by all iterations where a fixed cluster $V_{i}^{(b)}$ is chosen.
    Let $T(i, b)$ denote the iterations where $V_{i}^{(b)}$ is the chosen cluster.
    We again bound the number of iterations where $s \in S'$ in $T(i, b)$ and the size of $T(i, b)$, analogously to \Cref{lem:s-iteration-bound} and \Cref{lem:round-ub}.

    \begin{lemma}
        \label{lem:s-iteration-bound-wt}
        For any $s \in S_{\inT}$ and $V_{i}^{(b)}$, $s \in S'$ for at most $3d$ distinct $t \in T(i, b)$.
    \end{lemma}

    \begin{proof}[Proof of \Cref{lem:s-iteration-bound-wt}]
        Observe $s \in S'$ is only possible for an iteration $t \in T(i, b)$ when $\level(s) = \level(s, b)$.
        Fix some $\ell = \level(s) = \level(s, b)$.
        Observe that $s \in S'$ holds in an iteration where $t \in T(i, b)$ and $\level(s) = \level(s, b) = \ell$ if and only if 
        \begin{equation}
            \label{eq:layers-intersect-cluster}
            \emptyset \neq \calA(s, b) \cap V_{i}^{(b)} \subseteq \bigcup_{r = \ell - b + 1}^{\ell} L(s, r) \cap V_{i}^{(b)} \neq \emptyset \text{.}
        \end{equation}
        Furthermore, this holds in at most one iteration $t \in T(i, b)$ where $\ell = \level(s) = \level(s, b)$.
        This follows as we remove $V_{i}^{(b)}$ from $\calA(s, b)$ after the first iteration that this occurs, thus guaranteeing that $\calA(s, b)$ is disjoint from $V_{i}^{(b)}$ until $\level(s, b)$ is increased.
        
        Thus, as in \Cref{lem:s-iteration-bound}, it suffices to bound the number of $\ell$ where \Cref{eq:layers-intersect-cluster} holds.
        Since $V_{i}^{(b)}$ has hop diameter $d$, $V_{i}^{(b)}$ has diameter $2bd$.
        Then,
        \begin{equation*}
            V_{i}^{(b)} \subseteq \bigcup_{r = \ell_0}^{\ell_0 + 2bd} L(s, r)
        \end{equation*}
        for some $\ell_0$.
        We conclude by bounding the number of distinct values of $\ell$ by $\lceil \frac{2bd + 1}{b} \rceil \leq 3d$.
    \end{proof}

    Next, we may bound the size of $T(i, b)$ following identical arguments as \Cref{lem:round-ub}.

    \begin{lemma}
        \label{lem:round-ub-wt}
        For all $V_{i}^{(b)}$, $T(i, b) = \bigtO{\frac{d \phi n \log B \log n}{\left|N\left(V_i^{(b)}\right)\right|}}$.
    \end{lemma}

    \begin{proof}[Proof of \Cref{lem:round-ub-wt}]
        We repeat the argument for completeness.
        As before, let $T_{k}(i, b) \subseteq T(i, b)$ denote the iterations of the while loop where $\frac{|S_{\inT}|}{2^{k - 1}} \geq |S| \geq \frac{|S_{\inT}|}{2^{k}}$.
        Thus, $T(i, b) = \bigcup_{k = 1}^{\log |S_{\inT}|} T_{k}(i, b)$ as the while loop terminates once $|S|$ is empty.
        
        We bound $|T_{k}(i, b)|$ for $k \geq 1$.
        Fix an iteration of the while loop $t \in T_{k}(i, b)$ where $V_i^{(b)}$ is the chosen cluster.
        Since $|S| \geq \frac{|S_{\inT}|}{2^{k}}$, we have $\score(i, b) > \frac{|S|\left|N\left(V_i^{(b)}\right)\right|}{2 \phi n \log B} \geq \frac{|S_{\inT}|\left|N\left(V_i^{(b)}\right)\right|}{2^{k + 1} \phi n \log B}$.
        In particular, there are $\frac{|S_{\inT}|\left|N\left(V_i^{(b)}\right)\right|}{2^{k + 1} \phi n \log B}$ distinct sources in $S$ where $\calA(s) \cap V_{i}^{(b)} \neq \emptyset$: i.e. $S'(t)$ contains at least $\frac{|S_{\inT}|\left|N\left(V_i^{(b)}\right)\right|}{2^{k + 1} \phi n \log B}$ vertices, where $S'(t)$ denotes the set $S'$ chosen in the $t$-th iteration of the while loop.
        From \Cref{lem:s-iteration-bound-wt}, we know that each source $s \in S_{\inT}$ is in $S'$ for at most $3d$ distinct iterations in $T_{k}(i, b)$.
        Combining the above observations,
        \begin{equation*}
            3|S|d \geq \sum_{s \in S} \left| \set{t \in T_{k}(i, b) \given s \in S'(t)} \right| = \sum_{t \in T_{k}(i, b)} \left| \set{s \in S'(t)} \right| \geq |T_{k}(i, b)| \frac{|S_{\inT}|\left|N\left(V_i^{(b)}\right)\right|}{2^{k + 1} \phi n \log B} \text{.}
        \end{equation*}
        Rearranging and applying the upper bound $|S| \leq \frac{|S_{\inT}|}{2^{k - 1}}$,
        \begin{equation*}
            |T_{k}(i, b)| \leq \frac{2^{k + 1} \phi n \log B}{\left|N\left(V_i^{(b)}\right)\right||S_{\inT}|} \frac{3 |S_{\inT}| d}{2^{k - 1}} \leq \frac{12 d \phi n \log B}{\left|N\left(V_i^{(b)}\right)\right|} \text{.}
        \end{equation*}
        Summing over $k$ and using $|S_{\inT}| \leq n$ we get $|T(i, b)| = \bigO{\frac{d \phi n \log B \log n}{\left|N\left(V_i^{(b)}\right)\right|}}$.
    \end{proof}

    Finally, we bound the total running time of all iterations where $V_{i}^{(b)}$ is chosen using similar arguments as \Cref{thm:mssp-reduction}. 
    We break into two cases.

    {\bf Case 1:} $\left|N\left(V_i^{(b)}\right)\right| \leq \frac{n}{|S_{\inT}|}$.
    An iteration of the while loop requires time $\bigO{|S'(t)|\left|V_i^{(b)}\right|\left|N\left(V_i^{(b)}\right)\right|}$ (since the running time is dominated by the time taken to relax all edges from $\calA(s) \cap V_{i}^{(b)} \subseteq V_i^{(b)}$ to $N\left(\calA(s) \cap V_i^{(b)} \right) \subseteq N\left(V_i^{(b)}\right)$).
    Over all iterations in $T(i, b)$, the total time is
    \begin{align*}
        \sum_{t \in T(i, b)} \bigO{|S'(t)|\left|V_i^{(b)}\right|\left|N\left(V_i^{(b)}\right)\right|} &= \bigO{\left|V_i^{(b)}\right|\left|N\left(V_i^{(b)}\right)\right||S_{\inT}|d} \\
        &= \bigO{|S_{\inT}|\left|V_i^{(b)}\right|\frac{d n}{|S_{\inT}|}} \\
        &= \bigO{\left|V_i^{(b)} \right| d n} \text{.}
    \end{align*}
    In the first equality, we use $\sum_{t} |S'(t)| = \bigO{|S_{\inT}|d}$ by \Cref{lem:s-iteration-bound-wt}.
    In the second equality, we apply our upper bound on $|N(V_i)|$.
    
    {\bf Case 2:} $\left|N\left(V_i^{(b)}\right)\right| > \frac{n}{|S_{\inT}|}$.
    Following \Cref{thm:mssp-reduction}, we obtain the time bound 
    \begin{align*}
        &\quad~\bigO{\sum_{t \in T(i, b)} \TMPP\left(S'(t), N(V_i^{(b)}), N(V_{i}^{(b)}) \mid B \right)} \\
        &= \bigtO{ \frac{d \phi n \log B}{\left|N\left(V_{i}^{(b)}\right) \right|} \TMPP \left(\frac{|S_{\inT}| \left|N\left(V_{i}^{(b)}\right) \right|}{2 \phi n \log B}, N(V_i^{(b)}), N(V_{i}^{(b)}) \mid B \right)} \\
        &= \frac{n^{1 + o(1)} \log B}{\left|N\left(V_{i}^{(b)}\right) \right|} \TMPP \left(\frac{|S_{\inT}| \left|N\left(V_{i}^{(b)}\right) \right|}{n}, N(V_i^{(b)}), N(V_{i}^{(b)}) \mid B \right)
    \end{align*}
    over all iterations.
    Here, we use \Cref{lem:non-0-bounded-min-plus} instead of BMM to compute the matrix product $X * Y$.
    In the first equality, we apply \Cref{lem:mat-mul-convex} and \Cref{lem:round-ub-wt}.
    In the second, we plug in our values $d, \phi \gets n^{o(1)}$.

    Finally, to bound the time of the fourth step, we fix a source $s \in S_{\inT}$.
    Since $\hat{L}(s, r)$ are disjoint so are $\calA(s, b)$ over different iterations with distinct $\level(s, b)$, the total time required by all updates to $\hat{L}$, $\calA$, and $\score$ variables is $\tO{n}$.
    Updating $\level(s, b)$ requires $\tO{Bn}$ time.
    Overall, the fourth step over all sources requires $\bigtO{|S_{\inT}|Bn} = \bigtO{Bn^2}$ time.

    Summing over all $V_{i}^{(b)}$ (again bucketing $N\left(V_{i}^{(b)}\right)$ by size in increasing powers of $2$), we obtain the total time bound
    \begin{equation*}
        \max_{\frac{n}{|S_{\inT}|} \leq T \leq n} \frac{n^{2 + o(1)} \log B}{T^2} \TMPP \left( \frac{|S_{\inT}|T}{n}, T, T \mid B \right) = n^{o(1)} \log B \TMPP \left( |S_{\inT}|, n, n \mid B \right) 
    \end{equation*}
    where we apply \Cref{clm:mm-quadratic-growth} to argue that the time is maximized when $T = n$.
    This concludes the proof of \Cref{thm:bounded-weight-mssp}. 
\end{proof}

We prove the remaining lemmas.
We begin with a few useful claims.

\begin{claim}
    \label{clm:inc-only-if-equal}
    $\level(s, b)$ increases only in iterations where $\level(s, b) = \level(s)$ at the start of the iteration.
\end{claim}

\begin{proof}[Proof of \Cref{clm:inc-only-if-equal}]
    We prove this claim inductively.
    The base case follows as initially all $\level(s)$ and $\level(s, b) = 0$.
    Consider the first iteration where $\level(s) = \ell$, $\level(s, b) > \ell$, and $\level(s, b)$ is incremented.
    Thus, we can assume $\level(s, b) = \ell' > \ell$ is the smallest multiple of $b$ greater than $\ell$.
    Note that $\level(s, b)$ can only increase via $\increaseLevel$ or $\updateScores$.
    In $\increaseLevel$, we only increase $\level(s, b)$ if $\level(s, b) = \level(s)$.
    Otherwise, we run $\extendPaths$ and call $\updateScores$ if $\calA(s, b) = \emptyset$.
    In particular, in this iteration, since $s \in S'$, a cluster $V_{i}^{(b)}$ is selected such that $\calA(s, b) \cap V_{i}^{(b)} \neq \emptyset$.

    By induction, in the iteration where $\level(s, b)$ is set to $\ell'$, that iteration began with $\level(s) = \level(s, b) = \ell' - b$ and $\calA(s, b) = \emptyset$ at some point in that iteration (since this is a requirement to trigger $\increaseLevel$ or $\updateScores$).
    Regardless of which one was triggered, $\calA(s, b)$ is only set to be non-empty if $\level(s) = \level(s, b) = \ell'$, a contradiction as $\level(s)$ cannot decrease.
\end{proof}

\begin{claim}
    \label{clm:level-s-b-bound}
    At the start of any iteration of the while loop,
    \begin{equation*}
        \level(s) \leq \level(s, b) \leq \level(s) + b
    \end{equation*}
    for all $s \in S_{\inT}$ and $b$.
\end{claim}

\begin{proof}[Proof of \Cref{clm:level-s-b-bound}]
    Observe that $\level(s) \leq \level(s, b)$ by definition.
    Since at most one of $\updateScores$ or $\increaseLevel$ is called during any single iteration, $\level(s, b)$ increases by at most $b$.
    Furthermore, by \Cref{clm:inc-only-if-equal} $\level(s, b)$ does not increase whenever $\level(s, b) > \level(s)$.
    Thus, $\level(s, b)$ can never exceed $\level(s)$ by more than $b$.
\end{proof}

We are ready to prove the required lemmas.

\LevelIncBound*

\begin{proof}[Proof of \Cref{clm:level-inc-bound}]
    At the start of any iteration, we have $\ell = \level(s) \leq \level(s, 1) \leq \level(s) + 1$ by \Cref{clm:level-s-b-bound}.
    If $\level(s, 1) = \ell + 1$, then by \Cref{clm:inc-only-if-equal}, $\level(s, 1)$ is not incremented during this iteration, so that $\level(s) \leq \level(s, 1) \leq \ell + 1$ at the end of this iteration.
    On the other hand, if $\level(s, 1) = \ell$, then since at most one of $\increaseLevel$ or $\updateScores$ is invoked during any single iteration, we have $\level(s) \leq \level(s, 1) \leq \ell + 1$ at the end of the iteration, as desired.

\end{proof}

\MultipleBExist*

\begin{proof}[Proof of \Cref{clm:multiple-b-exists}]
    From \Cref{clm:level-inc-bound}, we have that $\level(s)$ increments by at most one in every iteration.
    Thus, since $\level(s) \leq \level(s, b)$, consider the first iteration where $\ell = \level(s) \leq \level(s, b)$.
    We claim that in fact $\level(s, b) = \level(s)$.
    Suppose not for contradiction.
    Then, we have that at the start of this iteration, $\level(s, b) \geq \ell + b$ since $\level(s, b)$ increases in increments of size $b$, and the same bound holds at the end of the previous iteration.

    Now, consider the previous iteration which began with $\level(s) = \ell - 1$.
    From \Cref{clm:level-s-b-bound}, at the start of this iteration $\level(s, b) \leq \ell - 1 + b < \ell + b$ so that $\level(s, b)$ must have been incremented in the previous iteration.
    Furthermore, at the start of the previous iteration, $\level(s, b) \geq \level(s) = \ell - 1$.
    We consider two cases.

    If $b > 1$, then $\level(s, b) = \ell$ since $\ell - 1 \mod b \not\equiv 0$.
    Then, regardless of whether $\level(s)$ is incremented via $\updateScores$ or $\increaseLevel$, $\level(s, b)$ is only incremented if $\level(s, b) = \ell - 1$ by \Cref{clm:inc-only-if-equal}, which is impossible unless $b = 1$ (again since $\level(s, b)$ increases in increments of size $b$).
    
    Thus, assume $b = 1$.
    In this case, $\level(s, b) = \ell - 1$ at the start of the previous iteration (otherwise $\level(s, b) \neq \level(s)$ is not increased) so that $\level(s, b) = \ell$ at the start of this iteration, as desired.
\end{proof}

\IncreaseLevelTimeBound*

\begin{proof}[Proof of \Cref{lem:increase-level-time-bound}]
    First, we claim that $\level(s, b) \leq Bn$ for all $b$.
    
    \begin{claim}
        \label{clm:level-bound}
        $\level(s, b) \leq Bn$ for all $b = 2^{t}$ and $s \in S$.
    \end{claim}

    \begin{proof}[Proof of \Cref{clm:level-bound}]
        Since $G$ is connected, we have $\ecc(s) \leq Bn$ so that by \Cref{lem:dist-s-v-ub-wt}, we have that $\bigcup_{r = 0}^{Bn} \hat{L}(s, r) = V$, so we remove $s$ from $S$ after level $Bn$.
    \end{proof}

    Now, fix a single source $s$.
    By \Cref{lem:dist-s-v-ub-wt} and the fact that we never remove vertices from $\hat{L}(s, r)$ for any $r$, we have that $\hat{L}(s, r)$ are always disjoint.
    Thus, the overall time required by step 1 of $\increaseLevel$ is $O(Bn)$, since $\increaseLevel$ is called at most $\ecc(s) \leq Bn$ times and in each iteration, step 1 requires $O\left(1 + \left|\hat{L}(s, \level(s) + 1))\right|\right)$ time.
    
    Note step 2 over all iterations requires $O(n B \log B)$ time.

    Finally, in step 3, we observe that any vertex is added to and removed from $\calA(s, b)$ exactly once.
    When we add or remove a vertex from $\calA(s, b)$, we can also update the score as appropriate by maintaining a membership data structure for each $V_{i}^{(b)}$.
    Thus, the total time required by step 3 is $O(n B \log B)$, since for every $b$ and $s$, the time at one iteration is $\bigtO{\left|\bigcup_{\level(s) - b < r \leq \level(s)}  L(s, r) \right| + 1}$.
    Initializing the membership data structures also requires $\tO{n \log B}$ time. 

    Thus, the total cost of all $\increaseLevel$ calls over all sources is at most $\bigtO{|S_{\inT}| n B}$ time.
\end{proof}

%% file: applications.tex
\subsection{Applications}
\label{sec:applications}

We give some applications of our \mssp~algorithm.

\subsubsection{Hop-Set Construction}

Our first application is the first non-trivial algorithm for explicitly constructing hop-sets.

\begin{definition}
    \label{def:hop-set}
    A set of weighted edges $F$ is a $\beta$-hop-set for $G = (V, E, \wt)$ if $\dist_{G}(u, v) = \dist_{G \cup F}(u, v)$ for all $u, v \in V$ and for every $u, v$, there is a path $\pi$ of at most $\beta$ hops from $u$ to $v$ such that $\pi$ is a shortest path between $u, v$.
\end{definition}

We give an efficient algorithm for computing exact $\beta$-hop-sets on dense graphs, improving upon the known $\min\left( n^3/\beta, n^{\omega} \right)$ time algorithm.

\begin{theorem}
    \label{thm:hop-set}
    For all $\beta \geq 1$, there is a randomized $Bn^{o(1)}\TMUL(n/\beta, n, n)$ time algorithm that computes a $\beta$-hop-set with $\tO{n^2/\beta^2}$ edges on undirected graphs with integer weights in $[B]$.
\end{theorem}

\begin{proof}
    When $\beta \leq 30 \log n$, we observe that the theorem is true by computing APSP in $B n^{o(1)} \TMUL(n, n, n)$ time and adding edges of weight $\dist(u, v)$ between every pair of vertices. 
    Thus, we assume $\beta > 30 \log n$.
    Consider the following algorithm.
    \begin{mdframed}
        \begin{enumerate}
            \item {\bf Input:} Undirected graph $G$ with integer weights in $[B]$. 
            \item Sample $S \subseteq V$ by including each vertex independently with probability $\frac{30 \log n}{\beta}$
            \item Compute $\hdist(S, V) \gets \mssp(G, S)$.
            \item Compute $F \gets \set{(a, b) \given a, b \in S}$ where edge $(a, b)$ has weight $\hdist(a, b)$.
            \item {\bf Output:} Weighted edge set $F$.
        \end{enumerate}
    \end{mdframed}

    First, we argue that $F$ is a $\beta$-hop-set.
    Consider a shortest path of at least $\beta/3$ nodes.
    Then, the probability that none of these nodes is included in $S$ is at most
    \begin{equation*}
        \left( 1 - \frac{30 \log n}{\beta} \right)^{\beta/3} \leq n^{-10} 
    \end{equation*}
    so that union bounding over $O(n^2)$ shortest paths, we can ensure with high probability that all shortest paths of at least $\beta/3$ nodes include one node in $S$.
    Then, for any shortest path of length $\ell \geq \beta$, $\pi = v_0, \dots, v_{\beta/3}, \dots, v_{\ell - \beta/3}, \dots, v_{\ell}$, we have $v_i, v_j \in S$ for $i \leq \beta/3$ and $j \geq \ell - \beta/3$.
    Then, adding $F$, we have the path $\pi' = v_0, \dots, v_i, v_j, \dots, v_{\ell}$ of hop length at most $2\beta/3 + 1$.
    Furthermore, $\wt(v_i, v_j) = \hdist(v_i, v_j) = \dist(v_i, v_j)$ by correctness of the \mssp~algorithm.

    To analyze the time, sampling requires $O(n)$ time, \mssp~requires $Bn^{o(1)}\TMUL(n/\beta, n, n)$ time by \Cref{thm:bounded-weight-mssp}, and constructing $F$ requires $O(|S|^2)$ time.
    We thus conclude by bounding $|S|$.
    By a Chernoff bound, we have with high probability $|S| = \bigO{\frac{n \log n}{\beta}}$ so that $|S|^2 = \tO{n^2/\beta^2}$, as desired.
\end{proof}

\subsubsection{Additive Emulator Construction}

We give a simple, efficient algorithm for constructing a $+4$-emulator on undirected graphs.

\begin{theorem}
    \label{thm:4-emulator}
    There is a randomized $O(Bn^{\omega(1, 2/3, 1) + o(1)}) = \bigO{Bn^{2.1321}}$ time algorithm to construct $+4B$-emulators of size $O(n^{4/3} \log n)$ on undirected graphs with integer weights at most $B$.
\end{theorem}

\begin{proof}
    Elkin, Gitlitz, and Neiman \cite[Section 7]{elkin2023improved} present the following algorithm for obtaining a $+4$-emulator.
    
    \begin{mdframed}
        {\bf Input: } Graph $G$.
        \begin{enumerate}
            \item Initialize $H \gets \emptyset$ and $S \subseteq V$ a random set chosen by sampling each vertex independently with probability $n^{-1/3}$.
            \item Add the $O(n^{1/3} \log n)$ lightest edges incident to every vertex to $H$.
            \item Compute \mssp~from $S$ and add $S \times S$ to $H$ with weights $\wt(s_1, s_2) = \dist(s_1, s_2)$.
        \end{enumerate}
    \end{mdframed}

    Since we follow the algorithm exactly, we omit the proof of correctness. 
    To analyze the running time, note that the first two steps require $O(n^2 \log n)$ time.
    By \Cref{thm:bounded-weight-mssp}, \mssp~requires $Bn^{\omega(1, 2/3, 1) + o(1)}$, which bounds the running time of our algorithm.
\end{proof}

%% file: directed_reach.tex
\section{Reachability and Shortest Paths on Directed Graphs}
\label{sec:dir}

We give a simple recursive scheme that obtains optimal running time for directed reachability.

\begin{definition}[Multiple Source Reachability]
    \label{def:msr}
    Given directed graph $G = (V, E)$ and a set of sources $S_{\inT} \subseteq V$, compute $\reach(s, v)$ for all $s \in S_{\inT}$ and $v \in V$.
\end{definition}

\let\oldfootnote\footnote
\renewcommand{\footnote}[1]{\textsuperscript{\textup{\ref{fn:msreach}}}}
\MultiSourceReach*
\let\footnote\oldfootnote

\begin{proof}
    Without loss of generality, we may assume that the graph $G$ is a directed acyclic graph (DAG) and $G$ is topologically sorted.
    In particular, we first decompose the graph into strongly connected components (SCC) and apply a topological sort to the remaining DAG.
    After this transformation, let $S_{\inT}$ denote all the strongly connected components of the input graph that intersect the input $S_{\inT}$.
    Note that both the SCC decomposition (which preserves reachability) and topological sort can be computed in linear time.

    Given the assumption, we let $a \leq b$ denote comparison in the topological order.
    Let $[a, b] \subseteq V$ denote all vertices between $a$ and $b$ (inclusive) in the topological order.
    Furthermore, given $s, v \in V$, let $\pred(s, v)$ denote the minimum vertex $a$ (in topological order) such that $\reach(s, a) = 1$ and $(a, v) \in E$. 
    That is, there is a path from $s$ to $v$ such that $a$ is the vertex immediately preceding $v$.
    
    With our assumptions, consider the following recursive procedure.
    \begin{mdframed}
        $\complete(G, S_{\inT}, L, R)$:

        {\bf Input:} Topologically sorted DAG $G$, sources $S_{\inT}$, and interval endpoints $L \leq R$.
        
        \Comment{(Assume $\reach(s, v)$ is computed for $s \in S_{\inT}$ and $v$ where either $v \leq L$ or $\pred(s, v) \leq L$.)}
        
        \begin{enumerate}
            \item If $R - L \leq 2$, compute $\reach(s, v)$ for all $s \in S_{\inT}$ and $L < v \leq R$.
            \item Set $M \gets \lfloor (L + R) / 2 \rfloor$ and run $\complete(G, S_{\inT}, L, M)$.
            \item Construct $S_{\inT} \times [L, M]$ matrix $X$ and $[L, M] \times [M, R]$ matrix $Y$ with entries
            \begin{equation*}
                X[s, u] = \reach(s, u),\quad Y[u, v] = \ind{(u, v) \in E} \text{.}
            \end{equation*}
            Compute $XY$ and update $\reach(s, v) \orget (XY)[s, v]$ for all $s \in S_{\inT}, v \in V$.
            \item Run $\complete(G, S_{\inT}, M, R)$.
        \end{enumerate}
    \end{mdframed}

    Our algorithm for Multiple-Source Reachability 
    computes $\complete(G, S_{\inT}, 0, |V|)$.
    We prove the correctness of our algorithm via the following inductive claim.

    \begin{claim}
        \label{clm:complete-induct-r}
        Suppose $\reach(s, v)$ is computed for all $(s, v)$ such that $s \in S_{\inT}$ and $v \leq L$ or $\pred(s, v) \leq L$ before the execution of $\complete(G, S_{\inT}, L, R)$.
        Then, after the execution of $\complete(G, S_{\inT}, L, R)$, $\reach(s, v)$ is computed for all $(s, v)$ with $s \in S_{\inT}$ and $v \leq R$.
    \end{claim}

    First, we argue that the claim suffices to compute $\reach(s, v)$ for all $s \in S_{\inT}$ and $v \in V$.
    When $L = 0$, there is exactly one vertex $\leq L$, say $s^*$ in $L$.
    If $s^* \not\in S_{\inT}$, then we set $\reach(s, v) = 0$ for all $s, v$.
    If $s^* \in S_{\inT}$, we note $\pred(s, v) \leq 0$ if and only if $v \in N^{+}(s^*)$ is an out-neighbor of $s^*$.
    In this case, we update $\reach(s^*, v) = 1$ for all $v \in N^{+}(s^*)$ in $O(n)$ time.
    In either case, the assumption of \Cref{clm:complete-induct-r} is satisfied.
    Thus, the claim ensures that $\reach(s, v)$ is computed for all $s \in S_{\inT}$ and $v \in V$.
    It remains to prove \Cref{clm:complete-induct-r}.

    \begin{proof}[Proof of \Cref{clm:complete-induct-r}]
        We begin by describing how to compute $\complete$ in the base case (i.e. $R - L \leq 2$).
        Note that in $O(1)$ time, we can compute $\reach(a, b)$ for all $a, b \in [L, R]$.
        Then, for all $s \in S_{\inT}$ and $v \in [L, R]$, we set
        \begin{equation}
            \label{eq:base-r-update}
            \reach(s, v) \orget \bigvee_{L \leq u < v} (\reach(s, u) \wedge \reach(u, v)) \text{.}
        \end{equation}
        This completes the description of our algorithm.

        To prove correctness, we first observe that whenever $\reach(s, v) = 1$, there exists a path from $s$ to $v$.
        Thus, it suffices to show that we set $\reach(s, v) = 1$ if there exists a path from $s$ to $v$.
        Fix $s \in S_{\inT}$ and $v \leq R$ such that there is a path from $s$ to $v$.
        We may assume $v > L$ and $\pred(s, v) > L$.
        Otherwise, we already have $\reach(s, v) = 1$ by the input assumption, and observe that our algorithm never decreases (i.e.\ flip from 1 to 0) $\reach(s, v)$.
        
        We then proceed by induction, beginning with the base case.
        In the base case, we need to compute $\reach(s, v)$ for $s \in S_{\inT}$ and $L < v \leq R \leq L + 2$.
        Since we compute the transitive closure of $[L, R]$, we compute reachability following the update rule of \Cref{eq:base-r-update} if $s \leq L$ and from the transitive closure of $[L, R]$ if $L \leq s \leq R$.

        We now conclude the proof of correctness with the inductive case.
        We may assume $v > M$, as otherwise the first recursive call $\complete(G, S_{\inT}, L, M)$ computes $\reach(s, v)$.
        We consider two cases.

        \paragraph{Case 1: $\pred(s, v) \leq M$.}
        We have $L < \pred(s, v) \leq M$, or there exists $L < u \leq M$ such that $\reach(s, u) = 1$ and $(u, v) \in E$.
        In this case, since $u \leq M$, we have $X[s, u] = \reach(s, u) = 1$ and $Y[u, v] = 1$.
        In particular, $(XY)[s, v] = 1$ so we set $\reach(s, v) \gets 1$ as desired.

        \paragraph{Case 2: $\pred(s, v) > M$.}
        From the previous case, when we call $\complete(G, S_{\inT}, M, R)$, we have computed $\reach(s, v)$ for all $s \in S_{\inT}$ and $v \leq M$ (handled by first recursive call) or $\pred(s, v) \leq M$ (handled by Case 1). 
        Thus, we satisfy the input assumption of \Cref{clm:complete-induct-r}, and obtain $\reach(s, v)$ for all $v \leq R$ by induction.
    \end{proof}
    
    Finally, we bound the running time.
    Note that an invocation of the base case requires $O(|S_{\inT}|)$ time.
    In the base case, we compute the transitive closure of $[L, R]$ in $O(1)$ time as the graph is sorted in topological order, so it suffices to compute the transitive closure of $G[L, R]$, a graph of $O(1)$ vertices.
    Next, compute \Cref{eq:base-r-update} in $O(|S_{\inT}|)$ time, which bounds the running time of the base case.
    
    Consider an invocation of the algorithm when $R - L = \ell$.
    In this case, we make two recursive calls to intervals of size $\ell/2$ and compute a $S_{\inT} \times \ell \times \ell$ matrix product in $\TBMM(S_{\inT}, \ell, \ell)$ time.
    In other words, we obtain the recursion 
    \begin{equation*}
        T(\ell) = 2 T(\ell/2) + \TBMM(S_{\inT}, \ell, \ell)
    \end{equation*}
    which we claim resolves to $T(\ell) = \tO{\TBMM(S_{\inT}, \ell, \ell)}$.
    Formally, we first show the following claim.

    \begin{claim}
        \label{clm:mm-exp-geq-2}
        For any $\ell \geq 0$,
        \begin{equation*}
            \TBMM(S_{\inT}, \ell, \ell) \geq 2 \TBMM(S_{\inT}, \ell/2, \ell/2) \text{.}
        \end{equation*}
        The claim holds also if we replace $\TBMM$ with $\TMPP(S_{\inT}, \ell, \ell \mid B_1, B_2)$.
    \end{claim}

    \begin{proof}
        Let $A, B$ be arbitrary $S_{\inT} \times \ell/2$ boolean  matrices (resp. integer matrices with finite entries in $[B_1]$) and $C, D$ be arbitrary $\ell/2 \times \ell/2$ boolean matrices (resp. integer matrices with finite entries in $[B_2]$).
        Then, define
        \begin{equation*}
            X = \begin{pmatrix}
            A & B 
            \end{pmatrix},\quad Y = \begin{pmatrix}
                C & 0_{\ell/2 \times \ell/2} \\
                0_{\ell/2 \times \ell/2} & D
            \end{pmatrix}
        \end{equation*}
        so that $XY = \begin{pmatrix} AC & BD \end{pmatrix}$.
        In particular, we compute two arbitrary $S_{\inT} \times \ell/2 \times \ell/2$ matrix products with one $S_{\inT} \times \ell \times \ell$ matrix product.
    \end{proof}
    
    Then, we can bound the running time of one recursive call as
    \begin{align*}
        T(\ell) &= 2 T(\ell/2) + \TBMM(S_{\inT}, \ell, \ell) \\
        &= \bigO{\sum_{0 \leq i \leq \log \ell} \frac{\ell}{2^{i}} \TBMM(S_{\inT}, 2^{i}, 2^{i})} \\
        &= \bigO{\TBMM(S_{\inT}, \ell, \ell) \log \ell} 
    \end{align*}
    where in the final inequality we apply \Cref{clm:mm-exp-geq-2}.
\end{proof}

In fact, our algorithm, with minor modifications, can be used to compute \mssp{} on unweighted directed acyclic graphs.

\DAGMSSP*

We in fact prove a slightly stronger result.

\begin{proposition}
    \label{prop:dag-mssp}
    Let $B, D \geq 1$.
    Then, there is a deterministic algorithm computing \mssp{} on directed acyclic graph with edge weights in $\set{0, \dots, B}$ and $\diam(G) \leq D$ in $\tO{\TMPP(S_{\inT}, n, n \mid D, B)}$ time.
\end{proposition}

\Cref{thm:dag-mssp} follows by applying $B = 1$ and $D = n$.

\begin{proof}
    As before, we assume the graph is sorted in topological order, and keep the same notation as \Cref{thm:msreach}.
    We let $\hdist(s, v)$ denote the estimated distance between $s \in S_{\inT}$ and $v \in V$, with all values initialized to $\infty$.
    We then run the following algorithm, replacing $\reach$ with $\hdist$ and $\BMM$ with $\MPP$.
    \begin{mdframed}
        $\complete(G, S_{\inT}, L, R)$:

        {\bf Input:} Topologically sorted DAG $G$, sources $S_{\inT}$, and interval endpoints $L \leq R$.
        
        \Comment{(Assume $\hdist(s, v) = \dist(s, v)$ is computed for $s \in S_{\inT}$ and $v$ where either $v \leq L$ or $\pred(s, v) \leq L$.)}
        
        \begin{enumerate}
            \item If $R - L \leq 2$, compute $\hdist(s, v)$ for all $s \in S_{\inT}$ and $L < v \leq R$.
            \item Set $M \gets \lfloor (L + R) / 2 \rfloor$ and run $\complete(G, S_{\inT}, L, M)$.
            \item Construct $S_{\inT} \times [L, M]$ matrix $X$ and $[L, M] \times [M, R]$ matrix $Y$ with entries
            \begin{equation*}
                X[s, u] = \hdist(s, u),\quad Y[u, v] = \begin{cases}
                    \wt(u, v) & (u, v) \in E \\
                    \infty & \otherwise
                \end{cases} \text{.}
            \end{equation*}
            Compute $X * Y$ and update $\hdist(s, v) \minget (X * Y)[s, v]$ for all $s \in S_{\inT}, v \in [M, R]$.
            \item Run $\complete(G, S_{\inT}, M, R)$.
        \end{enumerate}
    \end{mdframed}

    In the base case when $R - L \leq 2$, we compute \apsp{} on the induced subgraph $G[L, R]$ which has $O(1)$ vertices and contains the shortest path between any pair of vertices $u, v \in [L, R]$.
    This requires $O(1)$ time and we can update distances analogously to \Cref{eq:base-r-update} as follows for all $s \in S_{\inT}, v \in [L, R]$:
    \begin{equation}
        \label{eq:base-dag-update}
        \hdist(s, v) \minget \min_{L \leq u \leq R} \hdist(s, u) + \hdist(u, v) \text{.}
    \end{equation}
    Our algorithm for \mssp{} sets $\hdist(s, v) = 0$ for all $s \in S_{\inT}, v \in N(s)$ and computes $\complete(G, S_{\inT}, 0, |V|)$.
    
    The correctness proof follows similarly to \Cref{thm:msreach}, we highlight the necessary modifications here.
    As before, without loss of generality let the first vertex in topological order denoted $s^*$ be in $S_{\inT}$.
    At the root level, when we invoke $\complete(G, S_{\inT}, 0, |V|)$, we have $\hdist(s, v) = \dist(s, v)$ for all $s \in S_{\inT}$ and $\pred(s, v) \leq 0$ i.e. $\pred(s, v) = s^*$, satisfying the input assumption.

    In the base case (i.e.\ $R - L \leq 2$) we compute $\hdist(u, v) = \dist(u, v)$ for all $u, v \in [L, R]$ as $G[L, R]$ contains all shortest paths between $u, v \in [L, R]$ since $G$ is a DAG.
    Thus, we obtain correct distances following \Cref{eq:base-dag-update} if $s \leq L$ and from APSP on $G[L, R]$ if $L \leq s \leq R$.
    The proof of the inductive case follows identically to \Cref{thm:msreach}.
    
    To bound the running time, we have the recursion
    \begin{equation*}
        T(\ell) = 2 T(\ell/2) + \TMPP(S_{\inT}, \ell, \ell \mid D, B)
    \end{equation*}
    since $X$ has finite entries bounded by $D$ and $Y$ has finite entries bounded by $B$.
    Following the same analysis as \Cref{clm:mm-exp-geq-2}, we obtain the time bound $T(\ell) = \tO{\TMPP(S_{\inT}, \ell, \ell \mid D, B)}$.
\end{proof}
\subsection{Application: Shortcut Sets}

We give our improved algorithm for computing shortcut sets of linear size.
Recall that for a directed graph $G$, $\TC(G)$ denotes the transitive closure of $G$ and $\diam(G)$ is the maximum of $\dist(u, v)$ over $(u, v) \in \TC(G)$.

\begin{definition}
    \label{def:shortcut-set}
    A set of edges $F$ is a $D$-shortcut-set for directed graph $G = (V, E)$ if $\TC(G) = \TC(G \cup F)$ and $\diam(G \cup F) \leq D$.
\end{definition}

We give a faster algorithm for computing shortcut sets.

\begin{theorem}
    \label{thm:shortcut-set}
    For $D = O(\sqrt{n})$, there is an algorithm that computes a $D$-shortcut-set of size $\tO{n^2/D^3 + n}$ in time $\bigtO{n^{2} + \TMPP(n/D^2, n, n \mid D, 1)}$.
\end{theorem}

We remark that using our algorithm for Multiple Source Reachability, one can achieve a running time of $\TBMM(n/D, n, n)$. 
However, \Cref{thm:shortcut-set} is at least as fast for all $D = O(\sqrt{n})$.

\begin{proof}[Proof of \Cref{thm:shortcut-set}]
    We recap the algorithm of \cite{kogan2022beating}, beginning with the necessary definitions.
    Without loss of generality, we may assume that the input graph $G$ is directed and acyclic (DAG).
    Given a graph $G$, $\LP(G)$ denotes the length of the longest simple path in $G$.
    For a subset of vertices $U \subseteq V$, $G[U]$ is the induced subgraph of $G$ on $U$.
    We define $\ell$-covers and chains.

    \begin{definition}
        \label{def:l-cover}
        Given a directed graph $G$, an $\ell$-cover is a collection of $\ell$ (not necessarily disjoint) paths in $G$, $\calP = \set{P_{1}, \dots, P_{\ell}}$ satisfying:
        \begin{enumerate}
            \item (Size Bound) $\sum_{P \in \calP} |P| \leq \min(\ell, \diam(G)) \cdot n$.
            \item (Longest Path Bound) $\LP(G') \leq 2n/\ell$ where $G' = \TC(G)[U]$ and $U = V \setminus V(\calP)$.
        \end{enumerate}
    \end{definition}

    \begin{definition}
        \label{def:chain}
        Given a directed graph $G$, a chain is a path in $\TC(G)$.
    \end{definition}

    Given a vertex $v$ and a chain $C$, let $w := w(v, C)$ denote the first vertex in $C$ such that $(v, w) \in \TC(G)$, if such a vertex exists.
    We are now ready to state the algorithm of \cite{kogan2022beating}.

    \begin{mdframed}
        \begin{enumerate}
            \item Compute $H_0$, a shortcut set of diameter $D_0 = O(\sqrt{n})$ of size $\tO{n}$.
            \item Compute an $\ell$-cover $\calP$ of $G \cup H_{0}$ where $\ell = 16 n/D$.
            \item Compute $\calC$, a collection of vertex disjoint chains in $G$ satisfying (i) $|\calC| \leq |\calP|$ and (ii) $V(\calC) = V(P)$.
            \item For every $C \in \calC$, compute $H(C)$: a $2$-shortcut-set for $C$.
            \item Let $V' \subset V, \calC' \subset \calC$ be random subsets where each vertex (resp. chain) is sampled with probability $p = \Theta(\log n/D)$.
            \item Compute $\hat{H} = \set{(v, w(v, C)) \given v \in V', C \in \calC' \text{ if } w(v, C) \text{ exists}}$.
            \item Return $H_{0} \cup \left(\bigcup_{C \in \calC} H(C)\right) \cup \hat{H}$.
        \end{enumerate}
    \end{mdframed}

    The above algorithm returns a $D$-shortcut-set of size $\tO{n^2/D^3 + n}$ (see Theorem 1.4 of \cite{kogan2022beating}).
    It suffices for us to analyze the running time.
    Step 1 requires time $\tO{m + n^{3/2}}$ by Theorem 1.3 of \cite{kogan2022beating}.
    Step 2 requires time $\tO{m + n^{3/2} + n \cdot \min(\diam(G), \ell)}$ by Theorem 2.3 of \cite{kogan2022beating}.
    Step 3 requires time $\tO{\sum_{P \in \calP} |P|}$ by Lemma 1.15 of \cite{kogan2022beating}.
    Step 4 requires time $\tO{\sum_{C \in \calC} |C|}$ by Lemma 1.10 of \cite{kogan2022beating}.
    Step 5 requires time $O(n + |\calC|)$.
    Step 7 requires time $\tO{n}$ by the size of the shortcut set.
    In total, all steps (excluding step 6) require time $\tO{n^2}$.

    It thus suffices to analyze the running time of step 6, for which we propose a new algorithm.
    First, we need a simple bound on $|V'|, |\calC'|$.

    \begin{claim}
        \label{clm:chain-size-bound}
        With high probability, $|V'| = \bigO{n \log n /D}$ and $|\calC'| = \bigO{n \log n/D^2}$.
    \end{claim}

    \begin{proof}
        Note that $|V'|$ is a sum of independent Bernoulli variables with mean $O(n \log n/D)$.
        A Chernoff bound shows that $\Pr(|V'| > 10 \mu) < \exp(- \mu) \ll n^{-10}$.
        A similar bound holds for $|\calC'|$, a sum of independent Bernoulli variables with mean $\tO{n \log n /D^2}$.
    \end{proof}

    Next, we argue that we can assume that the length of each chain is at most twice the average.

    \begin{claim}
        \label{clm:chain-split}
        Given a collection of vertex disjoint chains $\calC_{0}$ of size $\ell$, there is a collection of vertex disjoint chains $\calC$ of size $\ell$ satisfying (1) $V(\calC) = V(\calC_{0})$ and (2) each $C \in \calC$ satisfies $|C| \leq 2 n/\ell$.
        Furthermore, $\calC$ can be computed in $O(n)$ time.
    \end{claim}

    \begin{proof}
        For every chain $C \in \calC_{0}$ of length $|C| > n/\ell$, we split $C$ into $|C| \ell / n$ chains of length at most $2 n/\ell$.
        Let $\calC$ be the union of all the split chains.
        Note that $\calC$ is vertex disjoint and $V(\calC) = V(\calC_{0})$.
        Furthermore, every $C \in \calC$ satisfies $|C| \leq 2 n/\ell$.
        We bound the size of $\calC$ as
        \begin{equation*}
            \sum_{C \in \calC_{0}} |C| \ell / n \leq \frac{\ell}{n} \sum_{C \in \calC_{0}} |C| \leq \ell 
        \end{equation*}
        where we use the fact that $\calC_{0}$ are vertex disjoint so $\sum_{C \in \calC_{0}} |C| \leq n$.
        It is clear that $\calC$ can be computed in $O(n)$ time.
    \end{proof}

    By running \Cref{clm:chain-split} immediately after step 3, we have the additional property that each chain has length $\leq 2n/\ell = D/8$.
    Note that we maintain the required properties of $\calC$, so that correctness is maintained.
    As a corollary of \Cref{clm:chain-size-bound} and \Cref{clm:chain-split}, we have that with high probability
    \begin{equation}
        \label{eq:total-chain-bound}
        V(\calC') = O(n \log n / D) \text{.}
    \end{equation}
    We now present our algorithm for computing $\hat{H}$.
    Consider the graph $\overline{G}$ constructed as follows:

    \begin{enumerate}
        \item Add a vertex for every $v \in V(G)$ and edge $e \in E(G)$ with weight $\wt(e) = 0$.
        \item For every $C \in \calC'$, add a vertex $t_{C}$.
        \item For every $C \in \calC'$, denote $C = \set{c_{1}, \dots, c_{p}}$. 
        Add a path $\pi_{C, i}$ of $i$ edges with weight $1$ from $c_{i}$ to $t_{C}$ for all $1 \leq i \leq p$.
        Note that distinct $\pi_{C, i}$ are vertex disjoint except for $t_{C}$.
    \end{enumerate}

    We claim that distances from $v \in V'$ to $t_{C}$ for $C \in \calC'$ allow us to compute $w(v, C)$.

    \begin{claim}
        \label{clm:mssp-computes-e(v,C)}
        Let $v \in V'$ and $C \in \calC'$.
        Let $C = \set{c_{1}, \dots, c_{p}}$.
        Then, $w(v, C) = c_{i}$ if and only if $\dist_{\overline{G}}(v, t_{C}) = i$.
    \end{claim}

    \begin{proof}
        Fix $v \in V', C \in \calC'$ and let $w := w(v, C)$.
        If $w = c_{i}$, there is a path from $v$ to $w$ using only edges in $E(G)$, and thus a path $P_{0}$ of weight $0$ from $v$ to $w$ in $\overline{G}$.
        Then, $P_{0} \circ \pi_{C, i}$ is a path of length $i$ from $v$ to $t_{C}$, establishing $\dist_{\overline{G}}(v, t_{C}) \leq i$.

        Suppose $\dist_{\overline{G}}(v, t_{C}) < i$.
        We claim there is a vertex $c_{j}$ with $j < i$ such that $(v, c_{j}) \in \TC(G)$, contradicting the minimality of $c_{i}$.
        Let $c_{j}$ be the final vertex in $C$ on the shortest path from $v$ to $t_{C}$.
        Note that $j < i$ since the $\pi_{C, j}$ has length $j$.
        However, this implies that $(v, c_{j}) \in \TC(G)$ since there is no path from vertices in $\pi_{C', i'} \setminus c'_{i'}$ for any $C', i'$ to vertices in $G$, so any path from $v$ to $c_{j}$ must be contained in $E(G)$.
        In particular, $(v, c_{j}) \in \TC(G)$, contradicting the minimality of $i$.
    \end{proof}

    Thus, it suffices to compute $\dist_{\overline{G}}(v, t_{C})$ for all $v \in V'$ and $C \in \calC'$.
    Equivalently, we can compute $\dist_{\overline{G}_{\rev}}(t_{C}, v)$, which requires one invocation of \mssp{} on $\overline{G}_{\rev}$.
    To bound the running time, it remains to show that $\overline{G}$ is a DAG with $\tO{n}$ vertices and diameter $O(D)$.
    The number of vertices in $\overline{G}$ is
    \begin{equation*}
        \bigO{n + |\calC'| D^2} = \bigO{n \log n} \text{.}
    \end{equation*}
    Next, we claim $\overline{G}$ is a DAG.
    First, observe that there is no path from any vertex not in $V(G)$ to a vertex in $V(G)$, so no \SCC{} can contain both a vertex in $V(G)$ and one not in $V(G)$.
    In particular, the additional vertices and edges introduce no paths between vertices in $V(G)$, so no \SCC{} can contain multiple vertices in $V(G)$.
    To conclude, observe that no \SCC{} can contain multiple vertices not in $V(G)$.
    Finally, we claim $\overline{G}$ has diameter $O(D)$.
    This follows as any $(u, v) \in \TC(G)$ has a path of length $0$ in $\overline{G}$ (by following edges in $G$) and any of the additional paths $\pi_{C, i}$ has length $O(D)$ by \Cref{clm:chain-split}.
    
    Thus, we invoke \Cref{prop:dag-mssp} on $\overline{G}_{\rev}$ with sources $\set{t_{C}}$ to compute the required distances and $\hat{H}$ in $\tO{\TMPP(n/D^{2}, n, n \mid D, 1)}$ time, completing the proof of \Cref{thm:shortcut-set}.
\end{proof}

%% file: dag_mssp_lb.tex
\section{Min-Plus Product with Distinct Entry Bounds}
\label{app:dir-mssp-lb}

In this appendix, we discuss lower bounds for directed \mssp{} from Min-Plus Matrix Multiplication.
We begin with the lower bound given in \Cref{eq:directed-lower-bound}.

\begin{proposition}
    \label{prop:dir-mssp-lb}
    Any algorithm that computes \mssp{} on unweighted directed acyclic graphs with $5n$ vertices and $s$ sources requires time 
    \begin{equation*}
        \bigOm{\max_{1 \leq t \leq n} \TMPP\left(s, t, n \mid \frac{n}{\min(s, t)}, \frac{n}{t} \right)} \text{.}
    \end{equation*}
\end{proposition}

\begin{proof}
    For any $1 \leq t \leq n$, consider an $s \times t$ matrix $X$ with entries in $[n/\min(s, t)] \cup \set{\infty}$ and a $t \times n$ matrix $Y$ with entries in $[n/t] \cup \set{\infty}$.
    When $t \leq s$, we immediately inherit the lower bound of
    \begin{equation*}
        \bigOm{\max_{1 \leq t \leq n} \TMPP\left(s, t, n \mid \frac{n}{t} \right)}
    \end{equation*}
    given by \cite{DBLP:conf/icalp/ChanWX21}.
    Thus, for the remainder of the proof, we assume $s \leq t$.
    Construct a graph $G$ with vertex sets
    \begin{equation*}
        A = \set{a_{i}}_{i \in [s]},\quad B = \set{b_{i}}_{i \in [n]},\quad C = \set{c_{i}}_{i \in [n]}, \quad P =  \set{p_{ij}}_{i \in [s], j \in [n/s]}, \quad Q = \set{q_{ij}}_{i \in [t], j \in [n/t]} \text{.}
    \end{equation*}
    Note that the number of vertices is $s + 4n \leq 5n$.
    We describe the edge sets.
    For each $i \in [s]$, there is a directed path $(a_{i}, p_{i1}, \dots, p_{i, n/s})$.
    For each $i \in [t]$, there is a directed path $(b_{i}, q_{i1}, \dots, q_{i, n/t})$.
    For each $(i, k) \in [s] \times [t]$, there is an edge $(p_{iv}, b_{k})$ where $X[i, k] = v$.
    For each $(k, j) \in [t] \times [n]$, there is an edge $(q_{kv}, c_{j})$ where $Y[k, j] = v$.
    This completes the description of our graph.
    Let $S_{\inT} = A$.
    To complete the proof, we claim that for all $(i, j) \in [s] \times [n]$,
    \begin{equation*}
        \dist(a_{i}, c_{j}) = (X * Y)[i, j] + 2 \text{.}
    \end{equation*}

    Suppose $(X * Y)[i, j] = X[i, k] + Y[k, j]$ for some $k \in [t]$.
    Then, the path $(a_{i}, p_{i1}, \dots, p_{iv}, b_{k}, q_{k1}, \dots, q_{kw}, c_{j})$ where $X[i, k] = v$ and $Y[k, j] = w$ is a path of length $v + w + 2$ from $a_{i}$ to $c_{j}$ so $\dist(a_{i}, c_{j}) \leq (X * Y)[i, j] + 2$.

    Suppose for contradiction $\dist(a_{i}, c_{j}) = \ell < (X * Y)[i, j] + 2$.
    Then, there is a path $(a_{i}, v_{1}, \dots, v_{\ell - 1}, c_{j})$.
    By construction, there exists $k', v', w'$ such that the path takes the form $(a_{i}, p_{i1}, \dots, p_{iv'}, b_{k'}, q_{k'1}, \dots, q_{k'w'}, c_{j})$.
    However, this implies that there exists $k'$ with $X[i, k'] = v', Y[k', j] = w'$ such that $v' + w' + 2 = \ell < (X * Y)[i, j] + 2$, contradicting the minimality of $k$.
\end{proof}

There are several known algorithms for computing $\MPP(n_1, n_2, n_3 \mid B_1, B_2)$:
\begin{enumerate}
    \item We can compute $\MPP(n_1, n_2, n_3 \mid \max(B_1, B_2))$ in time $\tO{\max(B_1, B_2) \TMUL(n_1, n_2, n_3)}$ \cite{shoshan_zwick}.
    \item \cite{DBLP:conf/icalp/ChanWX21} showed that (see Lemma B.2 of \cite{DBLP:conf/icalp/ChanWX21}) if $B_1 \leq B_2$, then
    \begin{equation*}
        \TMPP(n_1, n_2, n_3 \mid B_1, B_2) = \bigtO{\min_{1 \leq t \leq n_2 n_3} \TMPP \left(n_1, n_2, \frac{n_2n_3}{t} \mid B_1 \right) + t n_1 n_3} \text{.}
    \end{equation*}
    Similarly, if $B_2 \leq B_1$, then
    \begin{equation*}
        \TMPP(n_1, n_2, n_3 \mid B_1, B_2) = \bigtO{\min_{1 \leq t \leq n_1 n_2} \TMPP \left(\frac{n_1 n_2}{t} , n_2, n_3 \mid B_2 \right) + t n_1 n_3} \text{.}
    \end{equation*}
\end{enumerate}

We give a simple algorithm that beats the above two algorithms in certain regimes of parameters.

\begin{proposition}
    \label{prop:concat-min-plus}
    Suppose $\TBMM(a, b, c) = \Omega(ab + ac + bc)$.
    For any integers $n_1, n_2, n_3, B_1, B_2, B \geq 0$,
    \begin{equation*}
        \TMPP(n_1, n_2, n_3 \mid B_1, B_2) = \bigO{\TBMM(n_1 B_1, n_2, n_3 B_2)}
    \end{equation*}
    and if $B = B_1 + B_2$,
    \begin{equation*}
        \TMPP(n_1, n_2, n_3 \mid B_1, B_2) =\bigO{\min \left( \TBMM(n_1 B, n_2 B_2 , n_3), \TBMM(n_1, n_2 B_1, n_3 B)\right)} \text{.}
    \end{equation*}
\end{proposition}

\begin{proof}
    We begin with the first bound.
    Let $A, B$ be the input matrices.
    Construct a $n_1 B_1 \times n_2$ matrix $X$ and $n_2 \times n_3 B_2$ matrix $Y$ as follows
    \begin{equation*}
        X[(i, b), k] = \ind{A[i, k] = b}, \quad Y[k, (j, c)] = \ind{B[k, j] = c} \text{.}
    \end{equation*}
    Note that construction takes $O(n_1 n_2 B_1 + n_2 n_3 B_2)$ time.
    Compute $XY$ and note that $(A * B)[i, j] = \min \set{b + c \given XY[(i, b), (j, c)] = 1}$ can be recovered from $XY$ in $O(n_1 n_3 B_1 B_2)$ time, or linear time in the size of $XY$.

    We now prove the second bound.
    Construct $n_1 B \times n_2 B_2$ matrix $X$ and $n_2 B_2 \times n_3$ matrix $Y$ as follows
    \begin{equation*}
        X[(i, b), (k, c)] = \ind{A[i, k] = b - c}, \quad Y[(k, c), j] = \ind{B[k, j] = c} \text{.}
    \end{equation*}
    Note that $X, Y$ can be constructed in $O(n_1 B n_2 B_2 + n_2 B_2 n_3)$ time.
    Compute $XY$ and note that $(A * B)[i, j] = \min\set{b \given (XY)[(i, b), j] = 1}$ can be recovered from $XY$ in $O(n_1 n_3 B)$ time, or linear time in the size of $XY$.
    The second term in the minimum follows by a symmetric argument.
\end{proof}

\Cref{fig:dag-mssp} compares the running time of our directed \mssp{} algorithm (\Cref{thm:dag-mssp}), the directed \mssp{} algorithm of \cite{grandoni2019faster}, and the best known algorithms for computing \Cref{eq:directed-lower-bound}.

\begin{figure}
    \centering
    \includegraphics[width=0.5\linewidth]{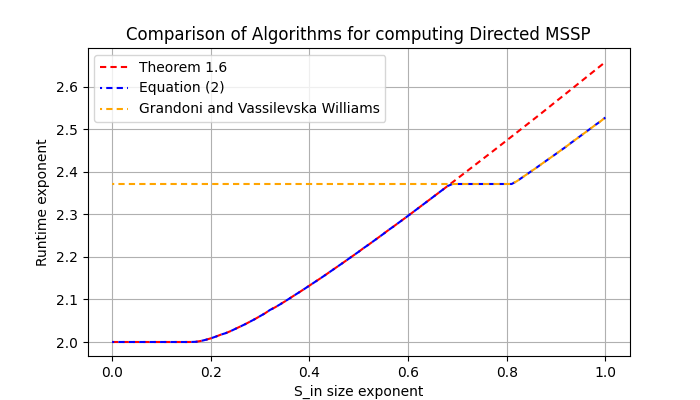}
    \caption{Comparison of known algorithms for computing \mssp{} on unweighted DAGs against known algorithms for computing \Cref{eq:directed-lower-bound} when $|S_{\inT}| = n^{\sigma}$. 
    The blue dotted line denotes the running time of known algorithms for \Cref{eq:directed-lower-bound}. 
    The red dotted line shows the running time of \Cref{thm:dag-mssp} and the yellow dotted line the running time of \cite{grandoni2019faster}.}
    \label{fig:dag-mssp}
\end{figure}

%% file: main.bbl
\newcommand{\etalchar}[1]{$^{#1}$}
\begin{thebibliography}{CGMW19}

\bibitem[AB17]{abboud20174}
Amir Abboud and Greg Bodwin.
\newblock The 4/3 additive spanner exponent is tight.
\newblock {\em Journal of the ACM (JACM)}, 64(4):1--20, 2017.

\bibitem[ABF23]{abboud2023stronger}
Amir Abboud, Karl Bringmann, and Nick Fischer.
\newblock Stronger 3-sum lower bounds for approximate distance oracles via additive combinatorics.
\newblock In {\em Proceedings of the 55th Annual ACM Symposium on Theory of Computing}, pages 391--404, 2023.

\bibitem[ABKZ22]{abboud2022hardness}
Amir Abboud, Karl Bringmann, Seri Khoury, and Or~Zamir.
\newblock Hardness of approximation in p via short cycle removal: cycle detection, distance oracles, and beyond.
\newblock In {\em Proceedings of the 54th Annual ACM SIGACT Symposium on Theory of Computing}, pages 1487--1500, 2022.

\bibitem[ACIM99]{aingworth1999fast}
Donald Aingworth, Chandra Chekuri, Piotr Indyk, and Rajeev Motwani.
\newblock Fast estimation of diameter and shortest paths (without matrix multiplication).
\newblock {\em SIAM Journal on Computing}, 28(4):1167--1181, 1999.

\bibitem[ADV{\etalchar{+}}25]{alman2025fmm}
Josh Alman, Ran Duan, Virginia {Vassilevska Williams}, Yinzhan Xu, Zixuan Xu, and Renfei Zhou.
\newblock More asymmetry yields faster matrix multiplication.
\newblock In {\em Proceedings of the ACM-SIAM Symposium on Discrete Algorithms (SODA)}, pages 2005--2039. SIAM, 2025.

\bibitem[AGM97]{alon1997exponent}
Noga Alon, Zvi Galil, and Oded Margalit.
\newblock On the exponent of the all pairs shortest path problem.
\newblock {\em Journal of Computer and System Sciences}, 54(2):255--262, 1997.

\bibitem[BGW20]{DBLP:conf/focs/BernsteinGW20}
Aaron Bernstein, Maximilian~Probst Gutenberg, and Christian Wulff{-}Nilsen.
\newblock Near-optimal decremental {SSSP} in dense weighted digraphs.
\newblock In {\em Proceedings of the 61st {IEEE} Annual Symposium on Foundations of Computer Science (FOCS)}, pages 1112--1122, 2020.

\bibitem[BH23]{bodwin2023folklore}
Greg Bodwin and Gary Hoppenworth.
\newblock Folklore sampling is optimal for exact hopsets: Confirming the {$\sqrt{n}$} barrier.
\newblock In {\em IEEE 64th Annual Symposium on Foundations of Computer Science (FOCS)}, pages 701--720. IEEE, 2023.

\bibitem[BK06]{baswana2006faster}
Surender Baswana and Telikepalli Kavitha.
\newblock Faster algorithms for approximate distance oracles and all-pairs small stretch paths.
\newblock In {\em 2006 47th Annual IEEE Symposium on Foundations of Computer Science (FOCS'06)}, pages 591--602. IEEE, 2006.

\bibitem[BKMP05]{baswana2005spanner}
Surender Baswana, Telikepalli Kavitha, Kurt Mehlhorn, and Seth Pettie.
\newblock New constructions of (alpha, beta)-spanners and purely additive spanners.
\newblock In {\em SODA}, volume~5, pages 672--681, 2005.

\bibitem[BNW22]{DBLP:conf/focs/BernsteinNW22}
Aaron Bernstein, Danupon Nanongkai, and Christian Wulff{-}Nilsen.
\newblock Negative-weight single-source shortest paths in near-linear time.
\newblock In {\em 63rd {IEEE} Annual Symposium on Foundations of Computer Science, {FOCS}}, 2022.

\bibitem[Bon22]{bonnet20224}
{\'E}douard Bonnet.
\newblock 4 vs 7 sparse undirected unweighted diameter is seth-hard at time n 4/3.
\newblock {\em ACM Transactions on Algorithms (TALG)}, 18(2):1--14, 2022.

\bibitem[BRS{\etalchar{+}}18]{backurs2018towards}
Arturs Backurs, Liam Roditty, Gilad Segal, Virginia~Vassilevska Williams, and Nicole Wein.
\newblock Towards tight approximation bounds for graph diameter and eccentricities.
\newblock In {\em Proceedings of the 50th Annual ACM SIGACT Symposium on Theory of Computing}, pages 267--280, 2018.

\bibitem[BV15]{bodwin2015very}
Gregory Bodwin and Virginia {Vassilevska Williams}.
\newblock Very sparse additive spanners and emulators.
\newblock In {\em Proceedings of the 2015 Conference on Innovations in Theoretical Computer Science}, pages 377--382, 2015.

\bibitem[CDW17]{CA2017oracle}
Vincent {Cohen-Addad}, Søren Dahlgaard, and Christian {Wulff-Nilsen}.
\newblock Fast and compact exact distance oracle for planar graphs.
\newblock In {\em IEEE 58th Annual Symposium on Foundations of Computer Science (FOCS)}, pages 962--973, 2017.

\bibitem[CGMW19]{Char2019oracle}
Panagiotis Charalampopoulos, Pawe\l{} Gawrychowski, Shay Mozes, and Oren Weimann.
\newblock Almost optimal distance oracles for planar graphs.
\newblock In {\em Proceedings of the 51st Annual ACM SIGACT Symposium on Theory of Computing}, page 138–151. Association for Computing Machinery, 2019.

\bibitem[CGR16]{cairo2016new}
Massimo Cairo, Roberto Grossi, and Romeo Rizzi.
\newblock New bounds for approximating extremal distances in undirected graphs.
\newblock In {\em Proceedings of the Twenty-seventh Annual ACM-SIAM Symposium on Discrete Algorithms}, pages 363--376. SIAM, 2016.

\bibitem[Cha05]{Chan05}
Timothy~M. Chan.
\newblock All-pairs shortest paths with real weights in {$O(n^3 / \log n)$} time.
\newblock In {\em 9th International Workshop on Algorithms and Data Structures, {WADS}}, 2005.

\bibitem[Cha10]{DBLP:journals/siamcomp/Chan10}
Timothy~M. Chan.
\newblock More algorithms for all-pairs shortest paths in weighted graphs.
\newblock {\em {SIAM} J. Comput.}, 39(5), 2010.

\bibitem[Che13]{chechik2014spanner}
Shiri Chechik.
\newblock New additive spanners.
\newblock In {\em Proceedings of the twenty-fourth annual ACM-SIAM symposium on Discrete algorithms}, pages 498--512. SIAM, 2013.

\bibitem[Che14]{Chechik2014oracle}
Shiri Chechik.
\newblock Approximate distance oracles with constant query time.
\newblock In {\em Proceedings of the Forty-Sixth Annual ACM Symposium on Theory of Computing}, page 654–663. Association for Computing Machinery, 2014.

\bibitem[CHL25]{chechik2025npsp}
Shiri Chechik, Itay Hoch, and Gur Lifshitz.
\newblock New approximation algorithms and reductions for n-pairs shortest paths and all-nodes shortest cycles.
\newblock In {\em ACM-SIAM Symposium on Discrete Algorithms (SODA)}, pages 5207--5238. SIAM, 2025.

\bibitem[CLR{\etalchar{+}}14]{chechik2014better}
Shiri Chechik, Daniel~H Larkin, Liam Roditty, Grant Schoenebeck, Robert~E Tarjan, and Virginia~Vassilevska Williams.
\newblock Better approximation algorithms for the graph diameter.
\newblock In {\em Proceedings of the twenty-fifth annual ACM-SIAM symposium on Discrete algorithms}, pages 1041--1052. SIAM, 2014.

\bibitem[Cow01]{cowen2001compact}
Lenore~J Cowen.
\newblock Compact routing with minimum stretch.
\newblock {\em Journal of Algorithms}, 38(1):170--183, 2001.

\bibitem[CVX21]{DBLP:conf/icalp/ChanWX21}
Timothy~M. Chan, Virginia {Vassilevska Williams}, and Yinzhan Xu.
\newblock Algorithms, reductions and equivalences for small weight variants of all-pairs shortest paths.
\newblock In {\em Proceedings of the 48th International Colloquium on Automata, Languages, and Programming (ICALP)}, pages 47:1--47:21, 2021.

\bibitem[CW00]{cowen2000compact}
Lenore~J Cowen and Christopher~G Wagner.
\newblock Compact roundtrip routing in directed networks.
\newblock In {\em Proceedings of the 19th Annual ACM Symposium on Principles of Distributed Computing}, pages 51--59, 2000.

\bibitem[CZ22]{chechik2022oracle}
Shiri Chechik and Tianyi Zhang.
\newblock Nearly 2-approximate distance oracles in subquadratic time.
\newblock In {\em ACM-SIAM Symposium on Discrete Algorithms (SODA)}, pages 551--580. SIAM, 2022.

\bibitem[DFK{\etalchar{+}}24]{DBLP:conf/soda/DoryFKNWV24}
Michal Dory, Sebastian Forster, Yael Kirkpatrick, Yasamin Nazari, Virginia {Vassilevska Williams}, and Tijn de~Vos.
\newblock Fast 2-approximate all-pairs shortest paths.
\newblock In {\em Proceedings of the 2024 {ACM-SIAM} Symposium on Discrete Algorithms, {SODA} 2024}, pages 4728--4757, 2024.

\bibitem[DHZ96]{DBLP:conf/focs/DorHZ96}
Dorit Dor, Shay Halperin, and Uri Zwick.
\newblock All pairs almost shortest paths.
\newblock In {\em Proceedings of the 37th Annual Symposium on Foundations of Computer Science, {FOCS} 1996}, pages 452--461, 1996.

\bibitem[DHZ00]{dor2000all}
Dorit Dor, Shay Halperin, and Uri Zwick.
\newblock All-pairs almost shortest paths.
\newblock {\em SIAM Journal on Computing}, 29(5):1740--1759, 2000.

\bibitem[Dij59]{dijkstra1959}
Edsger~W. Dijkstra.
\newblock A note on two problems in connexion with graphs.
\newblock {\em Numer. Math.}, 50:269--271, 1959.

\bibitem[DJVW22]{dalirrooyfard2022npsp}
Mina Dalirrooyfard, Ce~Jin, Virginia {Vassilevska Williams}, and Nicole Wein.
\newblock Approximation algorithms and hardness for n-pairs shortest paths and all-nodes shortest cycles.
\newblock In {\em IEEE 63rd Annual Symposium on Foundations of Computer Science (FOCS)}, pages 290--300. IEEE, 2022.

\bibitem[DKR{\etalchar{+}}22]{DBLP:conf/icalp/DengKRWZ22}
Mingyang Deng, Yael Kirkpatrick, Victor Rong, Virginia {Vassilevska Williams}, and Ziqian Zhong.
\newblock New additive approximations for shortest paths and cycles.
\newblock In {\em Proceedings of the 49th International Colloquium on Automata, Languages, and Programming, {ICALP} 2022}, pages 50:1--50:10, 2022.

\bibitem[DLV25]{dalirrooyfard2025hardness}
Mina Dalirrooyfard, Ray Li, and Virginia {Vassilevska Williams}.
\newblock Hardness of approximate diameter: Now for undirected graphs.
\newblock {\em Journal of the ACM}, 72(1):1--32, 2025.

\bibitem[Dob90]{Dob1990}
Wlodzimierz Dobosiewicz.
\newblock A more efficient algorithm for the min-plus multiplication.
\newblock {\em Int. J. Computer Math.}, 32, 1990.

\bibitem[EGN23]{elkin2023improved}
Michael Elkin, Yuval Gitlitz, and Ofer Neiman.
\newblock Improved weighted additive spanners.
\newblock {\em Distributed Computing}, 36(3):385--394, 2023.

\bibitem[ET24]{elkin2024reach}
Michael Elkin and Chhaya Trehan.
\newblock Faster multi-source directed reachability via shortcuts and matrix multiplication.
\newblock {\em arXiv preprint arXiv:2401.05628}, 2024.

\bibitem[ET25]{elkin2025reach}
Michael Elkin and Chhaya Trehan.
\newblock Faster multi-source reachability and approximate distances via shortcuts, hopsets and matrix multiplication.
\newblock {\em arXiv preprint arXiv:2507.13470}, 2025.

\bibitem[Fin20]{DBLP:journals/siamcomp/Fineman20}
Jeremy~T. Fineman.
\newblock Nearly work-efficient parallel algorithm for digraph reachability.
\newblock {\em {SIAM} J. Comput.}, 49(5), 2020.

\bibitem[FM71]{fischer1971boolean}
Michael~J Fischer and Albert~R Meyer.
\newblock Boolean matrix multiplication and transitive closure.
\newblock In {\em 12th annual symposium on switching and automata theory (swat 1971)}, pages 129--131. IEEE, 1971.

\bibitem[FN18]{DBLP:conf/focs/ForsterN18}
Sebastian Forster and Danupon Nanongkai.
\newblock A faster distributed single-source shortest paths algorithm.
\newblock In {\em Proceedings of the 59th {IEEE} Annual Symposium on Foundations of Computer Science (FOCS)}, pages 686--697, 2018.

\bibitem[FPZW04]{farley2004spanners}
Arthur~M Farley, Andrzej Proskurowski, Daniel Zappala, and Kurt Windisch.
\newblock Spanners and message distribution in networks.
\newblock {\em Discrete Applied Mathematics}, 137(2):159--171, 2004.

\bibitem[Fre76]{fredman1976new}
Michael~L. Fredman.
\newblock New bounds on the complexity of the shortest path problem.
\newblock {\em SIAM J. Comput.}, 5(1):83--89, 1976.

\bibitem[FT87]{fredman1987fibonacci}
Michael~L. Fredman and Robert~Endre Tarjan.
\newblock Fibonacci heaps and their uses in improved network optimization algorithms.
\newblock {\em J. ACM}, 34(3):596--615, 1987.

\bibitem[GV19]{grandoni2019faster}
Fabrizio Grandoni and Virginia {Vassilevska Williams}.
\newblock Faster replacement paths and distance sensitivity oracles.
\newblock {\em ACM Transactions on Algorithms (TALG)}, 16(1):1--25, 2019.

\bibitem[GW20]{DBLP:conf/soda/GutenbergW20a}
Maximilian~Probst Gutenberg and Christian Wulff{-}Nilsen.
\newblock Decremental {SSSP} in weighted digraphs: Faster and against an adaptive adversary.
\newblock In {\em 31st {ACM-SIAM} Symposium on Discrete Algorithms, {SODA}}, 2020.

\bibitem[Han04]{Han04}
Yijie Han.
\newblock Improved algorithm for all pairs shortest paths.
\newblock {\em Inf. Process. Lett.}, 91(5), 2004.

\bibitem[Han06]{Han06}
Yijie Han.
\newblock Achieving {$O(n^3/\log n)$} time for all pairs shortest paths by using a smaller table.
\newblock In {\em 21st International Conference on Computers and Their Applications, CATA}, 2006.

\bibitem[HJVX24]{harbuzova2024improved}
Alina Harbuzova, Ce~Jin, Virginia {Vassilevska Williams}, and Zixuan Xu.
\newblock Improved roundtrip spanners, emulators, and directed girth approximation.
\newblock In {\em Proceedings of the 2024 Annual ACM-SIAM Symposium on Discrete Algorithms (SODA)}, pages 4641--4669. SIAM, 2024.

\bibitem[HKN14]{DBLP:conf/stoc/HenzingerKN14}
Monika Henzinger, Sebastian Krinninger, and Danupon Nanongkai.
\newblock Sublinear-time decremental algorithms for single-source reachability and shortest paths on directed graphs.
\newblock In {\em 46th Annual Symposium on the Theory of Computing, {STOC}}, 2014.

\bibitem[HKN15]{DBLP:conf/icalp/HenzingerKN15}
Monika Henzinger, Sebastian Krinninger, and Danupon Nanongkai.
\newblock Improved algorithms for decremental single-source reachability on directed graphs.
\newblock In {\em Proceedings of the 42nd International Colloquium on Automata, Languages, and Programming (ICALP)}, pages 725--736, 2015.

\bibitem[Hop24]{hoppenworth2024simple}
Gary Hoppenworth.
\newblock Simple linear-size additive emulators.
\newblock In {\em 2024 Symposium on Simplicity in Algorithms (SOSA)}, pages 1--8. SIAM, 2024.

\bibitem[HP21]{huang2021lower}
Shang-En Huang and Seth Pettie.
\newblock Lower bounds on sparse spanners, emulators, and diameter-reducing shortcuts.
\newblock {\em SIAM Journal on Discrete Mathematics}, 35(3):2129--2144, 2021.

\bibitem[HT12]{HanT12}
Yijie Han and Tadao Takaoka.
\newblock An {$O(n^3 \log\log n/\log^2 n)$} time algorithm for all pairs shortest paths.
\newblock In {\em 13th Scandinavian Workshop on Algorithm Theory, {SWAT}}, 2012.

\bibitem[JKSV26]{DBLP:conf/soda/0001KSW26}
Ce~Jin, Yael Kirkpatrick, Michal Stawarz, and Virginia {Vassilevska Williams}.
\newblock Improved additive approximation algorithms for {APSP}.
\newblock In {\em Proceedings of the 2026 Annual {ACM-SIAM} Symposium on Discrete Algorithms (SODA)}, pages 3639--3651, 2026.

\bibitem[JX23]{jin2023removing}
Ce~Jin and Yinzhan Xu.
\newblock Removing additive structure in 3sum-based reductions.
\newblock In {\em Proceedings of the 55th Annual ACM Symposium on Theory of Computing}, pages 405--418, 2023.

\bibitem[KP22a]{kogan2022beating}
Shimon Kogan and Merav Parter.
\newblock Beating matrix multiplication for {$n^{1/3}$}-directed shortcuts.
\newblock In {\em 49th International Colloquium on Automata, Languages, and Programming (ICALP 2022)}, pages 82--1. Schloss Dagstuhl--Leibniz-Zentrum f{\"u}r Informatik, 2022.

\bibitem[KP22b]{koganparter2022}
Shimon Kogan and Merav Parter.
\newblock New diameter-reducing shortcuts and directed hopsets: Breaking the barrier.
\newblock In {\em Proceedings of the ACM-SIAM Symposium on Discrete Algorithms (SODA)}, pages 1326--1341. SIAM, 2022.

\bibitem[KP23a]{DBLP:conf/soda/KoganP23}
Shimon Kogan and Merav Parter.
\newblock Faster and unified algorithms for diameter reducing shortcuts and minimum chain covers.
\newblock In {\em Proceedings of the 2023 {ACM-SIAM} Symposium on Discrete Algorithms (SODA)}, pages 212--239, 2023.

\bibitem[KP23b]{kogan2023new}
Shimon Kogan and Merav Parter.
\newblock New additive emulators.
\newblock In {\em 50th International Colloquium on Automata, Languages, and Programming (ICALP 2023)}, pages 85--1. Schloss Dagstuhl--Leibniz-Zentrum f{\"u}r Informatik, 2023.

\bibitem[KS97]{DBLP:journals/jal/KleinS97}
Philip~N. Klein and Sairam Subramanian.
\newblock A randomized parallel algorithm for single-source shortest paths.
\newblock {\em J. Algorithms}, 25(2):205--220, 1997.

\bibitem[LJS19]{DBLP:conf/focs/LiuJS19}
Yang~P. Liu, Arun Jambulapati, and Aaron Sidford.
\newblock Parallel reachability in almost linear work and square root depth.
\newblock In {\em Proceedings of the 60th {IEEE} Annual Symposium on Foundations of Computer Science (FOCS)}, pages 1664--1686, 2019.

\bibitem[Mun71]{munro1971efficient}
Ian Munro.
\newblock Efficient determination of the transitive closure of a directed graph.
\newblock {\em Information Processing Letters}, 1(2):56--58, 1971.

\bibitem[Pet04]{pettie2004apsp}
Seth Pettie.
\newblock A new approach to all-pairs shortest paths on real-weighted graphs.
\newblock {\em Theoretical Computer Science}, 312(1):47--74, 2004.

\bibitem[PR10]{patrascu2010oracle}
Mihai Patrascu and Liam Roditty.
\newblock Distance oracles beyond the thorup-zwick bound.
\newblock In {\em IEEE 51st Annual Symposium on Foundations of Computer Science}, pages 815--823. IEEE, 2010.

\bibitem[PS89]{peleg1989graph}
David Peleg and Alejandro~A Sch{\"a}ffer.
\newblock Graph spanners.
\newblock {\em Journal of graph theory}, 13(1):99--116, 1989.

\bibitem[PU89]{peleg1989trade}
David Peleg and Eli Upfal.
\newblock A trade-off between space and efficiency for routing tables.
\newblock {\em Journal of the ACM (JACM)}, 36(3):510--530, 1989.

\bibitem[RTZ08]{roditty2008roundtrip}
Liam Roditty, Mikkel Thorup, and Uri Zwick.
\newblock Roundtrip spanners and roundtrip routing in directed graphs.
\newblock {\em ACM Transactions on Algorithms (TALG)}, 4(3):1--17, 2008.

\bibitem[RV13]{roditty2013diameter}
Liam Roditty and Virginia {Vassilevska Williams}.
\newblock Fast approximation algorithms for the diameter and radius of sparse graphs.
\newblock In {\em Proceedings of the 45th Annual ACM Symposium on Theory of Computing}, pages 515--524, 2013.

\bibitem[Sei95]{seidel1995apsp}
Raimund Seidel.
\newblock On the all-pairs-shortest-path problem in unweighted undirected graphs.
\newblock {\em Journal of computer and system sciences}, 51(3):400--403, 1995.

\bibitem[SY24]{DBLP:conf/soda/SahaY24}
Barna Saha and Christopher Ye.
\newblock Faster approximate all pairs shortest paths.
\newblock In {\em Proceedings of the 2024 {ACM-SIAM} Symposium on Discrete Algorithms, {SODA} 2024}, pages 4758--4827, 2024.

\bibitem[SZ99]{shoshan_zwick}
A.~Shoshan and U.~Zwick.
\newblock All pairs shortest paths in undirected graphs with integer weights.
\newblock In {\em 40th Annual Symposium on Foundations of Computer Science}, pages 605--614, 1999.

\bibitem[Tak92]{Takaoka92}
Tadao Takaoka.
\newblock A new upper bound on the complexity of the all pairs shortest path problem.
\newblock {\em Inf. Process. Lett.}, 43(4), 1992.

\bibitem[Tak98]{Takaoka98}
Tadao Takaoka.
\newblock Subcubic cost algorithms for the all pairs shortest path problem.
\newblock {\em Algorithmica}, 20(3), 1998.

\bibitem[Tho04]{Thorup2004planarOracle}
Mikkel Thorup.
\newblock Compact oracles for reachability and approximate distances in planar digraphs.
\newblock {\em J. ACM}, 51(6):993–1024, 2004.

\bibitem[TZ01]{thorup2001compact}
Mikkel Thorup and Uri Zwick.
\newblock Compact routing schemes.
\newblock In {\em Proceedings of the thirteenth annual ACM symposium on Parallel algorithms and architectures}, pages 1--10, 2001.

\bibitem[TZ05]{thorup2005oracle}
Mikkel Thorup and Uri Zwick.
\newblock Approximate distance oracles.
\newblock {\em Journal of the ACM (JACM)}, 52(1):1--24, 2005.

\bibitem[TZ06]{thorup2006spanners}
Mikkel Thorup and Uri Zwick.
\newblock Spanners and emulators with sublinear distance errors.
\newblock In {\em Proceedings of the seventeenth annual ACM-SIAM symposium on Discrete algorithm}, pages 802--809, 2006.

\bibitem[UY91]{DBLP:journals/siamcomp/UllmanY91}
Jeffrey~D. Ullman and Mihalis Yannakakis.
\newblock High-probability parallel transitive-closure algorithms.
\newblock {\em {SIAM} J. Comput.}, 20(1):100--125, 1991.

\bibitem[VXX24]{vxx2024shortcut}
Virginia {Vassilevska Williams}, Yinzhan Xu, and Zixuan Xu.
\newblock Simpler and higher lower bounds for shortcut sets.
\newblock In {\em Proceedings of the ACM-SIAM Symposium on Discrete Algorithms (SODA)}, pages 2643--2656. SIAM, 2024.

\bibitem[Wil18]{Williams18}
R.~Ryan Williams.
\newblock Faster all-pairs shortest paths via circuit complexity.
\newblock {\em {SIAM} J. Comput.}, 47(5), 2018.

\bibitem[Woo10]{woodruff2010additive}
David~P Woodruff.
\newblock Additive spanners in nearly quadratic time.
\newblock In {\em International Colloquium on Automata, Languages, and Programming}, pages 463--474. Springer, 2010.

\bibitem[Zwi02]{zwick2002apsp}
Uri Zwick.
\newblock All pairs shortest paths using bridging sets and rectangular matrix multiplication.
\newblock {\em Journal of the ACM (JACM)}, 49(3):289--317, 2002.

\bibitem[Zwi04]{ZwickAPSP04}
Uri Zwick.
\newblock A slightly improved sub-cubic algorithm for the all pairs shortest paths problem with real edge lengths.
\newblock In {\em 15th International Symposium on Algorithms and Computation, {ISAAC}}, 2004.

\end{thebibliography}
